\documentclass[a4paper,UKenglish,cleveref,pdftex,thm-restate,numberwithinsect]{lipics-v2021}
\newif\ifreport
\reporttrue 
   
\newif\iffull
\fullfalse

\iffull
\newcommand{\full}[1]{{color{blue}#1}}
\else 
\newcommand{\full}[1]{}
\fi

\nolinenumbers 

\ifreport
\nolinenumbers 
\hideLIPIcs  
\else
\relatedversion{An extended version is available at \url{https://arxiv.org/abs/2007.04213}.} 
\fi

\usepackage{hyperref}
\usepackage[disable,%
textsize=tiny]{todonotes}
\usepackage[all]{xy}
\SelectTips{cm}{}
\usepackage{proof}

\usepackage{wrapfig}

\newcommand{\mor}{\mathsf{Mor}}
\newcommand{\mon}{\mathsf{Mono}}
\newcommand{\reg}{\mathsf{Reg}}

\newlength{\myheight}

\usepackage{amssymb,graphicx,epsfig,color}
\usepackage{subcaption}
\usepackage{wrapfig}

\usepackage{pgf}
\usepackage{tikz}
\usepackage{tikz-cd}
\usetikzlibrary{arrows,shapes,snakes,automata,backgrounds,petri,fit,positioning,calc}
\tikzstyle{node}=[circle, draw=black, minimum size=1mm, inner sep=1.5pt, font=\tiny]

\tikzstyle{trans}=[font=\scriptsize]
\tikzstyle{lab}=[font=\small]
\tikzset{
  coloredge/.style={
        ->,
        color=red
    }
}

\tikzset{
  colorloop/.style={
        loop,
        color=red
    }
}
  
\tikzcdset{
  perm/.style={
    shorten >=-1mm, shorten <=-1mm,
    -,
  }
}

\newcommand{\pgfBox}{
  \begin{pgfonlayer}{background} 
    \fill[blue!2,thick,draw=black!50,rounded corners,inner sep=3mm] ([xshift=-1.5pt,yshift=-1.5pt]current bounding box.south west) rectangle ([xshift=1.5pt,yshift=1.5pt]current bounding box.north east);
  \end{pgfonlayer}
}
\usepackage{scalerel}

\usepackage{stmaryrd}

\newcommand{\cat}[1]{\ensuremath{\mathbf{#1}}}

\usepackage{amsmath}
\usepackage{amssymb}
\usepackage{amsthm}
\usepackage{enumerate}
\usepackage{xspace}
\usepackage{amsfonts}
\usepackage{mathrsfs}
\usepackage{cite}
\usepackage{float}
\usepackage{fancybox}
\usepackage{proof-at-the-end}
\usepackage{cleveref}

\newcommand{\zunf}[0]{\ensuremath{\zunf}}

\newcommand{\interval}[2][1]{\ensuremath{[{#1},{#2}]}}

\newcommand{\perm}{\sigma}

\newcommand{\gph}[1]{\textbf{\textup{{#1}-Graph}}}

\newcommand{\id}[1]{\mathsf{id}_{#1}}

\def\A{\textbf {\textup{A}}}
\newcommand{\mto}[0]{\scalebox{1}{$\rightarrowtail$}}

\def\R{\mathsf{R}}

\def\C{\textbf {\textup{C}}}
\def\D{\textbf {\textup{D}}}
\def\X{\textbf {\textup{X}}}
\def\Y{\textbf {\textup{Y}}}
\def\G{\textbf {\textup{G}}}

\newcommand{\dder}[1]{\mathscr{#1}}

\newcommand{\der}[1]{\underline{\dder{#1}}}

\def\Set{\textbf {\textup{Set}}}

\newcommand{\lpro}{\langle \hspace{-1.85pt}[}
\newcommand{\rpro}{]\hspace{-1.85pt}\rangle}
\newcommand{\tpro}[1]{\lpro \der{#1}\rpro}

\newcommand{\lgh}[0]{\mathsf{lg}}

\usepackage{xparse}

\newcommand{\inv}[1]{\mathsf{inv}({#1})}

\newcommand{\shift}[1]{\ensuremath{\mathrel{{\leftrightsquigarrow}_{#1}}}}

\newcommand{\shifteq}[1][]{\ensuremath{\mathrel{{\equiv}^\mathit{sh}_{#1}}}}

\NewDocumentCommand{\transp}{m o}{%
  \ensuremath{({#1},%
  \IfNoValueTF{#2}%
    {{#1}+1}%
    {#2}%
    )}
}

\NewDocumentCommand{\mycommand}{o}{%
  \IfNoValueTF{#1}
    {code when no optional argument is passed}
    {code when the optional argument #1 is present}%
}

\newtheorem*{notation}{Notation}

\newcommand{\rem}[2]{{\color{blue}#1}{\color{red}#2}}
\renewcommand{\rem}[2]{}

\author{Paolo Baldan} 
{Department of Mathematics, University of Padua, Italy}
{baldan@math.unipd.it}
{https://orcid.org/0000-0001-9357-5599}{}

\author{Davide Castelnovo}
{Department of Mathematics, University of Padua, Italy}
{davide.castelnovo@math.unipd.it}
{https://orcid.org/0000-0002-5926-5615}{}

\author{Andrea Corradini}
{Department of Computer Science, University of Pisa, Italy}
{andrea.corradini@unipi.it}
{https://orcid.org/0000-0001-6123-4175}{}

\author{Fabio Gadducci}
{Department of Computer Science, University of Pisa, Italy}
{fabio.gadducci@unipi.it}
{https://orcid.org/0000-0003-0690-3051}{}

\authorrunning{P.~Baldan, D.~Castelnovo, A.~Corradini, F.~Gadducci}
\Copyright{Paolo Baldan, Davide Castelnovo, Andrea Corradini, and Fabio Gadducci}

\funding{The research has been partially supported by the EuropeanUnion - NextGenerationEU under the National Recovery and Resilience Plan (NRRP) - Call PRIN 2022 PNRR - Project P2022HXNSC ``Resource Awareness in Programming: Algebra, Rewriting, and Analysis'', by the Italian MUR - Call PRIN 2022 - Project 20228KXFN2 
``Spatio-Temporal Enhancement of Neural nets for Deeply
Hierarchical Automatised Logic'' and by the University of Pisa - Call PRA 2022 - Project 2022\_99 ``Formal Methods for the Healthcare Domain based on Spatial Information''.}

\EventEditors{Rupak Majumdar and Alexandra Silva}
\EventNoEds{2}
\EventLongTitle{35th International Conference on Concurrency Theory (CONCUR 2024)}
\EventShortTitle{CONCUR 2024}
\EventAcronym{CONCUR}
\EventYear{2024}
\EventDate{Septembre 9--13, 2024}
\EventLocation{Calgary, Canada}
\EventLogo{}
\SeriesVolume{311}
\ArticleNo{6}

\title{Left-Linear Rewriting in Adhesive Categories}

\begin{document}
\maketitle

\begin{abstract}
  When can two sequential steps performed by a computing device be
  considered (causally) independent?  This is a relevant question for
  concurrent and distributed systems, since independence means that
  they could be executed in any order, and potentially in
  parallel. Equivalences identifying rewriting sequences which differ
  only for independent steps are at the core of the theory of
  concurrency of many formalisms. We investigate the issue in the context of the double pushout
  approach to rewriting in the general setting of adhesive categories.
  While a consolidated theory exists for linear rules, which
  can consume, preserve and generate entities, this paper focuses on
  left-linear rules which may also ``merge'' parts of the state.  This is an apparently minimal, yet technically hard enhancement, since  a standard characterisation of independence that -- in the
  linear case -- allows one to derive a number of properties,
  essential in the development of a theory of concurrency, no longer
  holds.  The paper performs an in-depth study of the notion of independence
  for left-linear rules: it introduces a novel characterisation of
  independence, identifies well-behaved classes of left-linear
  rewriting systems, and provides some fundamental results including a
  Church-Rosser property and the existence of canonical equivalence
  proofs for concurrent computations. These results properly extends
  the class of formalisms that can be modelled in the adhesive
  framework.
\end{abstract}

\section{Introduction}

One of the key pay-off of concurrency theory is the idea that the behaviour of 
a computational device can be modelled by abstracting its sequences of steps
via a suitable equivalence. Intuitively, the equivalence 
captures when steps are causally unrelated and thus could be executed in any 
order, and possibly in parallel. Two seminal contributions to this line of research have been
Mazurkiewicz's traces~\cite{Mazurkiewicz86} and Winskel's event structures~\cite{NPW:PNES}.
Concerning formalisms based on the rewriting paradigm, concurrency often boils down to having a notion of independence between two consecutive steps. 
Noteworthy examples are the
interchange law~\cite{Mes92} in term rewriting and permutation 
equivalence~\cite{JJL80} in $\lambda$-calculi. 

Among rewriting-based formalisms, those manipulating graph-like
structures proved useful in many settings for representing the
dynamics of distributed systems: the graph is used to represent the
entities and their relations within states, while rewriting steps,
modifying the graph, model computation steps that result in changes to
the system state. The DPO (double-pushout) approach~\cite{EhrigPS73}
is nowadays considered a standard one for these structures, thanks to
its locality (rule application has a local effect on a state,
differently from more recent approaches like
SqPO~\cite{CorradiniHHK06} or PBPO$^+$~\cite{OverbeekER23}) and to its
flexibility: the rules allow to specify that, in a rewriting step,
some entities in the current state are removed, some others are needed
for the rewriting step to occur but remain unaltered, and some new
entities can be generated. Concerning concurrency, the chosen notion
is switch (also referred to as shift) equivalence, and independence is formally captured by what is
called the Church-Rosser theorem of DPO
rewriting~\cite[Chapter 3, Section~3.4]{CorradiniMREHL97}. 

A noteworthy advantage of DPO rewriting is that it can be formulated over
adhesive categories and their variants~\cite{lack2005adhesive,ehrig2006weak}, which include, among others, toposes~\cite{johnstone2007quasitoposes} and string diagrams~\cite{bonchi2022string}. Operating within the framework of adhesive categories allows one to devise the foundations of a rewriting theory that can be then instantiated to the context of interest, thus avoiding the need of
repeatedly proving similar ad-hoc results for each specific setting. 

So far, most of the theoretical results focused on \emph{linear rules}: intuitively, a step is required only to consume, preserve and generate entities.  However, in some situations it is also necessary to ``merge'' parts of the state.  This happens in the context of string diagrams mentioned above, but also in graphical implementations of nominal calculi where, as a result of name passing, the received name is identified with a local one at the receiver~\cite{CVY:ESSPE,Gad07}, or in the visual modelling of bonding in biological/chemical processes~\cite{PUY:MBPE}. A similar mechanism is at the core of e-graphs (equality graphs), used for rewrite-driven compiler optimisations~\cite{WNW:egg}: rather thanmerging nodes corresponding to equal expressions, the idea -- conceptually and, as we will see, also mathematically similar -- is to maintain an
equivalence over the nodes.

Technically, to acquire the expressiveness needed for capturing fusions, i.e.~the merge of some elements in the state, requires to move from linear to left-linear rewriting rules. While left-linear rules have been considered in various instances (mainly in categories of graphs) and some results concerning the corresponding rewriting theory have been put forward~\cite{Ehrig1976,EHP:BRfTToHLRS,baldan2011adhesivity}, the phenomena arising when devising a theory of independence and concurrency for rewriting with left-linear rules in the general
setting of adhesive and quasi-adhesive categories have not been systematically explored, even if an investigation in the setting of quasi-topoi, considering non-linear rules yet restricting their applicability, has been presented in~\cite{behr2021concurrency,BehrHK23}.

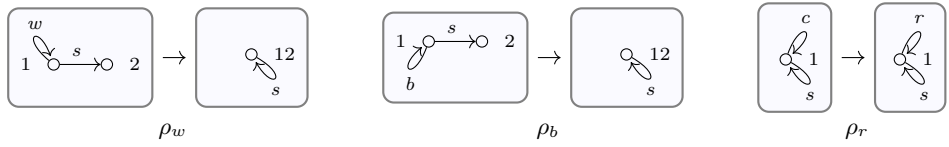
\begin{figure}
  \begin{center}

  %
  %
  \begin{tikzpicture}[node distance=3mm, font=\small]
      \node (l) {
      \begin{tikzpicture}[node distance=2mm, label distance=1mm, font=\scriptsize]
      \node at (0,1.0) [node, label=left:$1$] (1) {}
        edge [in=110, out=140, loop] node[above] {$\mathit{w}$} ();

        \node at (0,0.6) {};

        \node at (0.7,1.0) [node, label=right:$2$] (2) {};

        \draw[->] (1) to node[above, pos=0.4]{$\mathit{s}$} (2);        
      \pgfBox
      \end{tikzpicture} 
    };
    \node [right=of l] (r) {
      \begin{tikzpicture}[node distance=2mm, label distance=1mm, font=\scriptsize]
        \node at (0,1.45) {};
        \node at (0,0.5) {};
        \node at (0.6,1.0) [node, label=right:$12$] (12) {}
          edge [in=300, out=330, loop] node[below] {$\mathit{s}$} (); 
      \pgfBox
      \end{tikzpicture}
    };
    \path (l) edge[->] (r);
    \node at ($(l)!0.54!(r)$) [below=7.5mm] {$\rho_w$}; 
    \end{tikzpicture}
    \hspace{5mm}
    %
    %
    \begin{tikzpicture}[node distance=3mm, font=\small]
      \node (l) {
      \begin{tikzpicture}[node distance=2mm, label distance=1mm, font=\scriptsize]
       \node at (0,1.0) [node, label=left:$1$] (1) {}
        edge [in=220, out=250, loop] node[below] {$\mathit{b}$} ();

        \node at (0,0.6) {};

        \node at (0.7,1.0) [node, label=right:$2$] (2) {};

        \draw[->] (1) to node[above, pos=0.4]{$\mathit{s}$} (2);

      \pgfBox
      \end{tikzpicture} 
    };
    \node [right=of l] (r) {
      \begin{tikzpicture}[node distance=2mm, label distance=1mm, font=\scriptsize]
      \node at (0,1.45) {};  
      \node at (0,0.5) {};
      \node at (0.6,1.0) [node, label=right:$12$] (12) {}
          edge [in=300, out=330, loop] node[below] {$\mathit{s}$} ();        
      \pgfBox
      \end{tikzpicture}
    };
    \path (l) edge[->]  (r);
    \node at ($(l)!0.54!(r)$) [below=7.5mm] {$\rho_b$};         
    \end{tikzpicture}
    \hspace{5mm}
    %
    \begin{tikzpicture}[node distance=3mm, font=\small]
      \node (l) {
      \begin{tikzpicture}[node distance=2mm, label distance=1mm, font=\scriptsize]
      \node at (0.4,0.5) {};
      \node at (0.6,1.0) [node, label=right:$1$] (1) {}
        edge [in=40, out=70, loop] node[above] {$\mathit{c}$} ()
        edge [in=300, out=330, loop] node[below] {$\mathit{s}$} ();
      \pgfBox
      \end{tikzpicture} 
    };
    \node [right=of l] (r) {
      \begin{tikzpicture}[node distance=2mm, label distance=1mm, font=\scriptsize]
      \node at (0.4,0.5) {};
      \node at (0.6,1.0) [node, label=right:$1$] (1) {}
        edge [in=40, out=70, loop] node[above] {$\mathit{r}$} ()
        edge [in=300, out=330, loop] node[below] {$\mathit{s}$} ();
      \pgfBox
      \end{tikzpicture}
    };
    \path (l) edge[->] (r);    
    \node at ($(l)!0.54!(r)$) [below=7.5mm] {$\rho_r$}; 
    \end{tikzpicture}
\end{center}

  
  \caption{Rewriting rules for the coffee example.}
  \label{fi:coffee-rules}
\end{figure}

\begin{figure}
  \begin{center}

  %
  %
  %
  \begin{tikzpicture}[node distance=8mm, font=\small]
    \node (G0) {
      \begin{tikzpicture}[node distance=2mm, label distance=1mm, font=\scriptsize]
        \node at (0,1.0) [node, label=left:$1$] (1) {}
        edge [in=110, out=140, loop] node[above] {$\mathit{w}$} ()
        edge [in=220, out=250, loop] node[below] {$\mathit{b}$} ();
        
        \node at (0.7,1.0) [node, label=right:$2$] (2) {}
        edge [in=40, out=70, loop] node[above] {$\mathit{c}$} ();

        \draw[->] (1) to node[above, pos=0.4]{$s$} (2);        
      \pgfBox
      \end{tikzpicture} 
    };
    \node[below] at (G0.south) {$G_0$};
    \node [right=of G0] (G1) {
      \begin{tikzpicture}[node distance=2mm, label distance=1mm,, font=\scriptsize]
      \node at (0.6,1.0) [node, label=right:$12$] (12) {}
        edge [in=220, out=250, loop] node[below] {$\mathit{b}$} ()
        edge [in=40, out=70, loop] node[above] {$\mathit{c}$} ()
        edge [in=300, out=330, loop] node[below] {$\mathit{s}$} ();        
      \pgfBox
      \end{tikzpicture}
    };
    \node[below] at (G1.south) {$G_1$};

    \node [right=of G1] (G2) {        
      \begin{tikzpicture}[node distance=2mm, label distance=1mm, font=\scriptsize]
        \node at (0.4,0.5) {};
      \node at (0.6,1.0) [node, label=right:$12$] (12) {}
        edge [in=40, out=70, loop] node[above] {$\mathit{c}$} ()
        edge [in=300, out=330, loop] node[below] {$\mathit{s}$} ();        
      \pgfBox
    \end{tikzpicture}
    };
    \node[below] at (G2.south) {$G_2$};

    \node [right=of G2] (G3) {      
      \begin{tikzpicture}[node distance=2mm, label distance=1mm, font=\scriptsize]
        \node at (0.4,0.5) {};
      \node at (0.6,1.0) [node, label=right:$12$] (12) {}
        edge [in=40, out=70, loop] node[above] {$\mathit{r}$} ()
        edge [in=300, out=330, loop] node[below] {$\mathit{s}$} ();        
        \pgfBox
      \end{tikzpicture}
    };
    \node[below] at (G3.south) {$G_3$};
   
    \path (G0) edge[->] node[trans, above] {$\rho_w$} (G1);
    \path (G1) edge[->] node[trans, above] {$\rho_b$} (G2);
    \path (G2) edge[->] node[trans, above] {$\rho_r$} (G3);    
  \end{tikzpicture}
\end{center}

  \caption{A sequence of rewriting steps.}
  \label{fi:coffee-rewriting}
\end{figure}
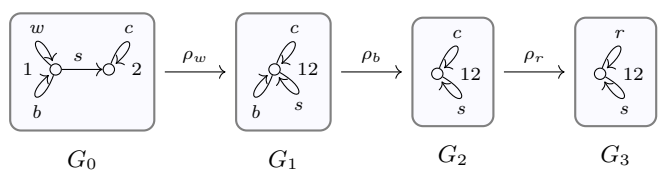

Just to get some basic intuition about left-linear rewriting systems, their capabilities and the problems that can arise, assume we want to model a situation in which a coffee comes with a bag of white sugar and a bag of brown sugar. The situation is represented by the graph $G_0$ in Fig.~\ref{fi:coffee-rewriting}: the loop $\mathit{c}$ is the coffee, the two loops $b$ and $w$  represent the white and brown sugar. The corresponding nodes are connected to the coffee node by an $s$-labelled edge to indicate that they can be possibly added to the coffee. The available rules are in Fig.~\ref{fi:coffee-rules}. Rules $\rho_w$ and $\rho_b$ model the addition of white and brown sugar, respectively, to the coffee. Note that adding one kind of sugar, say $x$, to the coffee is realised by deleting loop $x$ and merging the corresponding node with the coffee node so that a self-loop $s$ is created.

Once some sugar has been added, the coffee is ready for drinking. This is represented by rule $\rho_r$, which checks for the presence of sugar, represented by the loop $s$, in the coffee, and brings it to the $r$ state, ready for drinking.

A rewriting sequence is depicted in Fig.~\ref{fi:coffee-rewriting}, applying $\rho_w$, $\rho_b$, and $\rho_r$ in sequence, thus producing a coffee with both white and brown sugar, ready for drinking. A first interesting observation is that after applying rule $\rho_w$ to $G_0$, thus producing graph $G_1$, we can still apply rule $\rho_b$, with a non-injective match. Intuitively, once the coffee is ready for drinking because some sugar has been added, adding more sugar will not change its state. Indeed, the rules could also be applied in reverse order. Should they be considered independent? Moreover, since rule $\rho_r$ just requires the presence of sugar, it could be applied immediately after $\rho_w$. This might lead to think that $\rho_b$ and $\rho_r$ are independent. However, this is not completely clear, since in absence of $\rho_w$, the application of $\rho_b$ would be essential to enable $\rho_r$.

In this paper we perform a systematic exploration of these phenomena at the level of $\mathcal{M}$-adhesive categories, which encompass a large family of DPO-based formalisms.

As an initial step, we observe that for left-linear rules some fundamental properties that play a role
in the concurrency theory of linear rewriting need to be deeply revised, or plainly do not hold. Firstly, the notion of sequential independence and the Church-Rosser theorem, ensuring that sequential independent steps can be performed in reverse order, does not hold out of the box. Even worse, given two independent steps it might be possible to switch them in several
different ways. As we will observe, this prevents the development of a sensible theory
of concurrency, since switches are used as a basis for equating
sequences of rewriting steps that differ only in the order of
independent steps.

We single out a class of left-linear rewriting systems, referred to as \emph{well-switching rewriting systems}, where sequential independence ensures existence and uniqueness of the switch, and argue that this is the right setting for dealing with left-linear rewriting systems.

On the one hand, they capture most categories of interest for rewriting. Indeed, we identify general classes of left-linear rewriting systems on graph-like structures that are well-switching. As a sanity check, we prove that linear rewriting systems are well-switching.

On the other hand, we argue that a concurrency theory of rewriting can be developed for well-switching rewriting systems, recovering, sometimes in weakened form, most of the results that hold in the linear case. The development faces a main technical obstacle. In the literature on linear rewriting systems, a central result is the characterisation of switch equivalence of sequences of steps based on
what are called processes or consistent permutations~\cite{CMR:GP,Handbook,BaldanCHKS06,heindel2009category}. In turn this is the key to derive some fundamental properties for linear rewriting, namely
\begin{enumerate}
\item \emph{globality of independence}: if two independent steps are
  moved forward or backward due to the application of switchings to a
  sequence, they remain independent;
\item \emph{consistency of switching}: when transforming
  sequences of rewriting steps into equivalent ones by switching independent steps, the
  result does not depend on the order in which switchings are applied, but
  only on the associated permutation.
\end{enumerate}
Properties (1) and (2) allow one to deduce a canonical form for equivalence proofs 
between sequences of rewriting steps seen as concurrent computations.

The fact that for left-linear rewriting systems the characterisation of switch equivalence via consistent permutations fails, and thus cannot be used for proving (1) and (2), leads to asking if these
properties hold.  We show that much of the theory can be retained: globality of independence holds in a weaker form, conceptually linked to the fact that fusions introduce a form of disjunctive causality, while consistency of switching holds unchanged. Finally, we prove that a canonical form for switching
sequences can also be recovered. These results represent a solid basis for developing a satisfactory theory of concurrency for left-linear rules.

\emph{Synopsis.} 
In \S\ref{sec:ade} we recall the basic definitions of DPO rewriting for $\mathcal{M}$-adhesive categories. 
In \S\ref{sec:equi} we present the key novelties of our proposal, introducing the standard notion of independence and a novel axiomatic notion of switchability, showing that the former fails to imply the latter for left-linear rewriting systems. We then define a subclass of left-linear rewriting systems
that we call well-switching rewriting systems, and we prove that the classical results for linear rewriting systems are recovered for such a class. Finally, in \S\ref{conclusions} we outline directions for future work. The appendices recall basic results on adhesive categories (Appendix~\ref{app:ade}), present \emph{strong enforcing} rewriting systems, a class of systems where the local Church-Rosser Theorem holds (Appendix~\ref{app:fill}), and give proofs of the results (Appendix~\ref{omitted}).

\section{$\mathcal{M}$-adhesive categories and rewriting systems}\label{sec:ade}

This  section recalls the basic theory of \emph{$\mathcal{M}$-adhesive categories} \cite{azzi2019essence,ehrig2012,ehrig2014adhesive,lack2005adhesive,heindel2009category}. Given a category $\X$ we will not distinguish notationally between $\X$ and its class of objects, so
``$X\in \X$'' means that $X$ is an object of $\X$. We let $\mor(\X)$, $\mon(\X)$ and $\reg(\X)$ denote the class of all arrows, monos and regular monos of $\X$, respectively.  Given an integer $n\in \mathbb{Z}$, $[0,n]$ denotes the set of natural numbers less than or equal to $n$; in particular, $[0,n]=\emptyset$ if $n<0$.

\subsection{$\mathcal{M}$-adhesivity}\label{subsec:ade}
The key property of $\mathcal{M}$-adhesive categories is the \emph{Van Kampen condition}~\cite{brown1997van,johnstone2007quasitoposes,lack2005adhesive}. In order to define it we need to introduce some terminology.  Let  $\X$ be a category. A subclass $\mathcal{A}$ of
  $\mor(\X)$ is called
  \begin{itemize}
    \parbox{11cm}{\item
      \emph{stable under pushouts (pullbacks)} if for every pushout (pullback) square as the one on the right, 	if $m \in \mathcal{A}$ ($n\in \mathcal{A}$) then $n \in \mathcal{A}$ ($m \in \mathcal{A}$);
    \item \emph{closed under composition} if $h, k\in \mathcal{A}$ implies $h\circ k\in \mathcal{A}$ whenever $h$ and $k$ are composable.{\tiny }}\hfill
    \parbox{2cm}{$\xymatrix{A\ar[r]^f  \ar[d]_{m}& B \ar[d]^n \\ C \ar[r]_g & D}$}
    \parbox{11cm}{}\hfill
  \end{itemize}
  
  \begin{definition}[Van Kampen property]
    Let $\X$ be a category and consider the diagram\\
    \parbox{10cm}{
      aside.
    Given  a class of arrows $\mathcal{A}\subseteq \mor(\X)$, we say that the bottom square
    is an \emph{$\mathcal{A}$-Van Kampen square} if
      \begin{enumerate}
      \item it is a pushout square;
    \item 	whenever the cube above has pullbacks as back and left faces and the vertical arrows belong to $\mathcal{A}$, then its top face is a pushout 
      if and only if the front and right faces are pullbacks.
      \end{enumerate}}
    \parbox{2cm}{$\xymatrix@C=10pt@R=6pt{&A'\ar[dd]|\hole_(.65){a}\ar[rr]^{f'} \ar[dl]_{m'} && B' \ar[dd]^{b} \ar[dl]_{n'} \\ C'  \ar[dd]_{c}\ar[rr]^(.7){g'} & & D' \ar[dd]_(.3){d}\\&A\ar[rr]|\hole^(.65){f} \ar[dl]^{m} && B \ar[dl]^{n} \\C \ar[rr]_{g} & & D }$ }\\
    Pushout squares that enjoy only the ``if'' half of item (2) above are called \emph{$\mathcal{A}$-stable}.
    
    A $\mor(\X)$-Van Kampen square is called  \emph{Van
      Kampen} and a $\mor(\X)$-stable square  \emph{stable}.
  \end{definition}

We can now define $\mathcal{M}$-adhesive categories.

\begin{definition}[$\mathcal{M}$-adhesive category]
  Let $\X$ be a category and $\mathcal{M}$ a subclass of
  $\mon(\X)$  including  all isomorphisms, closed under composition,  and stable under pullbacks and pushouts.  The category  $\X$ is said to be \emph{$\mathcal{M}$-adhesive} if
  \begin{enumerate}
  \item it has \emph{$\mathcal{M}$-pullbacks}, i.e.~pullbacks along arrows of $\mathcal{M}$;
  \item it has \emph{$\mathcal{M}$-pushouts}, i.e.~pushouts along arrows of $\mathcal{M}$;
  \item  $\mathcal{M}$-pushouts are $\mathcal{M}$-Van Kampen squares.
  \end{enumerate}
  
  A category $\X$ is said to be \emph{strictly $\mathcal{M}$-adhesive}
  if $\mathcal{M}$-pushouts are Van Kampen squares.
\end{definition}

We write $m\colon X\mto Y$ to denote that an arrow $m\colon X\to Y$ belongs to $\mathcal{M}$.

\begin{remark}
  \label{rem:salva}
  \emph{Adhesivity} and \emph{quasiadhesivity} 
  \cite{lack2005adhesive,garner2012axioms} coincide with strict
  $\mon(\X) $-adhesivity and strict $\reg(\X)$-adhesivity,
  respectively.
\end{remark}

\begin{example}
  \label{rem:iso}
  Every category $\X$ is $\mathsf{I(\X)}$-adhesive, where $\mathsf{I(\X)}$ is the class of 
  isomorphisms.
\end{example}

$\mathcal{M}$-adhesivity is well-behaved with respect to  the construction of slice and functor categories \cite{mac2013categories}, as shown by the following theorems~\cite{ehrig2006fundamentals,lack2005adhesive}.

\begin{theorem}
  \label{thm:slice-functors}
  Let $\X$ be an $\mathcal{M}$-adhesive category. Given an object $X$
  the category $\X/X$ is $\mathcal{M}/X$-adhesive with
  $\mathcal{M}/X:=\{m\in \mor(\X/X) \mid m\in
  \mathcal{M}\}$. Similarly, $X/\X$ is $X/\mathcal{M}$-adhesive with
  $X/\mathcal{M}:=\{m\in \mor(X/\X) \mid m\in \mathcal{M}\}$.
  
  Moreover for every small category $\Y$, the category $\X^\Y$ of
  functors $\Y\to \X$ is $\mathcal{M}^{\Y}$-adhesive, where
  $\mathcal{M}^{\Y}:=\{\eta \in \mor(\X^\Y) \mid \eta_Y \in
    \mathcal{M} \text{ for every } Y\in \Y\}$.
\end{theorem}

We can list various examples of $\mathcal{M}$-adhesive categories (see
\cite{castelnovo2023thesis,CastelnovoGM22,lack2006toposes}).

\begin{example}
  \label{ex:adhesive}
  $\cat{Set}$ is adhesive, and, more generally, every topos is
  adhesive~\cite{lack2006toposes}. By the closure properties above, every presheaf $[\cat{X},\cat{Set}]$ is adhesive, thus the category
  $\cat{Graph} = [ E \rightrightarrows V, \cat{Set}]$ is adhesive
  where $E \rightrightarrows {V}$ is the two objects category with two
  morphisms $s,t \colon{E} \to {V}$. Similarly, various
  categories of hypergraphs can be shown to be adhesive, such as term
  graphs and hierarchical graphs~\cite{CastelnovoGM24}. Note that the category $\cat{sGraphs}$ of simple graphs, 
  i.e.~graphs without parallel edges, is
  $\reg{(\cat{sGraphs})}$-adhesive~\cite{BehrHK23} but not
  quasiadhesive.
\end{example}

A number of properties of adhesive categories that play a role in the
theory of rewriting generalise to $\mathcal{M}$-adhesive
categories. These include $\mathcal{M}$-pushout-pullback decomposition
and uniqueness of pushouts complements (details and proofs are in
\Cref{app:ade}).

\subsection{DPO rewriting systems derivations}\label{subsec:DPO}

$\mathcal{M}$-adhesive categories are the right context in which to
perform abstract rewriting using what is called the Double-Pushout
(DPO) approach.

\begin{definition}[\!\!\cite{habel2012mathcal,heindel2009category}]
  Let $\X$ be an $\mathcal{M}$-adhesive category. A \emph{left
    $\mathcal{M}$-linear} rule $\rho$ is a pair  of
  arrows $\rho = (l: K \to L, r: K \to R)$ such that $l$ belongs to
  $\mathcal{M}$.  The rule $\rho$ is \emph{$\mathcal{M}$-linear} if
  $r\in \mathcal{M}$ too.
  \full{A rule $\rho$ is said to be
    \emph{consuming} if $l$ is not an isomorphism.}
  The object $L$ is
  the \emph{left-hand side}, $R$ the \emph{right-hand side}, and
  $K$ the \emph{interface}. 

  A \emph{left-linear rewriting system} is a pair $(\X, \R)$ where
  $\X$ is an $\mathcal{M}$-adhesive category and $\R$ a set of left
  $\mathcal{M}$-linear rules. It  is  called $\mathcal{M}$-\emph{linear} if every rule in $\R$ is  $\mathcal{M}$-linear.

  \vspace{.1cm}
  \noindent
  \parbox{10cm}{\hspace{15pt}Given two objects $G$ and $H$ and a rule
    $\rho=(l,r)$ in $\R$, a \emph{direct derivation $\mathscr{D}$ from
      $G$ to $H$ applying rule $\rho$} is a diagram as the one
    aside, in which both squares are pushouts. The arrow $m$ is called
    the \emph{match} of the derivation, while $h$ is its
    \emph{co-match}.  We denote a direct derivation $\dder{D}$
    from $G$ to $H$ as $\dder{D}\colon G\Mapsto H$. }
  \parbox{3cm}{$\xymatrix{L \ar[d]_{m}& K \ar[d]^{k}\ar@{>->}[l]_{l}
      \ar[r]^{r} & R \ar[d]^{h}\\G & \ar@{>->}[l]^{f} D \ar[r]_{g}&
      H}$}
\end{definition}

A derivation can now be defined simply as a sequence of direct derivations.

\begin{definition}[Derivation]
  Let $\X$ be an $\mathcal{M}$-adhesive category and $(\X, \R)$ a
  left-linear rewriting system.  Given $G, H \in \X$, a
  \emph{derivation} $\der{D}$ with source $G$ and target $H$, written
  $\der{D}: G \Mapsto H$, is defined as a sequence
  $\{\dder{D}_i\}_{i=0}^n$ of direct derivations such that
  $\dder{D}_i : G_i \Mapsto G_{i+1}$ is a direct derivation for every
  index $i$ and $G_0=G$, $G_{n+1}=H$. Moreover, we have an \emph{empty derivation} $G : G \Mapsto G$ for each $G \in \X$.  We denote by $\lgh(\der{D})$ the \emph{length} of a derivation $\der{D}$.
 
  Given a derivation $\der{D}=\{\dder{D}_i\}_{i=0}^n$, we define $r(\der{D})  = \{\rho_i\}_{i=0}^n$ as the sequence of 
 rules such that every $\dder{D}_i$ applies rule $\rho_i\in \R$.  
\end{definition}

Derivations naturally compose: given $\der{D}: G \Mapsto H$ and $\der{E}: H \Mapsto F$, we can consider their composition $\der{D} \cdot \der{E}: G \to F$, defined in the obvious way.

Often, we are interested in derivations only up to some notion of coherent isomorphism between them. This leads us to the following definition. 

\noindent
\parbox{8cm}{
  \begin{definition}[Abstraction equivalence]
    Let  $\X$ be an $\mathcal{M}$-adhesive category and 
    $(\X, \R)$ a left-linear rewriting system. 
    An \emph{abstraction equivalence} between derivations $\der{D}$
    and $\der{D'}$ with the same length and
    $r(\der{D})=r(\der{D}')$ is a family of isomorphisms
    $\{\phi_X\}_{X\in \{G_i,D_i\}}$ such that the diagram on the right commutes
    for $i\in [0, \min(0, \lgh(\der{D})-1 )]$. We say that $\der{D}$ are $\der{D}'$ are \emph{abstraction
      equivalent}, written $\der{D}\equiv_a \der{D}'$, if there exists an abstraction equivalence between
    them.
  \end{definition}
}
\parbox{4cm}{\vspace{-1.35em}$\xymatrix@C=40pt{G'_i&D'_i \ar[r]^{g'_i}
    \ar@{>->}[l]_{f'_i}&G'_{i+1}\\ L_i \ar[u]^{m'_i} \ar[d]_{m_i}& K_i
    \ar[u]^{k'_i} \ar[d]_{k_i} \ar[r]^{r_i} \ar@{>->}[l]_{l_i}
    &R_i \ar[u]^{h'_i} \ar[d]_{h_i}\\G_i
    \ar@/_.45cm/[uu]_(.35){\phi_{G_i}}|\hole&D_i\ar@{>->}[l]^{f_i}\ar@/_.45cm/[uu]_(.35){\phi_{D_i}}|\hole
    \ar[r]_{g_i}&G_{i+1}\ar@/_.45cm/[uu]_{\phi_{G_{i+1}}}}$}

$\mathcal{M}$-adhesivity of $\X$ ensures the
uniqueness of the result of applying a rule to an object: two
direct derivations using the same match are abstraction equivalent
(see Proposition~\ref{prop:unique}).

\section{Independence in rewriting}
\label{sec:equi}

In this section we discuss the notion of independence between two
consecutive rewriting steps, a fundamental ingredient of the rewriting
theory, which comes into play when viewing a sequence of steps as a
concurrent computation. We observe that moving from linear to
left-linear rules leads to the failure of various basic properties of
the classical notion of independence used in the linear setting and we
single out a framework in which these can be re-established, possibly
in a weakened form.

\subsection{Sequentially independent and switchable derivations }\label{subsec:switch}

Sequential independence is the canonical notion of independence in the DPO approach.

\begin{definition}[Sequential independence]
  \label{de:sequential-independence}
  Let $\X$ be an $\mathcal{M}$-adhesive category 
  and $(\X, \R)$ a left-linear rewriting system. 
  Let also
  $\der{D}=\{\dder{D}_i\}_{i=0}^1$ be a derivation of length $2$, with
  $\dder{D}_i$ using rule $\rho_i = (l_i,r_i)$.  We say that $\dder{D}_0$ and
  $\dder{D}_1$ are
  \emph{sequentially independent}
  if there is a pair of arrows
  $i_0\colon R_0\to D_1$ and $i_1\colon L_1\to D_0$ such that the
  following diagram commutes. In this case $(i_0,i_1)$ is called an  \emph{independence
    pair}.
  \[\xymatrix@R=16pt@C=15pt{L_0 \ar[d]_{m_0}&& K_0
      \ar[d]_{k_0}\ar@{>->}[ll]_{l_0} \ar[r]^{r_0} & R_0
      \ar@{.>}@/^.35cm/[drrr]|(.3)\hole_(.4){i_0}
      \ar[dr]|(.3)\hole_{h_0} && L_1 \ar@{.>}@/_.35cm/[dlll]^(.4){i_1}
      \ar[dl]|(.3)\hole^{m_1}& K_1 \ar[d]^{k_1}\ar@{>->}[l]_{l_1}
      \ar[rr]^{r_1} && R_1 \ar[d]^{h_1}\\G_0 && \ar@{>->}[ll]^{f_0}
      D_0 \ar[rr]_{g_0}&& G_1 && \ar@{>->}[ll]^{f_1} D_1
      \ar[rr]_{g_1}&& G_2}
  \]
\end{definition}

Intuitively, the existence of the independence pair
captures the possibility of switching the application of the two rules:
$f_0 \circ i_1$ is a match for $\rho_1$ in $G_0$ and the
existence of $i_0 : R_0 \to D_1$ means that $\rho_1$ does not delete
items in the image of $K_0$, thus $\rho_0$ can be applied after
$\rho_1$.

\begin{remark}\label{rem:uni}
	It is worth mentioning that if $(\X, \R)$ is a linear rewriting system then two derivations $\dder{D}_0$ and $\dder{D}_1$ can have at most one independence pair. Indeed if $(i_0,i_1)$ and $(i'_0, i'_1)$ are independence pairs, then
	\[g_0\circ i_1 = g_0\circ i'_1 \qquad f_1\circ i_0=f_1\circ i'_0\]
But in linear systems $g_0$ and $f_1$ are both monos, entailing $i_0=i'_0$ and $i_1=i'_1$.
\end{remark}

We next formalise what it means for a derivation being the switch of another.

\begin{definition}[Switch]
  \label{def:switch}
  Let $\X$ be an $\mathcal{M}$-adhesive category 
  and $(\X, \R)$ a left-linear rewriting system. 
  Let also
  $\der{D}=\{\dder{D}_i\}_{i=0}^1$ be a derivation made
  of two sequentially independent derivations $\dder{D}_0$,
  $\dder{D}_1$ with independence pair $(i_0, i_1)$, as in the diagram
  on the left below. A \emph{switch} of $\der{D}$ along $(i_0,i_1)$ is
  a derivation $\der{E}=\{\dder{E}_i\}_{i=0}^1$, between the same objects and using the same rules
  in reverse order, as in the diagram on the right below
  \[
    \def\objectstyle{\scriptstyle}\def\labelstyle{\scriptstyle}
    
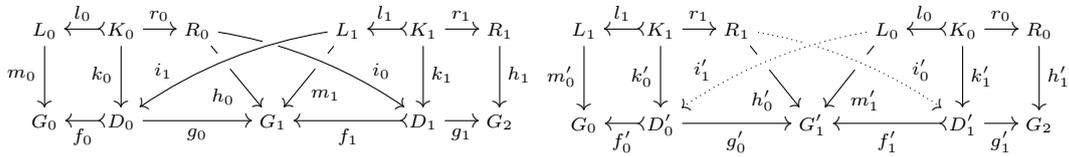
\xymatrix@C=1.1mm@R=8mm{
      {L_0} \ar[d]_{m_0}
      & & {K_0} \ar@{>->}[ll]_{l_0} \ar[rr]^{r_0} \ar[d]_{k_0}
      & & {R_0} \ar[drr]|(.35)\hole_(.6){h_0}  \ar@/^.22cm/[drrrrrr]|(.3)\hole^(.75){i_0}
      & & & & 
      {L_1}\ar[dll]|(.35)\hole^(.6){m_1} \ar@/_.22cm/_(.75){i_1}[dllllll]
      & & {K_1} \ar@{>->}[ll]_{l_1} \ar[rr]^{r_1} \ar[d]^{k_1}
      & & {R_1} \ar[d]^{h_1} \\
      {G_0}
      & & {D_0} \ar@{>->}[ll]^{f_0} \ar[rrrr]_{g_0}
      & & & & {G_1} & &
      & &  {D_1} \ar@{>->}[llll]^{f_1} \ar[rr]_{g_1}
      & & {G_2}
    }
    \xymatrix@C=1.1mm@R=8mm{
      {L_1} \ar[d]_{m_0'}
      & & {K_1} \ar@{>->}[ll]_{l_1} \ar[rr]^{r_1} \ar[d]_{k_0'}
      & & {R_1} \ar[drr]|(.35)\hole_(.6){h_0'}  \ar@/^.22cm/@{.>}^(.75){i_0'}|(.3)\hole[drrrrrr]
      & & & & 
      {L_0}\ar[dll]|(.35)\hole^(.6){m_1'} \ar@/_.22cm/@{.>}_(.75){i_1'}[dllllll] 
      & & {K_0} \ar@{>->}[ll]_{l_0} \ar[rr]^{r_0} \ar[d]^{k_1'}
      & & {R_0} \ar[d]^{h_1'} \\
      {G_0}
      & & {D_0'} \ar@{>->}[ll]^{f_0'} \ar[rrrr]_{g_0'}
      & & & & {G'_1} & &
      & &  {D_1'} \ar@{>->}[llll]^{f_1'} \ar[rr]_{g_1'}
      & & {G_2}  }
  \]
  such that there is an independence pair $(i_0', i_1')$ between
  $\dder{E}_0$ and $\dder{E}_1$ and 
  \begin{center}   
    $m_0=f_0' \circ i_1'$
    \qquad $h_1=g_1' \circ i_0'$
    \qquad $m_0'= f_0 \circ i_1$
    \qquad $h_1'= g_{1}\circ i_0$.
  \end{center}

  If a switch of $\dder{D}_0$ and $\dder{D}_1$ exists we say that
  they are \emph{switchable}.
\end{definition}

It can be shown that switches along the same independence pair are unique up to abstraction
equivalence (see Proposition~\ref{thm:switch_uni}).

\begin{example}
  \label{ex:seq-ind}
  Consider a rewriting system in $\cat{Graph}$, the category of
  directed graphs and graph morphisms, which is adhesive. The set of
  rules $\R = \{ \rho_0, \rho_1, \rho_2\}$ is in
  Fig.~\ref{fi:rules}, where numbers are used to represent the
  morphisms from the interface to the left- and right-hand
  sides. Rules $\rho_0$ and $\rho_1$ are linear: $\rho_0$ generates a
  new node and a new edge, while $\rho_1$ creates a self-loop edge. Instead,
  $\rho_2$ is left-linear but not linear, as it ``merges'' two nodes.

  \begin{figure}
      \begin{center}
    %
    %
    %
    %
  \begin{tikzpicture}[node distance=2mm, font=\small]
      \node (l) {
      \begin{tikzpicture}
        \node at (0,0.53) {};
        \node at (0,0) [node, label=below:$1$] (1) {} ;        
      \pgfBox
      \end{tikzpicture} 
    };
    \node[font=\scriptsize, above] at (l.north) {$L_0$};
    \node [right=of l] (k) {
      \begin{tikzpicture}
        \node at (0,0.53) {}; 
        \node at (0,0) [node, label=below:$1$] (1) {};
      \pgfBox
      \end{tikzpicture} 
    };
    \node[font=\scriptsize, above] at (k.north) {$K_0$};
    \node[below] at (k.south) {$\rho_0$};
    \node  [right=of k] (r) {
      \begin{tikzpicture}
        \node at (0,0.53) {}; 
        \node at (0,0) [node, label=below:$1$] (1) {};
        \node at (.5,0) [node, label=below:$2$] (2) {};
        \draw[coloredge] (1) to[out=20, in=160] (2);
        \pgfBox
      \end{tikzpicture}
    };
    \node[font=\scriptsize, above] at (r.north) {$R_0$};
    \path (k) edge[->] node[trans, above] {} (l);
    \path (k) edge[->] node[trans, above] {} (r);    
    \end{tikzpicture}
    %
    \hspace{1cm}
  %
    %
    \begin{tikzpicture}[node distance=2mm, font=\small]
      \node (l) {
      \begin{tikzpicture}
        \node at (0,0.53) {}; 
        \node at (0,0) [node, label=below:$1$] (1) {} ;
      \pgfBox
      \end{tikzpicture} 
    };
    \node[font=\scriptsize, above] at (l.north) {$L_1$};
    \node [right=of l] (k) {
      \begin{tikzpicture}
        \node at (0,0.53) {}; 
        \node at (0,0) [node, label=below:$1$] (1) {};
        \pgfBox
      \end{tikzpicture} 
    };
    \node[font=\scriptsize, above] at (k.north) {$K_1$};
    \node[below] at (k.south) {$\rho_1$};
    \node  [right=of k] (r) {
      \begin{tikzpicture}
        \node at (0,0) [node, label=below:$1$] (1) {}
         edge [in=55, out=85, loop] ();
        \pgfBox
      \end{tikzpicture}
    };
    \node[font=\scriptsize, above] at (r.north) {$R_1$};
    \path (k) edge[->] node[trans, above] {} (l);
    \path (k) edge[->] node[trans, above] {} (r);
    \end{tikzpicture}
  %
    \hspace{1cm}
  %
    %
    \begin{tikzpicture}[node distance=2mm, font=\small]
      \node (l) {
      \begin{tikzpicture}
        \node at (0,0.53) {}; 
        \node at (0,0) [node, label=below:$1$] (1) {} ;
        \node at (0.5,0) [node, label=below:$2$] (2) {} ;
      \pgfBox
      \end{tikzpicture} 
    };
    \node[font=\scriptsize, above] at (l.north) {$L_2$};
    \node [right=of l] (k) {
      \begin{tikzpicture}
        \node at (0,0.53) {}; 
        \node at (0,0) [node, label=below:$1$] (1) {} ;
        \node at (0.5,0) [node, label=below:$2$] (2) {} ;
      \pgfBox
      \end{tikzpicture} 
    };
    \node[font=\scriptsize, above] at (k.north) {$K_2$};
    \node[below] at (k.south) {$\rho_2$};
    \node  [right=of k] (r) {
      \begin{tikzpicture}
        \node at (0,0.53) {}; 
        \node at (0,0) [node, label=below:$12$] (12) {};
        \pgfBox
      \end{tikzpicture}
    };
    \node[font=\scriptsize, above] at (r.north) {$R_2$};
    \path (k) edge[->] node[trans, above] {} (l);
    \path (k) edge[->] node[trans, above] {} (r);
    \end{tikzpicture}
  %
\end{center}
    \caption{A rewriting system in $\cat{Graph}$.}
    \label{fi:rules}
  \end{figure}
  
  A derivation $\der{D}$ consisting of three steps $\dder{D}_0$,
  $\dder{D}_1$, and $\dder{D}_2$, each applying the corresponding rule
  $\rho_i$, is depicted in Fig.~\ref{fi:derD}. Note that the numbers in $\rho_1$
  were changed consistently to represent the vertical morphisms.
  Steps $\dder{D}_1$ and $\dder{D}_2$ are clearly sequential
  independent via the independence pair $(R_1 \to D_2, L_2 \to D_1)$,
  mapping nodes according to their numbering.
  A switch of $\der{D}$ along such independence pair is the derivation
  $\der{E}$ in Fig.~\ref{fi:derE}. Instead, in $\der{D}$ the first two
  steps applying $\rho_0$ and $\rho_1$ are not sequential independent,
  intuitively because $\rho_1$ uses the node produced by $\rho_0$ for
  attaching a self-loop.
  
  \begin{figure}[t]
      \begin{center}
    \begin{tikzpicture}[node distance=2mm, font=\small, baseline=(current bounding box.center)]      
      \node (L0) at (0,2) {
        \begin{tikzpicture}
          \node at (0,0.53) {};
          \node at (0,0) [node, label=below:$1$] (1) {} ;
          \pgfBox
        \end{tikzpicture} 
      };
      \node [right=of L0] (K0) {
        \begin{tikzpicture}
          \node at (0,0.53) {}; 
          \node at (0,0) [node, label=below:$1$] (1) {};
          \pgfBox
        \end{tikzpicture} 
      };
      \node [above] at (K0.north) {$\rho_0$};
      \node [right=of K0](R0) {
        \begin{tikzpicture}
          \node at (0,0.53) {}; 
          \node at (0,0) [node, label=below:$1$] (1) {};
          \node at (.5,0) [node, label=below:$2$] (2) {};
          \draw[coloredge] (1) to[out=20, in=160] (2);
          \pgfBox
        \end{tikzpicture}
      };
      \path (K0) edge[->] node[trans, above] {} (L0);
      \path (K0) edge[->] node[trans, above] {} (R0);

      \node at (4,2) (L1) {
        \begin{tikzpicture}
          \node at (0,0.53) {}; 
          \node at (0,0) [node, label=below:$2$] (2) {} ;
          \pgfBox
        \end{tikzpicture} 
      };
      \node [right=of L1] (K1) {
        \begin{tikzpicture}
          \node at (0,0.53) {}; 
          \node at (0,0) [node, label=below:$2$] (2) {};
          \pgfBox
        \end{tikzpicture} 
      };
      \node [above] at (K1.north) {$\rho_1$};
      \node [right=of K1] (R1) {
        \begin{tikzpicture}
          \node at (0,0) [node, label=below:$2$] (2) {}
           edge [in=55, out=85, loop] ();        
          \pgfBox
        \end{tikzpicture}
      };
      \path (K1) edge[->] node[trans, above] {} (L1);
      \path (K1) edge[->] node[trans, above] {} (R1);

      \node at (8,2) (L2) {
        \begin{tikzpicture}
          \node at (0,0.53) {}; 
          \node at (0,0) [node, label=below:$1$] (1) {} ;
          \node at (0.5,0) [node, label=below:$2$] (2) {} ;
          \pgfBox
        \end{tikzpicture} 
      };
      \node [right=of L2] (K2) {
        \begin{tikzpicture}
          \node at (0,0.53) {}; 
          \node at (0,0) [node, label=below:$1$] (1) {} ;
          \node at (0.5,0) [node, label=below:$2$] (2) {} ;
          \pgfBox
        \end{tikzpicture} 
      };
      \node [above] at (K2.north) {$\rho_2$};
      \node [right=of K2] (R2) {
        \begin{tikzpicture}
          \node at (0,0.53) {}; 
          \node at (0,0) [node, label=below:$12$] (12) {};
          \pgfBox
        \end{tikzpicture}
      };
      \path (K2) edge[->] node[trans, above] {} (L2);
      \path (K2) edge[->] node[trans, above] {} (R2);

      \node at (0,0) (G0) {
        \begin{tikzpicture}
          \node at (0,0.53) {};
          \node at (0,0) [node, label=below:$1$] (1) {} ;
          \pgfBox
        \end{tikzpicture} 
      };
      \node [right=of G0] (D0) {
        \begin{tikzpicture}
          \node at (0,0.53) {}; 
          \node at (0,0) [node, label=below:$1$] (1) {};
          \pgfBox
        \end{tikzpicture} 
      };
      \node at (3,0) (G1) {
        \begin{tikzpicture}
          \node at (0,0.53) {}; 
          \node at (0,0) [node, label=below:$1$] (1) {} ;
          \node at (0.5,0) [node, label=below:$2$] (2) {} ;
          \draw[coloredge] (1) to[out=20, in=160] (2);
          \pgfBox
        \end{tikzpicture} 
      };
      \path (D0) edge[->] node[trans, above] {} (G0);
      \path (D0) edge[->] node[trans, above] {} (G1);
      \path (L0) edge[->] node[trans, above] {} (G0);
      \path (K0) edge[->] node[trans, above] {} (D0);
      \path (R0) edge[->] node[trans, above] {} (G1);

      \node at (5,0) (D1) {
        \begin{tikzpicture}
          \node at (0,0.53) {};
          \node at (0,0) [node, label=below:$1$] (1) {} ;
          \node at (0.5,0) [node, label=below:$2$] (2) {} ;
          \draw[coloredge] (1) to[out=20, in=160] (2);
          \pgfBox
        \end{tikzpicture} 
      };
      \node at (7,0) (G2) {
        \begin{tikzpicture}
          \node at (0,0.53) {}; 
          \node at (0,0) [node, label=below:$1$] (1) {};
          \node at (0.5,0) [node, label=below:$2$] (2) {}
          edge [in=55, out=85, loop] (); 
          \draw[coloredge] (1) to[out=20, in=160] (2);
          \pgfBox
        \end{tikzpicture} 
      };

      \path (D1) edge[->] node[trans, above] {} (G1);
      \path (D1) edge[->] node[trans, above] {} (G2);
      \path (L1) edge[->] node[trans, above] {} (G1);
      \path (K1) edge[->] node[trans, above] {} (D1);
      \path (R1) edge[->] node[trans, above] {} (G2);
      
      \node at (9.46,0) (D2) {
        \begin{tikzpicture}
          \node at (0,0) [node, label=below:$1$] (1) {};
          \node at (0.5,0) [node, label=below:$2$] (2) {}
          edge [in=55, out=85, loop] ();
          \draw[coloredge] (1) to[out=20, in=160] (2);
          \pgfBox
        \end{tikzpicture} 
      };
      \node [right=of D2] (G3) {
        \begin{tikzpicture}
          \node at (0,0) [node, label=below:$12$] (12) {}
          edge [in=55, out=85, loop] ()
          edge [in=125, out=155, colorloop] ();
          \pgfBox
        \end{tikzpicture} 
      };
      \node[font=\scriptsize, below] at (G0.south) {$G_0$};
      \node[font=\scriptsize, below] at (D0.south) {$D_0$};      
      \node[font=\scriptsize, below] at (G1.south) {$G_1$};
      \node[font=\scriptsize, below] at (D1.south) {$D_1$};      
      \node[font=\scriptsize, below] at (G2.south) {$G_2$};
      \node[font=\scriptsize, below] at (D2.south) {$D_2$};      
      \node[font=\scriptsize, below] at (G3.south) {$G_3$};      
      \path (D2) edge[->] node[trans, above] {} (G2);
      \path (D2) edge[->] node[trans, above] {} (G3);
      \path (L2) edge[->] node[trans, above] {} (G2);
      \path (K2) edge[->] node[trans, above] {} (D2);
      \path (R2) edge[->] node[trans, above] {} (G3);
    \end{tikzpicture}
\end{center}
    \caption{The derivation $\der{D}$.}
    \label{fi:derD}
  \end{figure}
  \begin{figure}[b]
      \begin{center}
    \begin{tikzpicture}[node distance=2mm, font=\small,  baseline=(current bounding box.center)]      
      \node (L1) at (0,2) {
        \begin{tikzpicture}
          \node at (0,0.53) {};
          \node at (0,0) [node, label=below:$1$] (1) {} ;
          \pgfBox
        \end{tikzpicture} 
      };
      \node [right=of L1] (K1) {
        \begin{tikzpicture}
          \node at (0,0.53) {}; 
          \node at (0,0) [node, label=below:$1$] (1) {};
          \pgfBox
        \end{tikzpicture} 
      };
      \node [above] at (K1.north) {$\rho_0$};
      \node [right=of K1](R1) {
        \begin{tikzpicture}
          \node at (0,0.53) {}; 
          \node at (0,0) [node, label=below:$1$] (1) {};
          \node at (.5,0) [node, label=below:$2$] (2) {}; 
          \draw[coloredge] (1) to[out=20, in=160] (2);
          \pgfBox
        \end{tikzpicture}
      };
      \path (K1) edge[->] node[trans, above] {} (L1);
      \path (K1) edge[->] node[trans, above] {} (R1);

      \node at (4,2) (L3) {
        \begin{tikzpicture}
          \node at (0,0.53) {}; 
          \node at (0,0) [node, label=below:$1$] (1) {} ;
          \node at (0.5,0) [node, label=below:$2$] (2) {} ;
          \pgfBox
        \end{tikzpicture} 
      };
      \node [right=of L3] (K3) {
        \begin{tikzpicture}
          \node at (0,0.53) {}; 
          \node at (0,0) [node, label=below:$1$] (1) {} ;
          \node at (0.5,0) [node, label=below:$2$] (2) {} ;
          \pgfBox
        \end{tikzpicture} 
      };
      \node [above] at (K3.north) {$\rho_2$};
      \node [right=of K3] (R3) {
        \begin{tikzpicture}
          \node at (0,0.53) {}; 
          \node at (0,0) [node, label=below:$12$] (12) {};
          \pgfBox
        \end{tikzpicture}
      };
      \path (K3) edge[->] node[trans, above] {} (L3);
      \path (K3) edge[->] node[trans, above] {} (R3);

      \node at (8,2) (L2) {
        \begin{tikzpicture}
          \node at (0,0.53) {}; 
          \node at (0,0) [node, label=below:$12$] (1) {} ;
          \pgfBox
        \end{tikzpicture} 
      };
      \node [right=of L2] (K2) {
        \begin{tikzpicture}
          \node at (0,0.53) {}; 
          \node at (0,0) [node, label=below:$12$] (1) {};
          \pgfBox
        \end{tikzpicture} 
      };
      \node [above] at (K2.north) {$\rho_1$};

      \node [right=of K2] (R2) {
        \begin{tikzpicture}
          \node at (0,0) [node, label=below:$12$] (1) {}
          edge [in=55, out=85, loop] ();
          \pgfBox
        \end{tikzpicture}
      };
      \path (K2) edge[->] node[trans, above] {} (L2);
      \path (K2) edge[->] node[trans, above] {} (R2);

      \node at (0,0) (G1) {
        \begin{tikzpicture}
          \node at (0,0.53) {};
          \node at (0,0) [node, label=below:$1$] (1) {} ;
          \pgfBox
        \end{tikzpicture} 
      };
      \node [right=of G1] (D1) {
        \begin{tikzpicture}
          \node at (0,0.53) {}; 
          \node at (0,0) [node, label=below:$1$] (1) {};
          \pgfBox
        \end{tikzpicture} 
      };
      \node at (3,0) (G2) {
        \begin{tikzpicture}
          \node at (0,0.53) {}; 
          \node at (0,0) [node, label=below:$1$] (1) {} ;
          \node at (0.5,0) [node, label=below:$2$] (2) {} ;
          \draw[coloredge] (1) to[out=20, in=160] (2);
          \pgfBox
        \end{tikzpicture} 
      };
      \path (D1) edge[->] node[trans, above] {} (G1);
      \path (D1) edge[->] node[trans, above] {} (G2);
      \path (L1) edge[->] node[trans, above] {} (G1);
      \path (K1) edge[->] node[trans, above] {} (D1);
      \path (R1) edge[->] node[trans, above] {} (G2);

      \node at (5.5,0) (D2) {
        \begin{tikzpicture}
          \node at (0,0.53) {};
          \node at (0,0) [node, label=below:$1$] (1) {} ;
          \node at (0.5,0) [node, label=below:$2$] (2) {} ;
          \draw[coloredge] (1) to[out=20, in=160] (2);
          \pgfBox
        \end{tikzpicture} 
      };
      \node at (7.4,0) (G3) {
        \begin{tikzpicture}
          \node at (0,0.53) {}; 
          \node at (0,0) [node, label=below:$12$] (12) {}
          edge [in=125, out=155, colorloop] ();
          \pgfBox
        \end{tikzpicture} 
      };

      \path (D2) edge[->] node[trans, above] {} (G2);
      \path (D2) edge[->] node[trans, above] {} (G3);
      \path (L3) edge[->] node[trans, above] {} (G2);
      \path (K3) edge[->] node[trans, above] {} (D2);
      \path (R3) edge[->] node[trans, above] {} (G3);
      
      \node at (9.1,0) (D3) {
        \begin{tikzpicture}
          \node at (0,0.53) {}; 
          \node at (0,0) [node, label=below:$12$] (12) {}
          edge [in=125, out=155, colorloop] ();
          \pgfBox
        \end{tikzpicture} 
      };
      \node [right=of D3] (G4) {
        \begin{tikzpicture}
          \node at (0,0) [node, label=below:$12$] (12) {}
          edge [in=55, out=85, loop] ()
          edge [in=125, out=155, colorloop] ();
          \pgfBox
        \end{tikzpicture} 
      };

      \node[font=\scriptsize, below] at (G1.south) {$G_0$};
      \node[font=\scriptsize, below] at (D1.south) {$D_0$};      
      \node[font=\scriptsize, below] at (G2.south) {$G_1$};
      \node[font=\scriptsize, below] at (D2.south) {$D_1'$};      
      \node[font=\scriptsize, below] at (G3.south) {$G_2'$};
      \node[font=\scriptsize, below] at (D3.south) {$D_2'$};      
      \node[font=\scriptsize, below] at (G4.south) {$G_3$};      

      \path (D3) edge[->] node[trans, above] {} (G3);
      \path (D3) edge[->] node[trans, above] {} (G4);
      \path (L2) edge[->] node[trans, above] {} (G3);
      \path (K2) edge[->] node[trans, above] {} (D3);
      \path (R2) edge[->] node[trans, above] {} (G4);
    \end{tikzpicture}
\end{center}
    \caption{The derivation $\der{E}$.}
    \label{fi:derE}
  \end{figure}
\end{example}

As we mentioned, for linear rewriting systems
it is canonical to identify
derivations that are equal ``up to switching'', i.e.~that differ
only in the order of independent steps. The same notion can be given
for left-linear rewriting systems by relying on the notion of switch.

We  recall some notations on permutations.  A \emph{permutation} on
$\interval[0]{n}$ is a bijection
$\sigma : \interval[0]{n} \to \interval[0]{n}$. It is a
\emph{transposition} $\nu$ if there are $i, j \in \interval{n}$,
$i \neq j$ such that $\sigma(i)=j$, $\sigma(j) = i$, and
$\sigma(k) = k$ otherwise; it is denoted as $\transp{i}[j]$. Given a permutation $\perm$, 
an \emph{inversion} for $\sigma$ is a pair $(i,j)$ such that $i<j$ and
$\sigma(j)< \sigma(i)$; $\inv{\sigma}$ denotes the set of
inversions of $\perm$.

Switch equivalence is now defined as the equivalence relating derivations that 
can be obtained one from another by a sequence of switches. Moreover, intermediate graphs can be taken up to isomorphism according to abstraction equivalence.

\begin{definition}[Switch equivalence]
  \label{de:switch-equivalence}
  Let $(\X, \R)$ be a left-linear rewriting system.  Let also
  $\der{D}, \der{E} : G \Mapsto H$ be derivations with the same
  length, $\der{D}=\{\dder{D}_{i}\}_{i=0}^n$ and
  $\der{E}=\{\dder{E}_{i}\}_{i=0}^n$. If
  $\dder{D}_i \cdot \dder{D}_{i+1}$ is a switch of
  $\dder{E}_i \cdot \dder{E}_{i+1}$ for some $i \in [0,n-1]$ and  $\dder{D}_j = \dder{E}_j$ for each $j \not \in \{i,i+1\}$ then we write
  $\der{D} \shift{\transp{i}} \der{E}$. 
 A \emph{switching sequence} is a sequence $\{\der{D}_{k}\}_{k=0}^m$
  of derivations such that
    $\der{D}_0 \shift{\nu_1} \der{D}_1 \shift{\nu_2} \ldots
    \shift{\nu_m} \der{D}_m$  with $\nu_{k} = \transp{i_k}$.
 
    Let us denote by $\nu_{k,h}$ the composition
  $\nu_h \circ \nu_{h-1} \circ \ldots \nu_k$. We say that the
  switching sequence \emph{consists of inversions} if for all
  $k \in \interval[0]{m}$ the transposition $\nu_k$ is an inversion
  for $\nu_{1,m}$.
  
  Two derivations $\der{D}, \der{E}:G\Mapsto H$ are \emph{switch
    equivalent}, written $\der{D}\equiv^{sh} \der{E}$, if there is a
  switching sequence $\{\der{D}_{k}\}_{k=0}^m$ such that
  $\der{D}\equiv_a \der{D}_0 \shift{\nu_1} \der{D}_1 \shift{\nu_2}
  \ldots \shift{\nu_m} \der{D}_m \equiv _a \der{E}$.   
  To point out a chosen permutation of rewriting steps, we will also write $\der{D}\equiv^{sh}_{\sigma} \der{E}$, 
  where $\sigma$ is the composition of the transposition appearing in a chosen switching sequence. 
\end{definition}

\begin{remark}\label{rem:abst}
	The possibility of an empty switching sequence assures that two abstraction equivalent derivations are switch equivalent.
\end{remark}

\subsection{Church-Rosser: Sequential independence and switchability}\label{subsec:CR}

The fact that sequential independence implies switchability always
holds for linear rules (see \Cref{prop:equi}). The result is so
indispensable that typically, in the literature, it is not even
stated, in the sense that switchability is not introduced
axiomatically as in Definition~\ref{def:switch}, but it is based on
the explicit construction of a switch.

For left-linear rewriting systems
sequential independence does not imply switchability (while the
converse implication clearly holds), as shown by the next example.

\begin{figure}
  \begin{subfigure}{0.2\textwidth}
    {\small
      \xymatrix@C=16pt@R=1pt{
        & 0 \ar@{-}[dd] \ar@{-}[ddddddl] \ar@{-}[ddddddr] &  \\
        &  & \\
        & 1 \ar@{-}[dd] \ar@{-}[ddddl]  \ar@{-}[ddddr]   &  \\
        &  & \\
        & 2 \ar@{.}[d] \ar@{-}[ddl]   \ar@{-}[ddr]     &  \\
        &  \ar@{.}[dl]   \ar@{.}[dr] & \\
        b \ar@{-}[dr] & & c \ar@{-}[dl]\\
        & a 
      }
    }
    \caption{Poset $(P, \sqsubseteq)$}
    \label{fig:diff1:P}
  \end{subfigure}
  \begin{subfigure}{0.58\textwidth}
    \xymatrix@C=20pt{
      a \ar[d]_{\sqsubseteq} & a \ar[d]_{\sqsubseteq} \ar@{>->}[l]_{\id{a}} \ar[r]^{\sqsubseteq} & 
      0 \ar@{>->}@/^.35cm/[drrr]|(.3)\hole^(.75){\id{0}}
      \ar@{>->}[dr]|(.3)\hole_{\id{0}} &  &
      a \ar@/_.35cm/[dlll]_(.75){\sqsubseteq}
      \ar[dl]|(.3)\hole^{\sqsubseteq} & 
      a \ar[d]^{\sqsubseteq}\ar@{>->}[l]_{\id{a}} \ar[r]^{\sqsubseteq} &
      b \ar[d]_{\sqsubseteq}\\
      c &
      \ar@{>->}[l]^{\id{c}} c \ar[rr]_{\sqsubseteq}& &  0 &&
      \ar@{>->}[ll]^{\id{0}} 0 \ar@{>->}[r]_{\id{0}}& 0
    }
    \caption{Two rewriting steps}
    \label{fig:diff1:rew}
  \end{subfigure}
  \begin{subfigure}{0.20\textwidth}
    \xymatrix@C=20pt{a \ar[d]^{\sqsubseteq}& a
      \ar[d]^{\sqsubseteq}\ar@{>->}[l]_{\id{a}} \ar[r]^{\sqsubseteq} & b \\c & c
      \ar@{>->}[l]^{\id{c}}
    }
    \caption{No rewriting step}
    \label{ex:diff1:no-step}
  \end{subfigure}
  \caption{Sequential independence does not imply switchability}
  \label{fig:diff1} 
\end{figure}
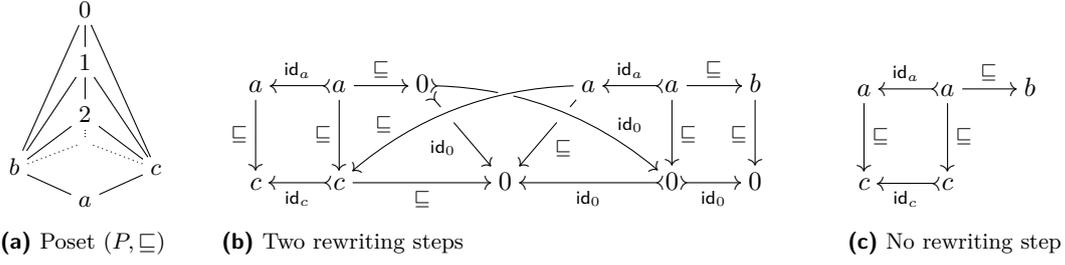

\begin{example}
  \label{ex:diff1}
  Consider the poset $(P, \sqsubseteq)$ in Fig.~\ref{fig:diff1:P} where
  $P = \mathbb{N} \cup \{a,b,c\}$ and $\sqsubseteq$ is defined by
  $m \sqsubseteq n$ if $m \geq n$, $a \sqsubseteq x$ for all $x \in P$
  and $b, c \sqsubseteq n$ for all $n \in \mathbb{N}$.
  Let $\X$ be the  
  category associated with this order, which by \Cref{rem:iso} is
  $\mathsf{I(\X)}$-adhesive. 
  Consider a rewriting
  system whose set of rules contains the following 
  \[\xymatrix{a & a \ar[r]^{\sqsubseteq} \ar@{>->}[l]_{\id{a}} & 0 & a & a
      \ar[r]^{\sqsubseteq} \ar@{>->}[l]_{\id{a}} & b }\]
  We can then consider the  derivation
  $\der{D}=\{\dder{D}_i\}_{i=0}^1$ in Fig.~\ref{fig:diff1:rew}.
  %
  Note that the two direct derivations are sequential
  independent. However there is no switch since the rule applied by
  $\dder{D}_1$ cannot be applied to $c$. In fact, there is a unique morphism
  $a\to c$, yielding the diagram in Fig.~\ref{ex:diff1:no-step}. But $b$ and $c$
  do not have a supremum in the poset underlying $\X$, thus the arrows
  $a\to b$ and $a\to c$ do not have a
  pushout. Hence we do not get a direct derivation from $c$.
\end{example}

The conditions guaranteeing switchability are inspired
by the notion of \emph{canonical filler}
\cite{heindel2009category}.

\begin{definition}[Strong independence pair]
  \label{def:filler}
  Let $\X$ be an $\mathcal{M}$-adhesive category 
  and $(\X, \R)$ a left-linear rewriting system. 
  Let also $(i_0, i_1)$ be an independence pair between two direct
  derivations $\dder{D}_0$, $\dder{D}_1$
  as in the solid part of the diagram below
  
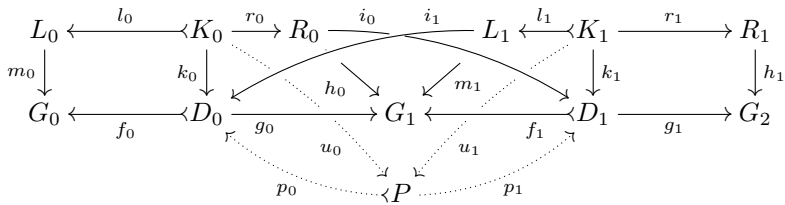
\[
    \xymatrix@C=18pt@R=17pt{L_0 \ar[d]_{m_0}&& K_0
      \ar[d]_{k_0}\ar@{>->}[ll]_{l_0} \ar[r]^{r_0} \ar@{.>}@/^.20cm/[ddrr]|(.32)\hole|(.57)\hole_(.7){u_0} & R_0
      \ar@/^.35cm/[drrr]|(.3)\hole^(.2){i_0} \ar[dr]|(.3)\hole_{h_0}
      && L_1 \ar@/_.35cm/[dlll]_(.2){i_1} \ar[dl]|(.3)\hole^{m_1}& K_1
      \ar[d]^{k_1}\ar@{>->}[l]_{l_1} \ar[rr]^{r_1} \ar@{.>}@/_.20cm/[ddll]|(.32)\hole|(.57)\hole^(.7){u_1} && R_1 \ar[d]^{h_1} \\
      G_0 &&
      \ar@{>->}[ll]^{f_0} D_0 \ar[rr]_(.3){g_0}&& G_1 &&
      \ar@{>->}[ll]^(.3){f_1} D_1 \ar[rr]_{g_1}&& G_2\\
      & & & & P \ar@/^.3cm/@{>.>}[ull]^{p_0} \ar@/_.3cm/@{.>}[urr]_{p_1}
    }
  \]

  Consider the pullback of $g_0$ and $f_1$, which yields $p_i : P \to D_i$
  for $i \in \{0,1\}$ and the mediating arrows $u_i\colon K_i\to P$ 
  for $i \in \{0,1\}$ 
  into the pullback object (see \Cref{prop:tec} for details). 
  We say that $(i_0, i_1)$ is a \emph{strong independence pair} if
  the first two squares depicted  
  below are pushouts and if the pushout of $r_1 : K_1 \to R_1$ and $u_1 : K_1 \to P$ exists
  \[
    \xymatrix@R=16pt{
      K_0 \ar[r]^{r_0} \ar[d]_{u_0}& R_0 \ar[d]^{i_0} & K_1
      \ar@{>->}[r]^{l_1} \ar[d]_{u_1}& L_1 \ar[d]^{i_1}
       &K_1 \ar[r]^{r_1} \ar[d]_{u_1}& R_1\ar@{.>}[d]^{j_0}
      \\
      P \ar[r]_{p_1} & D_1 &P \ar[r]_{p_0} & D_0
       & P \ar@{.>}[r]_{q_1} & Q_1
    }
  \]
\end{definition}

We can now prove a Local Church-Rosser Theorem for strong independence pairs.

\begin{restatable}[Local Church-Rosser Theorem]{proposition}{prChurch}
  \label{pr:church}
  Let $(i_0, i_1)$ be a strong independence pair
  between $\dder{D}_0$ and $\dder{D}_1$. Then $\dder{D}_0$ and
  $\dder{D}_1$ are switchable.
\end{restatable}

The correspondence between sequential independence and switchability
is fundamental. We name the class of rewriting systems where this property holds.

\begin{definition}[Strong enforcing rewriting systems]
  A left-linear rewriting system is \emph{strong enforcing} if
  every independence pair between two direct derivations is strong.
\end{definition}

We can identify a large class of adhesive categories such that all
left-linear rewriting systems over such categories are strong
enforcing. This class includes $\cat{Set}$ and it is closed
under comma and functor category constructions.  As such, it includes
essentially all categories (e.g., presheaves over set) that are
typically considered for modelling purposes. Notably, it contains the
category $\cat{Graph}$ of directed graphs.
This is a natural generalisation from adhesive to $\mathcal{M}$-adhesive categories of a class studied in \cite{baldan2011adhesivity} (see \Cref{app:fill} for details).

Still, there are $\mathcal{M}$-adhesive rewriting systems that are not strong enforcing.

\begin{example}[Non-strong enforcing left-linear rewriting system]
  \label{ex:diff2}
  In light of \Cref{pr:church}, \Cref{ex:diff1} provides an example of
  an independence pair that is not strong. This gives an example of a
  left-linear rewriting system in a $\mathcal{M}$-adhesive category that is quite
  pathological since $\mathcal{M}=\mathsf{I(\X)}$. However, this is expected, 
  since all natural examples seem to belong to the well-behaved class mentioned above. 
\end{example}

\begin{remark}
  By \Cref{pr:church} the existence of a strong independence pair
  entails switchability, which in turn entails sequential
  independence by construction. Strong enforcing rewriting systems are exactly
  those rewriting systems in which these three notions coincide.
\end{remark}

Consistency with the theory of linear rewriting systems is ensured by the fact that all linear rewriting systems are strong enforcing (see \cref{prop:equi}).

\subsection{Well-switching rewriting systems}
\label{subsec:verytame}

Even if we work in strong enforcing rewriting systems where sequential
independence ensures switchability, when dealing with left-linear
rules there is a further, possibly more serious issue, 
namely that there can be more than one independence pair between
the same derivations (cfr. \Cref{rem:uni}). This hinders the very idea of using
sequential independence as a basis of a theory of concurrency for
rewriting systems, since exchanges performed using different
independence pairs may lead to derivations that are not abstraction
equivalent, thus equating computations that should
definitively be taken apart, as shown in the example below.

\begin{example}
  Consider the derivation $\der{E}$ from Example~\ref{ex:seq-ind} (see Fig.~\ref{fi:derE}).
   The last two steps are sequential independent, but one easily sees
  that there are two distinct independence pairs, as the left-hand side
  of $\rho_1$ can be mapped either to node $1$ or to node
  $2$ in $D_1'$. Correspondingly, there
  are two switches of $\der{E}$: one is the derivation $\der{D}$ in
  Fig.~\ref{fi:derD} we started from, the other is the derivation
  $\der{D}'$ in Fig.~\ref{fi:derD1}.
  
  \begin{figure}
      \begin{center}
    \begin{tikzpicture}[node distance=2mm, font=\small, baseline=(current bounding box.center)]      
      \node (L1) at (0,2) {
        \begin{tikzpicture}
          \node at (0,0.53) {};
          \node at (0,0) [node, label=below:$1$] (1) {} ;
          \pgfBox
        \end{tikzpicture} 
      };
      \node [right=of L1] (K1) {
        \begin{tikzpicture}
          \node at (0,0.53) {}; 
          \node at (0,0) [node, label=below:$1$] (1) {};
          \pgfBox
        \end{tikzpicture} 
      };
      \node [above] at (K1.north) {$\rho_0$};
      \node [right=of K1](R1) {
        \begin{tikzpicture}
          \node at (0,0.53) {}; 
          \node at (0,0) [node, label=below:$1$] (1) {};
          \node at (.5,0) [node, label=below:$2$] (2) {};
          \draw[coloredge] (1) to[out=20, in=160] (2);
          \pgfBox
        \end{tikzpicture}
      };
      \path (K1) edge[->] node[trans, above] {} (L1);
      \path (K1) edge[->] node[trans, above] {} (R1);

      \node at (4,2) (L2) {
        \begin{tikzpicture}
          \node at (0,0.53) {}; 
          \node at (0,0) [node, label=below:$1$] (1) {} ;
          \pgfBox
        \end{tikzpicture} 
      };
      \node [right=of L2] (K2) {
        \begin{tikzpicture}
          \node at (0,0.53) {}; 
          \node at (0,0) [node, label=below:$1$] (1) {};
          \pgfBox
        \end{tikzpicture} 
      };
            \node [above] at (K2.north) {$\rho_1$};
      \node [right=of K2] (R2) {
        \begin{tikzpicture}
          \node at (0,0) [node, label=below:$1$] (1) {}
          edge [in=55, out=85, loop] ();        
          \pgfBox
        \end{tikzpicture}
      };
      \path (K2) edge[->] node[trans, above] {} (L2);
      \path (K2) edge[->] node[trans, above] {} (R2);

      \node at (8,2) (L3) {
        \begin{tikzpicture}
          \node at (0,0.53) {}; 
          \node at (0,0) [node, label=below:$1$] (1) {} ;
          \node at (0.5,0) [node, label=below:$2$] (2) {} ;
          \pgfBox
        \end{tikzpicture} 
      };
      \node [right=of L3] (K3) {
        \begin{tikzpicture}
          \node at (0,0.53) {}; 
          \node at (0,0) [node, label=below:$1$] (1) {} ;
          \node at (0.5,0) [node, label=below:$2$] (2) {} ;
          \pgfBox
        \end{tikzpicture} 
      };
      \node [above] at (K3.north) {$\rho_2$};
      \node [right=of K3] (R3) {
        \begin{tikzpicture}
          \node at (0,0.53) {}; 
          \node at (0,0) [node, label=below:$12$] (12) {};
          \pgfBox
        \end{tikzpicture}
      };
      \path (K3) edge[->] node[trans, above] {} (L3);
      \path (K3) edge[->] node[trans, above] {} (R3);

      \node at (0,0) (G1) {
        \begin{tikzpicture}
          \node at (0,0.53) {};
          \node at (0,0) [node, label=below:$1$] (1) {} ;
          \pgfBox
        \end{tikzpicture} 
      };
      \node [right=of G1] (D1) {
        \begin{tikzpicture}
          \node at (0,0.53) {}; 
          \node at (0,0) [node, label=below:$1$] (1) {};
          \pgfBox
        \end{tikzpicture} 
      };
      \node at (3,0) (G2) {
        \begin{tikzpicture}
          \node at (0,0.53) {}; 
          \node at (0,0) [node, label=below:$1$] (1) {} ;
          \node at (0.5,0) [node, label=below:$2$] (2) {} ;
          \draw[coloredge] (1) to[out=20, in=160] (2);
          \pgfBox
        \end{tikzpicture} 
      };
      \path (D1) edge[->] node[trans, above] {} (G1);
      \path (D1) edge[->] node[trans, above] {} (G2);
      \path (L1) edge[->] node[trans, above] {} (G1);
      \path (K1) edge[->] node[trans, above] {} (D1);
      \path (R1) edge[->] node[trans, above] {} (G2);

      \node at (5,0) (D2) {
        \begin{tikzpicture}
          \node at (0,0.53) {};
          \node at (0,0) [node, label=below:$1$] (1) {} ;
          \node at (0.5,0) [node, label=below:$2$] (2) {} ;
          \draw[coloredge] (1) to[out=20, in=160] (2);
          \pgfBox
        \end{tikzpicture} 
      };
      \node at (7,0) (G3) {
        \begin{tikzpicture}
          \node at (0,0.53) {}; 
          \node at (0,0) [node, label=below:$1$] (1) {}
          edge [in=55, out=85, loop] ();
          \node at (0.5,0) [node, label=below:$2$] (2) {} ;
          \draw[coloredge] (1) to[out=20, in=160] (2);
          \pgfBox
        \end{tikzpicture} 
      };

      \path (D2) edge[->] node[trans, above] {} (G2);
      \path (D2) edge[->] node[trans, above] {} (G3);
      \path (L2) edge[->] node[trans, above] {} (G2);
      \path (K2) edge[->] node[trans, above] {} (D2);
      \path (R2) edge[->] node[trans, above] {} (G3);
      
      \node at (9.46,0) (D3) {
        \begin{tikzpicture}
          \node at (0,0) [node, label=below:$1$] (1) {}
          edge [in=55, out=85, loop] ();
          \node at (0.5,0) [node, label=below:$2$] (2) {} ;
          \draw[coloredge] (1) to[out=20, in=160] (2);
          \pgfBox
        \end{tikzpicture} 
      };
      \node [right=of D3] (G4) {
        \begin{tikzpicture}
          \node at (0,0) [node, label=below:$12$] (12) {}
          edge [in=55, out=85, loop] ()
          edge [in=125, out=155, colorloop] ();
          \pgfBox
        \end{tikzpicture} 
      };
      \node[font=\scriptsize, below] at (G1.south) {$G_0$};
      \node[font=\scriptsize, below] at (D1.south) {$D_0$};      
      \node[font=\scriptsize, below] at (G2.south) {$G_1$};
      \node[font=\scriptsize, below] at (D2.south) {$D_1''$};      
      \node[font=\scriptsize, below] at (G3.south) {$G_2''$};
      \node[font=\scriptsize, below] at (D3.south) {$D_2''$};      
      \node[font=\scriptsize, below] at (G4.south) {$G_3$};      

      \path (D3) edge[->] node[trans, above] {} (G3);
      \path (D3) edge[->] node[trans, above] {} (G4);
      \path (L3) edge[->] node[trans, above] {} (G3);
      \path (K3) edge[->] node[trans, above] {} (D3);
      \path (R3) edge[->] node[trans, above] {} (G4);
    \end{tikzpicture}
\end{center}
    \caption{The derivation $\der{D}'$.}
    \label{fi:derD1}
  \end{figure}
  
  As a consequence, $\der{D}$ and $\der{D}'$ would be switch
  equivalent, but this is not acceptable when viewing equivalence
  classes of derivations as concurrent computations: in $\der{D}$ the
  first two steps are not sequential independent, while in $\der{D}'$
  they are, intuitively because in $\der{D}$ rule $\rho_1$ uses the
  node generated by $\rho_0$ (adding a self-loop to it), while in
  $\der{D}'$ rule $\rho_1$ uses the node that was in the initial
  graph. Also observe that the graphs $G_2$ and $G_2'$ produced after
  two steps in $\der{D}$ and $\der{D'}$ are not isomorphic. From the
  technical point of view, the property of being switch equivalent is
  not closed by prefix, and this prevents deriving a sensible
  concurrent semantics: In fact
  $\der{D} = \dder{D}_0\cdot \dder{D}_1 \cdot \dder{D}_2$ and
  $\der{D}' = \dder{D}_0'\cdot \dder{D}_1' \cdot \dder{D}_2'$ are
  switch equivalent, while if we consider the first two steps,
  derivations $\dder{D}_0 \cdot \dder{D}_1$ and
  $\dder{D}_0' \cdot \dder{D}_1'$ are not switch equivalent.
\end{example} 

\full{Moreover, limiting sequential independence to the case in which the independence
pair is unique (as suggested in~\cite{baldan2017domains}) brings to the same problems.}

For these reasons we believe a theory of rewriting for
left-linear rules in adhesive categories should be developed for
systems where the uniqueness of the independence pair is ensured.

\begin{definition}[Well-switching rewriting systems]
A left-linear rewriting system $(\X, \R)$ is \emph{well-switching} if it is strong enforcing and, for every derivation $\der{D}:=\{\dder{D}_{i}\}_{i=0}^1$, there is at most one independence pair between $\dder{D}_0$ and $\dder{D}_1$.
\end{definition}

Clearly, linear rewriting systems are well-switching (see Proposition~\ref{pr:weak}).
Moreover, we next observe that various classes of rewriting systems,
comprising all the ones used in modelling the applications to the
encoding of process calculi or of bio and chemical systems mentioned
in the introduction, are actually well-switching.

A first class consists of those rewriting systems over
possibly hierarchical graphical structures obtained as algebras of
suitable signatures where rules are constrained not to merge elements
of top level sorts in the hierarchy (for graphs, nodes can be
merged while edges cannot). The idea here is to consider graph
structures as presheaves on categories in which there are objects that play 
the role of \emph{roots}, i.e.~objects that are not the codomain of 
any arrow besides the identity. 

\begin{definition}[Root-preserving graphical rewriting systems]
	Let $\X$ be a category, an object $X\in \X$ is a \emph{root} if the only arrow with codomain $X$ is $\id{X}$.
	The category $\gph{X}$ of \emph{$\X$-graphs} is the category
	$\Set^{\X}$. A
	\emph{root-preserving graphical rewriting system} is a left-linear
	rewriting system $(\gph{X}, \R)$ such that for each rule
	$(l\colon K\to L, r\colon K\to R)$ in $\R$ it holds
	\begin{enumerate}
		\item for every $X\in \X$ and $x\in L(X)$, there exists a root $Y$
		 and an arrow $f\colon Y\to X$ such that $x$ belongs to the image of
		$L(f)\colon L(Y)\to L(X)$;
		\item $r\colon K\to R$ is mono on the roots, i.e.~for every root $X\in \X$ the component $r_X:K(X)\to R(X)$ is injective.
	\end{enumerate}
\end{definition}

For instance, the category $\cat{Graph}$ can be obtained
by taking as $\X$ the category generated by $E \rightrightarrows V$. In this case $E$ is the only root, hence, condition $1$ asks that in the left-hand side of each rule
there are no isolated nodes, while condition $2$ asks that the
morphism $r: K \to R$ is injective on edges, i.e.~it can only merge nodes.

\begin{restatable}{lemma}{lemVTame}
\label{bono}
  All root-preserving graphical rewriting systems are well-switching.
\end{restatable}

Another interesting class of well-switching rewriting systems is given by e-graphs.

\begin{example}[E-graphs]
  Consider the category $\cat{EGraphs}$, where objects are (directed)
  graphs endowed with an equivalence over nodes, and arrows are graph
  morphisms that preserve the equivalence, as considered
  in~\cite{BaldanGM06}, closely related to e-graphs~\cite{WNW:egg}. 
  Formally, $\cat{EGraphs}$ can be seen as the
  subcategory of the presheaf
  $[E \rightrightarrows V \to Q, \cat{Set}]$ where objects are
  constrained to have the component $V \to Q$ surjective. 
  
  Explicitly, an
  e-graph $G$ is a triple $\langle s_G, t_G, q_G \rangle$ where
  $s_G, t_G: E_G \rightrightarrows V_G$ provides the graphical
  structure, while the surjective function $q_G : V_G \to Q_G$ maps
  each node to the corresponding equivalence class. 
  Notice that the inclusion functor into  $[E \rightrightarrows V \to Q, \cat{Set}]$ 
  creates pullbacks and pushouts \cite{mac2013categories}, so that they are computed component-wise.
  
  A morphism in $\cat{EGraph}$ is mono if the components over $E$
  and $V$ are mono, i.e.~if it is mono as a morphism in
  $\cat{Graph}$. It is regular mono if also the component on
  $Q$ is mono, i.e.~if it reflects equivalence classes besides
  preserving them. This characterisation of regular monos and the fact that pullbacks and pushouts are computed component-wise allows us to prove quasi-adhesivity of $\cat{EGraphs}$ at once. Moreover, one can deduce that every rewriting system $(\X, \R)$ that is left-linear with respect to $\reg(\cat{EGraphs})$ is strong enforcing: this is done exploiting again the inclusion functor into $[E \rightrightarrows V \to Q, \cat{Set}]$.
  
Left-linear rewriting systems with respect to $\reg(\cat{EGraphs})$ are well-switching. They have been used in~\cite{BaldanGM06} for the graphical implementation of nominal calculi, where,
differently from~\cite{Gad07}, as a result of name passing the received name is not merged with a local one, but put in the same equivalence class, better tracing the causal dependencies among reductions.
\end{example}

\subsection{A canonical form for  switch equivalences}
\label{subsec:canonical}

As discussed in the introduction, in the linear case a number of
basic properties of switch equivalence, i.e.~the globality of
independence, the consistency of switching, and the existence of a
canonical form for equivalence proofs, can be derived from a
characterisation of switch equivalence in terms of consistent
permutations or processes.

When dealing with left-linear rules, such a characterisation
fails. Leaving aside the formal details (the interested reader can
refer to \Cref{def:permcon}), the intuition is quite simple. In the
linear case one proves that two derivations using the same rules in
different orders are switch equivalent when the colimits of the two
derivations, seen as diagrams, are isomorphic and the isomorphism
properly commutes with the embeddings of the rules. When considering
left-linear derivations, this is no longer true. Intuitively this is
due to the fact that a rewriting step can merge parts of an
object thus making the colimit little informative.
This effect can be seen for derivations $\der{D}$ in
Fig.~\ref{fi:derD} and $\der{D}'$ in Fig.~\ref{fi:derD1}. They are not
abstraction equivalent but the colimit of both derivations is a single node with two self-loops, in a way that  all commutativity requirements are
trivially satisfied.  

In this section we show that, despite the failure of such
characterisation, actually much of the theory can be retained, taking
a different and sensibly more complex route for the proofs: globality
of independence holds in a weak form, conceptually linked to the fact
that fusions introduce a form of disjunctive causality, while
consistency of switching holds unchanged. Finally, a canonical form
for switching sequences can also be recovered.

We start by proving two lemmas dealing with derivations of length
$3$. Despite being technical, these lemmas are fundamental building blocks for obtaining our results.

\smallskip

\noindent
\begin{minipage}{.8\textwidth}
  \setlength\parindent{1.5em}
  The first lemma is at the core of globality for independence. It shows that if we have a derivation consisting of three rewriting steps, say $\dder{D}_0\cdot \dder{D}_1 \cdot \dder{D}_2$ where $\dder{D}_0$ and $\dder{D}_1$ are independent, then if we switch $\dder{D}_2$ to the first position, 
obtaining a derivation $\dder{D}_2'' \cdot  \dder{D}_0'\cdot\dder{D}_1'$, as in the diagram on the right where vertical lines represent permutations, the steps $\dder{D}_0'$ and $\dder{D}_1'$ are still independent.
\end{minipage}
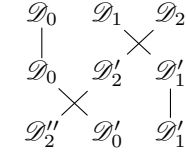
\begin{minipage}{.18\textwidth}
\ \hspace{-14mm}
\begin{tikzpicture}[node/.style={draw, minimum size=8mm}]
  \begin{tikzcd}[row sep=2mm, column sep=1mm]
    \dder{D}_0 \arrow[perm,d]  & \dder{D}_1 \arrow[perm,dr] & \dder{D}_2 \arrow[perm,dl]\\
    \dder{D}_0 \arrow[perm,dr]  & \dder{D}_2' \arrow[perm,dl] & \dder{D}_1' \arrow[perm,d]\\
    \dder{D}_2''   & \dder{D}_0' & \dder{D}_1'
  \end{tikzcd}
\end{tikzpicture}
\end{minipage}

\smallskip

\noindent
\begin{minipage}{.8\textwidth}
  If instead $\dder{D}_1$ and $\dder{D}_2$ are
independent in $\dder{D}_0\cdot \dder{D}_1 \cdot \dder{D}_2$ and we switch $\dder{D}_0$ to the last position, obtaining a
derivation $\dder{D}_1'\cdot \dder{D}_2'\cdot  \dder{D}_0''$, as in the diagram on the right, the steps $\dder{D}_1'$ and $\dder{D}_2'$ could be no
longer independent. 
Intuitively, this is due to the fact that dealing with
left-linear rewriting systems introduces disjunctive forms of
causality. Hence, if in $\dder{D}_0\cdot \dder{D}_1\cdot \dder{D}_2$ we have that $\dder{D}_0$ and $\dder{D}_1$ are
disjunctive causes of $\dder{D}_2$, at least one of the two is needed to
enable the application of $\dder{D}_2$ and this explain why in
$\dder{D}_1'\cdot \dder{D}_2'\cdot \dder{D}_0''$ we have that $\dder{D}_1'$ and $\dder{D}_2'$ are no longer
independent. This is exemplified below.
\end{minipage}
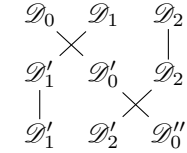
\begin{minipage}{.18\textwidth}
\ \hspace{-14mm}
\begin{tikzpicture}[node/.style={draw, minimum size=8mm}]
     \begin{tikzcd}[row sep=2mm, column sep=1mm]
      \dder{D}_0 \arrow[perm,dr]  & \dder{D}_1 \arrow[perm,dl] & \dder{D}_2 \arrow[perm,d]\\
      \dder{D}_1' \arrow[perm,d]  & \dder{D}_0' \arrow[perm,dr] & \dder{D}_2 \arrow[perm,dl]\\
      \dder{D}_1'   & \dder{D}_2' & \dder{D}_0''
    \end{tikzcd}
  \end{tikzpicture} 
\end{minipage}

\begin{example}
  \label{ex:non-global}
  Consider the rewriting system in $\cat{Graph}$ consisting of the
  rules $\lambda_0$, $\lambda_1$, and $\lambda_2$ in
  Fig.~\ref{fi:rules-cas}. Consider the derivation $\der{F}$ in Fig.~\ref{fi:derF}. Note that rule
  $\lambda_2$ needs the black self-loop on node $12$ to be applied. This can be
  generated either by $\lambda_0$ or by $\lambda_1$ merging nodes $1$ and
  $2$, hence at least one of them must precede the application of
  $\lambda_2$. Indeed, in $\der{F}$ it is easily seen that the second
  and the third step applying $\lambda_1$ and $\lambda_2$ are
  independent. However, we could switch twice the application of
  $\lambda_0$, bringing it in the last position, obtaining a derivation
  $\der{F}'$ as depicted in Fig.~\ref{fi:derF1}, where the
  applications of $\lambda_1$ and $\lambda_2$ are no longer independent.
  \begin{figure}[t]
    \begin{center}
    %
    \begin{tikzpicture}[node distance=2mm, font=\small]
      \node (l) {
      \begin{tikzpicture}
        \node at (0,0.53) {}; 
        \node at (0,0) [node, label=below:$1$] (1) {};
        \node at (0.5,0) [node, label=below:$2$] (2) {} ;
      \pgfBox
      \end{tikzpicture} 
    };
    \node[font=\scriptsize, above] at (l.north) {$L_0$};
    \node [right=of l] (k) {
      \begin{tikzpicture}
        \node at (0,0.53) {}; 
        \node at (0,0) [node, label=below:$1$] (1) {};
        \node at (0.5,0) [node, label=below:$2$] (2) {} ;
      \pgfBox
      \end{tikzpicture} 
    };
    \node[font=\scriptsize, above] at (k.north) {$K_0$};
    \node[below] at (k.south) {$\lambda_0$};
    \node  [right=of k] (r) {
      \begin{tikzpicture}
        \node at (0,0.53) {}; 
        \node at (0,0) [node, label=below:$12$] (12) {};
        \pgfBox
      \end{tikzpicture}
    };
    \node[font=\scriptsize, above] at (r.north) {$R_0$};
    \path (k) edge[->] node[trans, above] {} (l);
    \path (k) edge[->] node[trans, above] {} (r);
    \end{tikzpicture}
  %
  \hspace{1cm} 
  %
  %
  \begin{tikzpicture}[node distance=2mm, font=\small]
      \node (l) {
      \begin{tikzpicture}
        \node at (0,0.53) {}; 
        \node at (0,0) [node, label=below:$1$] (1) {};
        \node at (0.5,0) [node, label=below:$2$] (2) {} ;
        \draw[coloredge] (1) to[out=20, in=160] (2);
      \pgfBox
      \end{tikzpicture} 
    };
    \node[font=\scriptsize, above] at (l.north) {$L_1$};
    \node [right=of l] (k) {
      \begin{tikzpicture}
        \node at (0,0.53) {}; 
        \node at (0,0) [node, label=below:$1$] (1) {};
        \node at (0.5,0) [node, label=below:$2$] (2) {} ;
        \draw[coloredge] (1) to[out=20, in=160] (2);                
      \pgfBox
      \end{tikzpicture} 
    };
    \node[font=\scriptsize, above] at (k.north) {$K_1$};
    \node[below] at (k.south) {$\lambda_1$};
    \node  [right=of k] (r) {
      \begin{tikzpicture}
        \node at (0,0.53) {}; 
        \node at (0,0) [node, label=below:$12$] (12) {}
        edge [in=65, out=35, colorloop] ();
        \pgfBox
      \end{tikzpicture}
    };
    \node[font=\scriptsize, above] at (r.north) {$R_1$};
    \path (k) edge[->] node[trans, above] {} (l);
    \path (k) edge[->] node[trans, above] {} (r);
    \end{tikzpicture}
  %
  \hspace{1cm} 
  %
    %
    \begin{tikzpicture}[node distance=2mm, font=\small]
      \node (l) {
      \begin{tikzpicture}
        \node at (0,0.53) {}; 
        \node at (0,0) [node, label=below:$1$] (1) {} 
        edge [in=145, out=115, loop] ();
      \pgfBox
      \end{tikzpicture} 
    };
    \node[font=\scriptsize, above] at (l.north) {$L_2$};
    \node [right=of l] (k) {
      \begin{tikzpicture}
        \node at (0,0.53) {}; 
        \node at (0,0) [node, label=below:$1$] (1) {};
        \pgfBox
      \end{tikzpicture} 
    };
    \node[font=\scriptsize, above] at (k.north) {$K_2$};
    \node[below] at (k.south) {$\lambda_2$};
    \node  [right=of k] (r) {
      \begin{tikzpicture}
        \node at (0,0.53) {}; 
        \node at (0,0) [node, label=below:$1$] (1) {};        
        \pgfBox
      \end{tikzpicture}
    };
    \node[font=\scriptsize, above] at (r.north) {$R_2$};
    \path (k) edge[->] node[trans, above] {} (l);
    \path (k) edge[->] node[trans, above] {} (r);
    \end{tikzpicture}
  %
  %
\end{center}    
    \caption{A new rewriting system in $\cat{Graph}$.}
    \label{fi:rules-cas}
  \end{figure}
  \begin{figure}[t]
      \begin{center}
    \begin{tikzpicture}[node distance=2mm, font=\small, baseline=(current bounding box.center)]      
      \node (L0) at (0,2) {
        \begin{tikzpicture}
          \node at (0,0.53) {}; 
          \node at (0,0) [node, label=below:$1$] (1) {};
          \node at (0.5,0) [node, label=below:$2$] (2) {} ;          
          \pgfBox
        \end{tikzpicture} 
      };
      \node [right=of L0] (K0)  {
        \begin{tikzpicture}
          \node at (0,0.53) {}; 
          \node at (0,0) [node, label=below:$1$] (1) {};
          \node at (0.5,0) [node, label=below:$2$] (2) {} ;
          \pgfBox
        \end{tikzpicture} 
      };
      \node [above] at (K0.north) {$\lambda_0$};
      \node [right=of K0](R0) {
        \begin{tikzpicture}
          \node at (0,0.53) {}; 
          \node at (0,0) [node, label=below:$12$] (12) {};
          \pgfBox
        \end{tikzpicture}
      };
      \path (K0) edge[->] node[trans, above] {} (L0);
      \path (K0) edge[->] node[trans, above] {} (R0);

      \node at (4.5,2) (L1) {
        \begin{tikzpicture}
          \node at (0,0.53) {}; 
          \node at (0,0) [node, label=below:$12$] (1) {} ;
          \node at (0.5,0) [node, label=below:$12$] (2) {} ;
          \draw[coloredge] (1) to[out=20, in=160] (2);
          \pgfBox
        \end{tikzpicture} 
      };
      \node [right=of L1] (K1) {
        \begin{tikzpicture}
          \node at (0,0.53) {}; 
          \node at (0,0) [node, label=below:$12$] (1) {} ;
          \node at (0.5,0) [node, label=below:$12$] (2) {} ;
          \draw[coloredge] (1) to[out=20, in=160] (2);
          \pgfBox
        \end{tikzpicture} 
      };
      \node [above] at (K1.north) {$\lambda_1$};
      \node [right=of K1] (R1) {
        \begin{tikzpicture}
          \node at (0,0.53) {}; 
          \node at (0,0) [node, label=below:$12$] (12) {}
          edge [in=65, out=35, colorloop] ();
          \pgfBox
        \end{tikzpicture}
      };
      \path (K1) edge[->] node[trans, above] {} (L1);
      \path (K1) edge[->] node[trans, above] {} (R1);

      \node at (9,2) (L2) {
        \begin{tikzpicture}
          \node at (0,0) [node, label=below:$12$] (2) {}
          edge [in=145, out=115, loop] ();        
          \pgfBox
        \end{tikzpicture} 
      };
      \node [right=of L2] (K2) {
        \begin{tikzpicture}
          \node at (0,0.53) {}; 
          \node at (0,0) [node, label=below:$12$] (12) {};
          \pgfBox
        \end{tikzpicture} 
      };
      \node [above] at (K2.north) {$\lambda_2$};
      \node [right=of K2] (R2) {
        \begin{tikzpicture}
          \node at (0,0.53) {}; 
          \node at (0,0) [node, label=below:$12$] (12) {} ;
          \pgfBox
        \end{tikzpicture}
      };
      \path (K2) edge[->] node[trans, above] {} (L2);
      \path (K2) edge[->] node[trans, above] {} (R2);

      \node at (0,0) (G0) {
        \begin{tikzpicture}
          \node at (0,0.53) {}; 
          \node at (0,0) [node, label=below:$1$] (1) {};
          \node at (0.5,0) [node, label=below:$2$] (2) {} ;
          \draw[coloredge] (1) to[out=30, in=150] (2);
          \draw[->] (2) to[out=210, in=330] (1);
          \pgfBox
        \end{tikzpicture} 
      };
      \node [right=of G0] (D0) {
        \begin{tikzpicture}
          \node at (0,0.53) {}; 
          \node at (0,0) [node, label=below:$1$] (1) {};
          \node at (0.5,0) [node, label=below:$2$] (2) {} ;
          \draw[coloredge] (1) to[out=30, in=150] (2);
          \draw[->] (2) to[out=210, in=330] (1);
          \pgfBox
        \end{tikzpicture} 
      };
      \node at (3.5,0) (G1) {
        \begin{tikzpicture}
          \node at (0,0.53) {}; 
          \node at (0,0) [node, label=below:$12$] (12) {}
          edge [in=65, out=35, colorloop] ()
          edge [in=145, out=115, loop] ();
          \pgfBox
        \end{tikzpicture} 
      };
      \path (D0) edge[->] node[trans, above] {} (G0);
      \path (D0) edge[->] node[trans, above] {} (G1);
      \path (L0) edge[->] node[trans, above] {} (G0);
      \path (K0) edge[->] node[trans, above] {} (D0);
      \path (R0) edge[->] node[trans, above] {} (G1);
      \node at (6.1,0) (D1) {
        \begin{tikzpicture}
          \node at (0,0.53) {}; 
          \node at (0,0) [node, label=below:$12$] (12) {}
          edge [in=65, out=35, colorloop] ()
          edge [in=145, out=115, loop] (); 
          \pgfBox
        \end{tikzpicture} 
      };
      \node at (8.5,0) (G2) {
        \begin{tikzpicture}
          \node at (0,0.53) {}; 
          \node at (0,0) [node, label=below:$12$] (12) {}
          edge [in=65, out=35, colorloop] ()
          edge [in=145, out=115, loop] ();
          \pgfBox
        \end{tikzpicture} 
      };

      \path (D1) edge[->] node[trans, above] {} (G1);
      \path (D1) edge[->] node[trans, above] {} (G2);
      \path (L1) edge[->] node[trans, above] {} (G1);
      \path (K1) edge[->] node[trans, above] {} (D1);
      \path (R1) edge[->] node[trans, above] {} (G2);
      
      \node at (10.2,0) (D2) {
        \begin{tikzpicture}
          \node at (0,0) [node, label=below:$12$] (12) {}
          edge [in=65, out=35, colorloop] ();
          \pgfBox
        \end{tikzpicture} 
      };
      \node [right=of D2] (G3) {
        \begin{tikzpicture}
          \node at (0,0.53) {};                    
          \node at (0,0) [node, label=below:$12$] (12) {}
          edge [in=65, out=35, colorloop] ();
          \pgfBox
        \end{tikzpicture} 
      };
      \node[font=\scriptsize, below] at (G0.south) {$G_0$};
      \node[font=\scriptsize, below] at (D0.south) {$D_0$};      
      \node[font=\scriptsize, below] at (G1.south) {$G_1$};
      \node[font=\scriptsize, below] at (D1.south) {$D_1$};      
      \node[font=\scriptsize, below] at (G2.south) {$G_2$};
      \node[font=\scriptsize, below] at (D2.south) {$D_2$};      
      \node[font=\scriptsize, below] at (G3.south) {$G_3$};      
      \path (D2) edge[->] node[trans, above] {} (G2);
      \path (D2) edge[->] node[trans, above] {} (G3);
      \path (L2) edge[->] node[trans, above] {} (G2);
      \path (K2) edge[->] node[trans, above] {} (D2);
      \path (R2) edge[->] node[trans, above] {} (G3);
    \end{tikzpicture}
\end{center}
    \caption{The derivation $\der{F}$.} 
    \label{fi:derF}
  \end{figure}
  \begin{figure}[t]
      \begin{center}
    \begin{tikzpicture}[node distance=2mm, font=\small, baseline=(current bounding box.center)]      
      \node (L1) at (0,2) {
        \begin{tikzpicture}
          \node at (0,0.53) {}; 
          \node at (0,0) [node, label=below:$12$] (1) {} ;
          \node at (0.5,0) [node, label=below:$12$] (2) {} ;
          \draw[coloredge] (1) to[out=20, in=160] (2);
          \pgfBox
        \end{tikzpicture} 
      };
      \node [right=of L1] (K1) {
        \begin{tikzpicture}
          \node at (0,0.53) {}; 
          \node at (0,0) [node, label=below:$12$] (1) {} ;
          \node at (0.5,0) [node, label=below:$12$] (2) {} ;
          \draw[coloredge] (1) to[out=20, in=160] (2);
          \pgfBox
        \end{tikzpicture}         
      };
      \node [above] at (K1.north) {$\lambda_1$};
      \node [right=of K1](R1) {
        \begin{tikzpicture}
          \node at (0,0.53) {}; 
          \node at (0,0) [node, label=below:$12$] (12) {}
          edge [in=65, out=35, colorloop] ();
          \pgfBox
        \end{tikzpicture}
      };
      \path (K1) edge[->] node[trans, above] {} (L1);
      \path (K1) edge[->] node[trans, above] {} (R1);

      \node at (5,2) (L2) {
        \begin{tikzpicture}
          \node at (0,0) [node, label=below:$12$] (2) {}
          edge [in=145, out=115, loop] ();        
          \pgfBox
        \end{tikzpicture}
      };
      \node [right=of L2] (K2) {
        \begin{tikzpicture}
          \node at (0,0.53) {}; 
          \node at (0,0) [node, label=below:$12$] (12) {};
          \pgfBox
        \end{tikzpicture} 
      };
      \node [above] at (K2.north) {$\lambda_2$};
      \node [right=of K2] (R2) {
        \begin{tikzpicture}
          \node at (0,0.53) {}; 
          \node at (0,0) [node, label=below:$12$] (12) {};
          \pgfBox
        \end{tikzpicture}
      };
      \path (K2) edge[->] node[trans, above] {} (L2);
      \path (K2) edge[->] node[trans, above] {} (R2);

      \node at (9,2) (L0) {
        \begin{tikzpicture}
          \node at (0,0.53) {}; 
          \node at (0,0) [node, label=below:$12$] (1) {};
          \node at (0.5,0) [node, label=below:$12$] (2) {} ;
          \pgfBox
        \end{tikzpicture} 
      };
      \node [right=of L0] (K0) {
        \begin{tikzpicture}
          \node at (0,0.53) {}; 
          \node at (0,0) [node, label=below:$12$] (1) {};
          \node at (0.5,0) [node, label=below:$12$] (2) {} ;
          \pgfBox
        \end{tikzpicture} 
      };
      \node [above] at (K0.north) {$\lambda_0$};
      \node [right=of K0] (R0) {
        \begin{tikzpicture}
          \node at (0,0.53) {}; 
          \node at (0,0) [node, label=below:$12$] (12) {} ;
          \pgfBox
        \end{tikzpicture}
      };
      \path (K0) edge[->] node[trans, above] {} (L0);
      \path (K0) edge[->] node[trans, above] {} (R0);

      \node at (0,0) (G0) {
        \begin{tikzpicture}
          \node at (0,0.53) {}; 
          \node at (0,0) [node, label=below:$1$] (1) {};
          \node at (0.5,0) [node, label=below:$2$] (2) {} ;
          \draw[coloredge] (1) to[out=30, in=150] (2);
          \draw[->] (2) to[out=210, in=330] (1);
          \pgfBox
        \end{tikzpicture}  
      };
      \node at (1.6,0) (D0) {
         \begin{tikzpicture}
          \node at (0,0.53) {}; 
          \node at (0,0) [node, label=below:$1$] (1) {};
          \node at (0.5,0) [node, label=below:$2$] (2) {} ;
          \draw[coloredge] (1) to[out=30, in=150] (2);
          \draw[->] (2) to[out=210, in=330] (1);
          \pgfBox
        \end{tikzpicture}  
      };
      \node at (4,0) (G1) {
        \begin{tikzpicture}
          \node at (0,0.53) {}; 
          \node at (0,0) [node, label=below:$12$] (12) {}
          edge [in=65, out=35, colorloop] ()
          edge [in=145, out=115, loop] (); 
          \pgfBox
        \end{tikzpicture} 
      };
      \path (D0) edge[->] node[trans, above] {} (G0);
      \path (D0) edge[->] node[trans, above] {} (G1);
      \path (L1) edge[->] node[trans, above] {} (G0);
      \path (K1) edge[->] node[trans, above] {} (D0);
      \path (R1) edge[->] node[trans, above] {} (G1);

      \node at (6.2,0) (D1) {
        \begin{tikzpicture}
          \node at (0,0.53) {};                    
          \node at (0,0) [node, label=below:$12$] (12) {}
          edge [in=65, out=35, colorloop] ();
          \pgfBox
        \end{tikzpicture}
      };
      \node at (8.2,0) (G2) {
        \begin{tikzpicture}
          \node at (0,0.53) {};                    
          \node at (0,0) [node, label=below:$12$] (12) {}
          edge [in=65, out=35, colorloop] ();
          \pgfBox
        \end{tikzpicture}
      };

      \path (D1) edge[->] node[trans, above] {} (G1);
      \path (D1) edge[->] node[trans, above] {} (G2);
      \path (L2) edge[->] node[trans, above] {} (G1);
      \path (K2) edge[->] node[trans, above] {} (D1);
      \path (R2) edge[->] node[trans, above] {} (G2);
      
      \node at (10.6,0) (D2) {
        \begin{tikzpicture}
          \node at (0,0.53) {};                    
          \node at (0,0) [node, label=below:$12$] (12) {}
          edge [in=65, out=35, colorloop] ();
          \pgfBox
        \end{tikzpicture}
      };
      \node [right=of D2] (G3) {
        \begin{tikzpicture}
          \node at (0,0.53) {};                    
          \node at (0,0) [node, label=below:$12$] (12) {}
          edge [in=65, out=35, colorloop] ();
          \pgfBox
        \end{tikzpicture}
      };
      \node[font=\scriptsize, below] at (G0.south) {$G_0$};
      \node[font=\scriptsize, below] at (D0.south) {$D_0'$};      
      \node[font=\scriptsize, below] at (G1.south) {$G_1'$};
      \node[font=\scriptsize, below] at (D1.south) {$D_1'$};      
      \node[font=\scriptsize, below] at (G2.south) {$G_2'$};
      \node[font=\scriptsize, below] at (D2.south) {$D_2'$};      
      \node[font=\scriptsize, below] at (G3.south) {$G_3$};      
      \path (D2) edge[->] node[trans, above] {} (G2);
      \path (D2) edge[->] node[trans, above] {} (G3);
      \path (L0) edge[->] node[trans, above] {} (G2);
      \path (K0) edge[->] node[trans, above] {} (D2);
      \path (R0) edge[->] node[trans, above] {} (G3);
    \end{tikzpicture}
\end{center}
    \caption{The derivation $\der{F}'$.} 
    \label{fi:derF1}
  \end{figure}
\end{example}
\begin{restatable}[Globality of independence]{lemma}{lemIndepGlobalLeft}
  \label{lem:indep-global-left}
  Let $\X$ be an $\mathcal{M}$-adhesive category and $(\X,\R)$ a
  well-switching rewriting system.  Let
  $\dder{D}_0\cdot \dder{D}_1 \cdot \dder{D}_2$ be a three-steps
  derivation.
  \begin{enumerate}
  \item
    \label{lem:indep-global-left:1}
    If $\dder{D}_0$ and $\dder{D}_1$ are switchable and there is a
    switching sequence
    $\dder{D}_0\cdot \dder{D}_1 \cdot \dder{D}_2 \shift{(1,2)}
    \dder{D}_0 \cdot \dder{D}_2' \cdot \dder{D}_1' \shift{(0,1)}
    \dder{D}_2'' \cdot \dder{D}_0' \cdot \dder{D}_1'$ then
    in the last derivation $\dder{D}_0'$ and $\dder{D}_1'$ are switchable;
    
  \item
    \label{lem:indep-global-left:2}
    Suppose that $\dder{D}_1$ and $\dder{D}_2$ are switchable, hence $\dder{D}_0\cdot \dder{D}_1 \cdot \dder{D}_2 \shift{(1,2)}
    \dder{D}_0 \cdot \dder{D}_2' \cdot \dder{D}_1'$ and $\dder{D}_0$ and $\dder{D}_2'$ are switchable.  If there
    is a switching sequence
    $\dder{D}_0\cdot \dder{D}_1 \cdot \dder{D}_2 \shift{(0,1)}
    \dder{D}_1' \cdot \dder{D}_0' \cdot \dder{D}_2 \shift{(1,2)}
    \dder{D}_1' \cdot \dder{D}_2' \cdot \dder{D}_0''$ then in the last derivation $\dder{D}_1'$ and $\dder{D}_2'$ are switchable.
  \end{enumerate}
\end{restatable}
  
The second lemma captures the essence of the consistency of
switchings. If in a derivation of three
rewriting steps, say $\dder{D}_0\cdot \dder{D}_1\cdot  \dder{D}_2$, we invert them leading to a
sequence reordered as $\dder{D}_2\cdot  \dder{D}_1\cdot  \dder{D}_0$, we obtain the same derivation,
independently from the order of switchings.

\begin{restatable}[Consistency of switchings]{lemma}{lemSwitchConfluence}
	\label{lem:switch-confluence}
	Let $\X$ be an $\mathcal{M}$-adhesive category and $(\X,\R)$ a well-switching rewriting system.
	Consider a derivation $\der{D}=\{\dder{D}_i\}_{i=0}^2$ and suppose that we have the following two switching sequences
	\[\der{D} \shift{(0,1)} \der{D}' \shift{(1,2)} \der{D}'' \shift{(0,1)}\der{D}'''\\
          \qquad \der{D}\shift{(1,2)}\der{E}'\shift{(0,1)} \der{E}'' \shift{(1,2)}\der{E}''' \]
Then $\der{D}'''$ and $\der{E}'''$ are abstraction equivalent. 
\end{restatable}

The lemmas above open the way to the key result of this section, which establishes a canonical form for
switching derivations.

The first result shows that when equating two derivation sequences by performing a series of switchings, we can limit ourselves to switchings that invert the order of steps that in the target derivation has to be reversed, i.e.~we can limit ourselves to applying inversions.

\begin{restatable}[No need of useless switches]{theorem}{thNoNeed}
  \label{th:no-need}
  Let $\der{D}$ and $\der{D}'$ be derivation sequences. If $\der{D} \shifteq \der{D}'$ then there is a switching sequence 
    $\der{D} \equiv_a \der{D}_0 \shift{\nu_1} \der{D}_1 \shift{\nu_2} \ldots
    \shift{\nu_n} \der{D}_n \equiv_a \der{D}'$
  consisting only of inversions.
\end{restatable}

\full{[Proof~\ref{thNoNeed-proof}]}

While inversions can be
performed in any order in the case of linear rewriting systems, this is no longer true for left-linear
rewriting systems due to the fact that, as already observed, globality of independence holds in a weak form. 
For instance, if we get back to Example~\ref{ex:non-global} and consider derivation $\der{F}$ in Fig.~\ref{fi:derF}, 
the last two steps applying rules $\lambda_1$ and $\lambda_2$ can be switched, thus leading to a derivation, 
call it $\der{F}''$,  applying $\lambda_0, \lambda_2, \lambda_1$. This $\der{F}''$ is switch equivalent to $\der{F}'$ in Fig.~\ref{fi:derF1}, 
hence starting from $\der{F}''$ we can obtain $\der{F}'$ via a switching sequence, but we cannot start 
by switching $\lambda_0$ and $\lambda_2$, even if this is an inversion.

\newpage
Still, we can identify a
canonical form for switching derivations: at each step one can apply a switch to the inversion with largest index. 

\begin{restatable}[Canonical form]{corollary}{coCanonical}
  \label{co:canonical}
  Let $\der{D}$ and $\der{D}'$ be derivation sequences. If
  $\der{D} \shifteq[\sigma]\der{D}'$ then there is a switching sequence
  \begin{center}
    $\der{D} \equiv_a \der{D}_0 \shift{\nu_1} \der{D}_1 \shift{\nu_2} \ldots
    \shift{\nu_n} \der{D}_n \equiv_a\der{D}'$
  \end{center}
  where for $i \in \interval{n}$ we have $\nu_i  = \transp{k}$ with 
  $k = \max \{ j \mid \transp{j} \in \inv{\nu_{i,n}}\}$.
\end{restatable}

\section{Conclusions and further work}
\label{conclusions}

We performed an in-depth investigation of the notion of independence between
rewriting steps in the setting of left-linear rewriting systems
over $\mathcal{M}$-adhesive categories, which encompasses most of the
structures commonly used in analysing and modelling applications.

We showed that the canonical notion of independence adopted for linear
systems does not enjoy some properties that are essential for
developing a sensible theory of concurrency, notably a Church-Rosser
theorem, and we identified a subclass of left-linear rewriting systems,
which we call well-switching, as an appropriate setting where 
many key results can be re-established. Specifically, a Church-Rosser 
theorem, i.e.~the possibility of performing independent steps in any order, is
fully recovered. Moreover, the switch construction and correspondingly
the switch equivalence, at the core of a theory of concurrency for
rewriting systems, enjoy a number of fundamental properties: a weak form of globality of independence, 
consistency of switching and
the existence of a canonical form for switch equivalence proofs.

The class of well-switching $\mathcal{M}$-adhesive rewriting systems is large. Generalising a result in~\cite{baldan2011adhesivity},
we showed that one of its defining properties (the fact that sequential independence implies switchability) holds 
for a general class of $\mathcal{M}$-adhesive categories that include all presheaves over $\cat{Set}$. Moreover, the second property, uniqueness 
of switching, is ensured when rewriting possibly hierarchical graphical structures, as long as elements belonging to the sorts 
of the roots  are not merged (e.g., one can consider rewriting systems over directed graphs where only nodes are merged, as it 
happens in most of the approaches to the modelling of calculi and biological systems that we cited in the introduction, 
as well as in string diagrams).

In the DPO approach to rewriting, the notion of sequential
independence, identifying consecutive rewriting steps that 
can be performed in reverse order is typically
closely linked to that of \emph{parallel
	independence}~\cite{ehrig2014adhesive}. Parallel independence
relates two rewriting steps starting from
the same object, which can be applied in
any order producing, up to isomorphism, the same resulting
object. Parallel independent steps, when applied in sequence, lead to
sequential independent steps and, conversely, from two sequential
independent steps, as a byproduct of the switch construction, one can
get two parallel independent steps.  

Under mild additional conditions, one can have a notion of \emph{parallel rule} in a way that two
parallel independent steps can be also combined in a single rewriting step.
Even if this is not discussed explicitly in the paper, our theory can be easily accommodated to encompass a notion of parallel independence and parallel rule application, enjoying the properties outlined above. To ensure this, however, one must require that the underlying category $\X$ has binary coproducts and that the class $\mathcal{M}$ is closed under them.

The results in this paper open the way to the development of a
concurrent semantics for left-linear rewriting systems over $\mathcal{M}$-adhesive
categories. In~\cite{baldan2017domains} the authors argue that for
computational systems that allow one to express merging or fusions,
the right semantical models are weak prime domains, a generalisation
of prime domains, a staple in concurrency theory, and connected event
structures, their event-based counterpart. The mentioned paper
discusses only the case of graph rewriting and does not focus into the
problems induced by the occurring of ``merges'' to the theory of rewriting (the systems
considered there are implicitly assumed to be well-switching). 
We plan to consolidate the statement that weak prime domains are 
``the'' model for systems with fusions by showing that they allow to provide a semantics
for left-linear rewriting systems in $\mathcal{M}$-adhesive categories.

On the practical side, the present paper just surveyed the possibility
of modelling e-graphs. This appears to be quite interesting as
concurrent rewriting on such structures is an active area of
research~\cite{abs-2208-06295}.
We plan to make the correspondence precise
and investigate how the analysis techniques enabled by a concurrent
semantics of rewriting can impact on e-graphs.

\bibstyle{plainurl}
\bibliography{bibliog.bib}

\appendix
\section{Properties of $\mathcal{M}$-adhesive categories}\label{app:ade}

This first appendix is devoted to the proofs of a few well-known results about $\mathcal{M}$-adhesive categories.
Observe that our notion of $\mathcal{M}$-adhesivity follows
\cite{ehrig2012,ehrig2014adhesive} and is different from the one of
\cite{azzi2019essence}. What is called $\mathcal{M}$-adhesivity in
the latter paper corresponds to our strict
$\mathcal{M}$-adhesivity. Moreover, in \cite{azzi2019essence} the
class $\mathcal{M}$ is assumed to be only stable under
pullbacks. However, if $\mathcal{M}$ contains all split monos, then
stability under pushouts can be deduced from the other
axioms~\cite[Prop.~$5.1.21$]{castelnovo2023thesis}.

\subsection{Some results on $\mathcal{A}$-stable and $\mathcal{A}$-Van Kampen squares}
We start proving some general results regarding $\mathcal{A}$-Van Kampen and $\mathcal{A}$-stable squares. Let us begin recalling some classical results about  about pullbacks and pushouts\cite{mac2013categories}.

\noindent 
\parbox{10.5cm}{\begin{lemma}\label{lem:popb1} \label{lem:pb1} \label{lem:po1}
Let $\X$ be a category, and consider the diagram  on the right, then the following hold
	\begin{enumerate}
		\item if the right square is a pullback, then the whole rectangle is a pullback if and only if the left square is one;
		\item if the left square is a pushout, 	then the whole rectangle is a pushout if and only if the right square is one.
	\end{enumerate}
\end{lemma}}
\parbox{4cm}{$\xymatrix{X \ar[d]_{a} \ar[r]^{f}& \ar[r]^{g} Y \ar[d]^{b}& Z \ar[d]^{c}\\ A \ar[r]_{h}& B \ar[r]_{k}& C}$}\\

The following proposition establishes a key property of $\mathcal{A}$-Van Kampen squares with a mono as a side: they are not only pushouts, but also pullbacks \cite{ehrig2004adhesive,BehrHK23,lack2005adhesive}.

\noindent 
\parbox{11.5cm}{\begin{proposition}\label{prop:pbpo} Let $\mathcal{A}$ be a class of arrows stable under pushouts and containing all the isomorphisms.  If the square to the right is $\mathcal{A}$-Van Kampen and $m\colon A\to C$ is mono and belongs to $\mathcal{A}$, then the square is a pullback and $n$ is a monomorphism.
\end{proposition}}
\parbox{2cm}{\vspace{1.5ex}$\xymatrix{A\ar[r]^{g} \ar[d]_{m} & B \ar[d]^{n} \\ C \ar[r]_{f}  & D}$}

The previous proposition allows us to establish the following result.

\noindent \parbox{7.4cm}{
\begin{lemma}\label{lem:varie}Let $\mathcal{A}$ be a class of arrows stable under pullbacks, pushouts and containing all isomorphisms.  Suppose that the left square aside is $\mathcal{A}$-Van Kampen, while the vertical faces in the right cube are pullbacks. 	Suppose moreover that $m\colon A\to C$ and $d\colon D'\to D$ are mono and that $d$ belongs to $\mathcal{A}$. Then $d\leq n$ if and only if $c \leq m$.
\end{lemma}}
	\parbox{6cm}{$
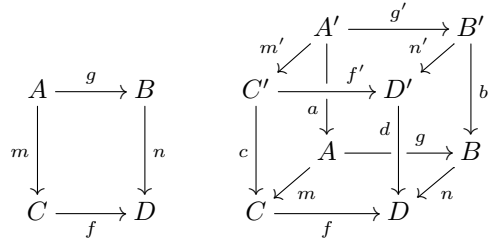
\xymatrix@C=10pt@R=10pt{&&&&&A'\ar[dd]|\hole_(.65){a}\ar[rr]^{g'} \ar[dl]_{m'} && B' \ar[dd]^{b} \ar[dl]_{n'} \\ A \ar[dd]_{m}\ar[rr]^{g}&&B\ar[dd]^{n}&&C'  \ar[dd]_{c}\ar[rr]^(.7){f'} & & D' \ar[dd]_(.3){d}\\&&&&&A\ar[rr]|\hole^(.65){g} \ar[dl]^{m} && B \ar[dl]^{n} \\C\ar[rr]_{f} &&D&&C \ar[rr]_{f} & & D}$}

\noindent 
\parbox{11cm}{\begin{remark}
	Recall that, given two monos $m:M\to X$ and $n:N\to X$ with the same codomain, $m\leq n$ means that there exists a, necessarily unique and necessarily mono, $k:M\to N$ fitting in the triangle aside.	
	
	\hspace{15pt}Notice, moreover, that  if $m\leq n$ and $n\leq m$, then the arrow $k:M\to N$ is an isomorphism.
\end{remark}}
\parbox{4cm}{$\xymatrix@C=15pt{M\ar@{.>}[rr]^{k}  \ar[dr]_{m}&& N \ar[dl]^{n}\\ & X}$}\\

Finally, we show that $\mathcal{A}$-stable pushouts enjoy a \emph{pullback-pushout decomposition} property.

\noindent 
\parbox{10cm}{
\begin{proposition}\label{prop:stab}Let $\X$ be a category and $\mathcal{A}$ a class of arrows stable under pullbacks. Suppose that, in the diagram aside, the whole rectangle is an $\mathcal{A}$-stable pushout and the right square a pullback. If the arrow $k$ is in $\mathcal{A}$ and it is a monomorphism,  then both squares are pushouts.
\end{proposition}}
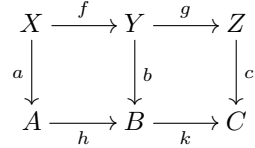
\parbox{4cm}{\vspace{-.3cm}
	$\xymatrix{X \ar[d]_{a} \ar[r]^{f}& \ar[r]^{g} Y \ar[d]^{b}& Z \ar[d]^{c}\\ A \ar[r]_{h}& B \ar[r]_{k}& C}$}

\subsection{Useful properties of $\mathcal{M}$-adhesive categories}

We are now going to apply the results of the previous section to $\mathcal{M}$-adhesive categories in order to establish some \emph{high-level replacement properties} \cite{ehrig2004adhesive,ehrig2014adhesive,ehrig2006fundamentals}.  
A first important  result that can be immediately established, with the aid of \Cref{prop:pbpo}, is the following one.

\begin{proposition}
	\label{prop:pbpoad}
	Let $\X$ be an $\mathcal{M}$-adhesive category. Then
	$\mathcal{M}$-pushouts are pullbacks.
\end{proposition}

From \Cref{prop:pbpoad}, in turn, we can derive the following corollaries.
\begin{corollary}\label{cor:rego}
	In a $\mathcal{M}$-adhesive category $\X$, every $m\in\mathcal{M}$ is a regular mono.
\end{corollary}

The following result now follows at once noticing that a regular monomorphism which is also epic is automatically an isomorphism.

\begin{corollary}\label{prop:bal}
	If $\X$ is an $\mathcal{M}$-adhesive categories, then every epimorphism in $\mathcal{M}$ is an isomorphism. In particular, every adhesive category $\X$ is \emph{balanced}: if a morphism is monic and epic, then it is an isomorphism.
\end{corollary}

\noindent 
\parbox{10cm}{\begin{lemma}[$\mathcal{M}$-pushout-pullback decomposition]\label{lem:popb} Let $\X$ be an $\mathcal{M}$-adhesive category  and suppose that, in the diagram aside, the whole rectangle is a pushout and the right square a pullback. Then the following statements hold
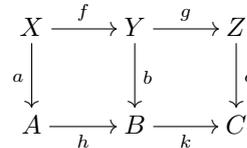
\parbox{13cm}	{\begin{enumerate}
		\item if $a$ belongs to $\mathcal{M}$ and $k$ is a mono,  then both squares are pushouts and pullbacks;
		\item if $f$ and $k $ are in  $\mathcal{M}$, then both squares are pushouts and pullbacks.
	\end{enumerate}}
\end{lemma}}
\parbox{4cm}{\vspace{-1cm}
	$\xymatrix{X \ar[d]_{a} \ar[r]^{f}& \ar[r]^{g} Y \ar[d]^{b}& Z \ar[d]^{c}\\ A \ar[r]_{h}& B \ar[r]_{k}& C}$}\\

Let us turn our attention to pushout complements.

\noindent
\parbox{11cm}{
	\begin{definition}[Pushout complement]
		Let $f\colon X\to Y$ and $g\colon Y\to Z$ be two composable arrows in a category $\X$. A \emph{pushout complement} for the pair $(f,g)$ is a pair $(h,k)$ with $h\colon X\to W$ and $k\colon W\to Z$ such that the square aside is a pushout.
	\end{definition}} \quad
\parbox{4cm}{\vspace{.1cm}
	$
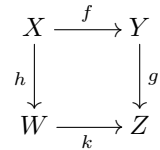
\xymatrix{X \ar[r]^{f} \ar[d]_{h}& Y \ar[d]^{g} \\ W \ar[r]_{k}& Z}$}\\

Working in an $\mathcal{M}$-adhesive category we can guarantee that pushout complements are unique. 

\noindent
\parbox{8.5cm}{
	\begin{lemma}\label{lem:radj}
		Let $f\colon X \to Y$ be an arrow in an $\mathcal{M}$-adhesive category $\X$ and suppose that the left square aside is a pushout while the  right one is a pullback, with $ m\colon M \rightarrowtail X$ and $n\colon N \rightarrowtail Y$ in $\mathcal{M}$.
		Then $n\leq k$ if and only if $p_2\leq m$.
	\end{lemma}}
\parbox{3cm}{\vspace{-2ex}
	$\xymatrix{M \ar@{>->}[d]_{m}\ar[r]^{h}& Q \ar@{>->}[d]^{k} & P \ar[r]^{p_1} \ar@{>->}[d]_{p_2}& N \ar@{>->}[d]^{n}\\X\ar[r]_{f}&Y &X \ar[r]_f&Y}$}

\noindent
\parbox{10cm}{
	\begin{corollary}[Uniqueness of pushout complements]\label{lem:pocomp}
		Let $\X$ be a $\mathcal{M}$-adhesive category. Given $m\colon X\mto Y$ in $\mathcal{M}$ and $f\colon Y\to Z$, let $(h_1, k_1)$ and $(h_2, k_2)$ be pushout complements of $m$  and $f$ depicted aside. Then there exists a unique isomorphism $\phi\colon W_1\to W_2$ making the diagram on the right commutative.
	\end{corollary}}
\parbox{3cm}{\vspace{1em}$\xymatrix{&X \ar@{>->}[r]^{m} \ar[d]_{h_1} \ar@/_.3cm/[dl]_{h_2}& Y \ar[d]^{f} \\ W_2 \ar@{>->}@/_.4cm/[rr]_{k_2} &W_1 \ar@{.>}[l]_{\phi} \ar@{>->}[r]^{k_1} & Z }$}

\section{$\mathcal{M}$-adhesivity is not enough}\label{app:fill}
In \Cref{sec:equi} we introduced the notion of strong enforcing left-linear rewriting system, a class that contains all the linear rewriting systems. This result can be further refined: in \cite{baldan2011adhesivity} a class  $\mathbb{B}$ of (quasi-)adhesive category is defined for which the local Church-Rosser Theorem holds for left-linear rewriting system. In our language, this means that every left-linear rewriting system based on a category in $\mathbb{B}$ is strong enforcing. This section is devoted to repropose, and slightly generalise, the results of \cite{baldan2011adhesivity} to our context.

\noindent
\parbox{10.4cm}{
	\begin{definition}Let $\mathcal{M}$ be a class of monos in a category $\X$, closed under composition, containing all isomorphisms and stable under pullbacks and pushouts. Suppose that the  diagram aside is given and the whole rectangle is a pushout. We say that $\X$ satisfies
		\parbox{14cm}{\begin{itemize}
				\item the \emph{$\mathcal{M}$-mixed decomposition} property if, whenever $k$ belongs to $\mathcal{M}$ and the right half of the previous diagram is a pullback, then the left one is a pushout;
				\item the \emph{$\mathcal{M}$-pushout decomposition} property if whenever $a, b$ and $c$ belongs to $\mathcal{M}$ and the right half of the diagram above is a pushout, then its left half is a pushout too.
			\end{itemize}}
	\end{definition}} \parbox{3cm}{\vspace{-2cm}$\xymatrix{X \ar[d]_{a} \ar[r]^{f}& \ar[r]^{g} Y \ar[d]^{b}& Z \ar[d]^{c}\\ A \ar[r]_{h}& B \ar[r]_{k}& C}$}\\

The class of categories of type $\mathbb{B}$ is closed under the same constructions of \Cref{thm:slice-functors}. This can be deduced at once from the fact that in such categories pullbacks and pushouts are computed component-wise.

\begin{lemma}\label{lem:closed} Let $\X$ be a category satisfying the $\mathcal{M}$-mixed and $\mathcal{M}$-pushout decomposition properties, then the following hold
	\begin{enumerate}
		\item  for every object $X$, $\X/X$ satisfies  the $\mathcal{M}/X$-mixed and the $\mathcal{M}/X$-pushout decomposition properties, while $X/\mathcal{M}$-adhesive satisfies their $X/\mathcal{M}$ variants, where
		\[\mathcal{M}/X:=\{m\in \mathcal{A}(\X/X) \mid m\in \mathcal{M} \}
		\qquad X/\mathcal{M}:=\{m\in \mathcal{A}(X/\X) \mid m\in
		\mathcal{M} \}
		\]
		\item  for every small
		category $\Y$, the category $\X^\Y$ satisfies the 
		$\mathcal{M}^{\Y}$-mixed and the $\mathcal{M}^{\Y}$-pushout decomposition properties, where
		\[\mathcal{M}^{\Y}:=\{\eta \in \mathcal{A}(\X^\Y) \mid \eta_Y \in
		\mathcal{M} \text{ for every } Y\in \Y\}\]
	\end{enumerate}
\end{lemma}

The following result shows that the mixed and pushout decomposition properties guarantee that every independence pair is strong.

\begin{theorem}\label{thm:good1}Let $\X$ be an $\mathcal{M}$-adhesive category with all pushouts and satisfying the $\mathcal{M}$-mixed and the $\mathcal{M}$-pushout decomposition properties. Then every left-linear rewriting system on $\X$ is strong enforcing.
\end{theorem}
\begin{proof}
	Consider the derivation $\der{D}=\{\dder{D}_i\}_{i=0}^1$, made by two sequentially independent derivations, depicted below
	
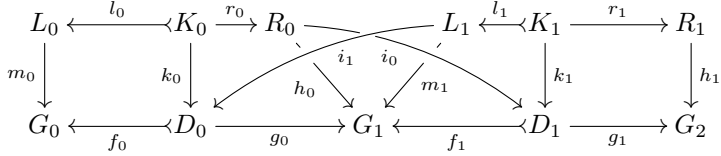
\[\xymatrix@C=15pt{L_0 \ar[d]_{m_0}&& K_0 \ar[d]_{k_0}\ar@{>->}[ll]_{l_0} \ar[r]^{r_0} & R_0 \ar@/^.35cm/[drrr]|(.3)\hole_(.4){i_0} \ar[dr]|(.3)\hole_{h_0} && L_1 \ar@/_.35cm/[dlll]^(.4){i_1} \ar[dl]|(.3)\hole^{m_1}& K_1 \ar[d]^{k_1}\ar@{>->}[l]_{l_1} \ar[rr]^{r_1} && R_1 \ar[d]^{h_1} \\G_0 && \ar@{>->}[ll]^{f_0} D_0 \ar[rr]_{g_0}&& G_1  && \ar@{>->}[ll]^{f_1} D_1 \ar[rr]_{g_1}&& G_2 }\]

	We have to show that $(i_0, i_1)$ is a strong independence pair. Now, by hypothesis $\X$ has all pushouts, thus the only thing to show is that the squares below are pushouts (the third one is the usual pullback of $f_1\colon D_1\mto G_0$ along $g_0\colon D_0\to G_1$)
	\[\xymatrix{K_0 \ar[r]^{r_0}  \ar[d]_{u_0}& R_0 \ar[d]^{i_0} & K_1 \ar@{>->}[r]^{l_1}  \ar[d]_{u_1}& L_1 \ar[d]^{i_1}&P \ar@{>->}[r]^{p_0} \ar[d]_{p_1}& D_0\ar[d]^{g_0} \\ P \ar[r]_{p_1} & D_1 &P \ar[r]_{p_0}  & D_0&D_1 \ar@{>->}[r]_{f_1} & G_1}\]

	To see this, consider the following two diagrams
	
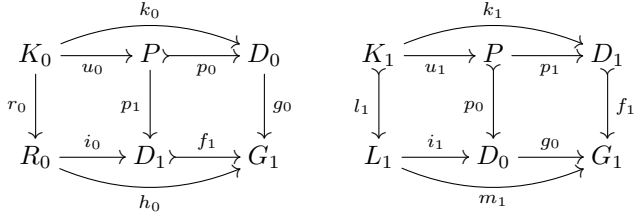
\[\xymatrix{K_0 \ar@/^.4cm/[rr]^{k_0} \ar[d]_{r_0} \ar[r]_{u_0}& \ar@{>->}[r]_{p_0} P \ar[d]_{p_1}& D_0 \ar[d]^{g_0}&K_1 \ar@/^.4cm/[rr]^{k_1} \ar@{>->}[d]_{l_1} \ar[r]_{u_1}& \ar[r]_{p_1} P \ar@{>->}[d]_{p_0}& D_1 \ar@{>->}[d]^{f_1}\\ R_0 \ar@/_.4cm/[rr]_{h_0}  \ar[r]^{i_0}& D_1 \ar@{>->}[r]^{f_1}& G_1 &L_1 \ar@/_.4cm/[rr]_{m_1}  \ar[r]^{i_1}& D_0 \ar[r]^{g_0}& G_1}\]

	The thesis now follows from the $\mathcal{M}$-mixed and the $\mathcal{M}$-pushout decomposition property.
\end{proof}

Our next step is to identify sufficient conditions for a category $\X$ to satisfy the mixed and $\mathcal{M}$-pushout decomposition properties.

\begin{definition} 
	Let $\X$ be an  $\mathcal{M}$-adhesive category, the pair $(\X, \mathcal{M})$ is of \emph{type $\mathbb{B}$} if
	\begin{enumerate}
		\item every arrow in $\mathcal{M}$ is a coproduct coprojection;
		\item $\X$ has all pushouts;
		\item $\X$ has strict initial objects and, for every object $X$, the unique arrow $?_X\colon 0\to X $ belongs to $\mathcal{M}$;
		\item all pushouts are $\mathcal{M}$-stable.
	\end{enumerate}
\end{definition}

\begin{remark}
	It is worth to examine more closely conditions $1$ and $3$ of the above definition.
	\smallskip 
	\begin{itemize}
	\parbox{10.3cm}{\item Let $m_0\colon X_0 \rightarrowtail Y$ be an arrow in $\mathcal{M}$, the first condition means that there exists $m_1\colon X_1\to Y$ such that $(Y, \{m_i\}_{i=0}^1)$ is a coproduct. This, together with property $3$, entails that every coprojection in a coproduct is in $\mathcal{M}$. This follows since any coproduct $(X_1+X_2, \{\iota_{X_i}\}_{i=0}^1)$ fits in a pushout diagram as the one aside.}
	  \parbox{3cm}{\vspace{-0em}$\xymatrix@C=25pt{0  \ar@{>->}[r]^-{?_{X_1}} \ar@{>->}[d]^{?_{X_2}}& X_1 \ar@{>->}[d]^{\iota_{X_1}}\\ X_2 \ar@{>->}[r]_-{\iota_{X_2}} & X_1+X_2}$}
		
	\smallskip 	\noindent 
	\parbox{13.5cm}{\item $\X$ has strict initial objects if it has initial objects and every arrow $f:X\to 0$ is an isomorphism. Notice that if initial objects are strict, then $?_X\colon 0\to X$ is mono for every $X$: indeed for every pair $f,g\colon Y\rightrightarrows 0$ then, by strictness, $Y$ is initial and so $f=g$.}
	\end{itemize} 
\end{remark}

\begin{example}
The category $\Set$ of sets and functions, with its class of monos, is of type $\mathbb{B}$. Similarly, the category $\textbf{Inj}$ of sets and injective functions is quasiadhesive and, with its class of regular monos, of type $\mathbb{B}$. 
\end{example}

Categories of type $\mathbb{B}$ satisfy a property resembling \emph{extensivity} \cite{carboni1993introduction}.

\noindent
\parbox{10cm}{\begin{proposition}\label{prop:ext}
Let $(\X, \mathcal{M})$ be a pair of type $\mathbb{B}$. Then for every diagram as the one aside, in which the bottom row is a coproduct cocone and the vertical arrows are in $\mathcal{M}$, the top row is a coproduct if and only if the two squares are pullbacks.
\end{proposition}}
\parbox{2cm}{\vspace{-.4em}$\xymatrix{A\ar[r]^-{f} \ar@{>->}[d]_{r}& C  \ar@{>->}[d]_{s}& B \ar@{>->}[d]^{t}\ar[l]_-{g}\\X  \ar[r]_-{\iota_X}& X+Y & Y \ar[l]^-{\iota_Y}}$}

\noindent 
\parbox{4cm}{$\xymatrix@C=10pt@R=10pt{&I\ar@{>->}[dd]|\hole_(.65){a}\ar[rr]^{k} \ar[dl]_{h} && B \ar@{>->}[dd]^{t} \ar[dl]_{g} \\ A  \ar@{>->}[dd]_{r}\ar[rr]^(.7){f} & & C \ar@{>->}[dd]_(.3){s}\\&0\ar@{>->}[rr]|\hole^(.65){?_Y} \ar@{>->}[dl]^{?_X} && Y \ar@{>->}[dl]^{\iota_Y} \\X \ar@{>->}[rr]_{\iota_X} & & X+Y}$}
 \parbox{10cm}{\begin{proof}
Consider the cube aside, in which the back faces are pullbacks. The bottom faces is an $\mathcal{M}$-Van Kampen pushout, thus the top face is a pushout if and only if the front faces are pullbacks. By strictness of $0$, $a\colon I\to 0$ is an isomorphism, so that $I$ is initial, therefore the $\mathcal{M}$-Van Kampen condition reduces to the request that $(C, \{f,g\})$ is a coproduct cocone if and only if the front faces are pullbacks, as claimed.
\end{proof}}

The previous result entails the following one, needed to show that  in any pair $(\X, \mathcal{M})$ of type $\mathbb{B}$, the category $\X$ satisfies the $\mathcal{M}$-mixed and $\mathcal{M}$-pushout decomposition properties.

\noindent
\parbox{10.7cm}{\begin{proposition}\label{prop:po2}
		Let $(\X, \mathcal{M})$ be a pair of type $\mathbb{B}$, and suppose that the square aside is an $\mathcal{M}$-pullback. Then there exists $E\in \X$, $e\colon E\mto C$ in $\mathcal{M}$, $\phi:E\to B_1$ such that $(C, \{m, e\})$ is a coproduct and $g=f+\phi$. Moreover, such a square is a pushout if and only if $\phi$ is an isomorphism.
\end{proposition}} \parbox{3cm}{\vspace{-.2cm}$\xymatrix{A  \ar[r]^{f} \ar@{>->}[d]_{m}& B_0 \ar@{>->}[d]^{\iota_{B_{0}}}\\ C\ar[r]_-{g} & B_0+B_1}$}

\begin{lemma}\label{lem:prop} Let $(\X, \mathcal{M})$ be a pair of type $\mathbb{B}$, then $\X$ satisfies the $\mathcal{M}$-mixed and $\mathcal{M}$-pushout decomposition properties.
\end{lemma}
\begin{proof}$\mathcal{M}$-mixed decomposition property. This follows at once from \Cref{prop:stab}.
	
	\medskip \noindent
	\parbox{3cm}{$\xymatrix{X \ar[r]^f \ar@{>->}[d]_{\iota_X}& Y\ar[r]^g \ar@{>->}[d]^{\iota_Y} & Z\ar@{>->}[d]^{\iota_Z}\\X+A  \ar@/_.4cm/[rr]_{(g\circ f) + (\varphi \circ \phi)}\ar[r]^{f+\phi}& Y+B \ar[r]^{g+\varphi}& Z+C}$} \qquad \qquad  \qquad   \qquad \parbox{8cm}{ $\mathcal{M}$-pushout-decomposition property. Using \Cref{prop:ext,prop:po2} and the fact that arrows in $\mathcal{M}$ are coproduct coprojections, we can reduce to prove the property for diagrams as the one on the left. By hypothesis and \Cref{prop:po2}, $\varphi$ and $\varphi\circ \phi$ are both isomorphisms, therefore $\phi$ is an isomorphism too and we can conclude again using \Cref{prop:po2}.\qedhere 
}	
\end{proof}

\begin{definition} Let $\X$ be an $\mathcal{M}$-adhesive category, we say that the pair $(\X, \mathcal{M})$ is of \emph{type $\mathbb{B}^+$} if (at least) one of the following holds
	\begin{enumerate}
		\item $(\X, \mathcal{M})$ is of type $\mathbb{B}$;
	\item $\X$ is $\Y/Y$ and $\mathcal{M}=\mathcal{N}/Y$for some $(\Y, \mathcal{N})$ of type $\mathbb{B}^+$ and $Y\in \Y$;
 	\item $\X$ is $Y/\Y$ and $\mathcal{M}=Y/\mathcal{N}$for some $(\Y, \mathcal{N})$ of type $\mathbb{B}^+$ and $Y\in \Y$:
 	\item $\X$ is $\Y^{\A}$ and $\mathcal{M}=\mathcal{N}^{\A}$ for some category category $\A$ and $(\Y, \mathcal{N})$ of type $\mathbb{B}^+$.
	\end{enumerate}
\end{definition}
\begin{example}
  \label{ex:graph-type-bb}
	The pair $(\mathbf{Graph}, \mon(\mathbf{Graph}))$ made by the topos of graphs and its class of monos is of type $\mathbb{B}^+$  but not of type $\mathbb{B}$. Indeed, not every monomorphism of graphs is a coproduct coprojection.
\end{example}

From \Cref{lem:closed,lem:prop} we can now deduce at once the following.
\begin{corollary}
	For every $(\X, \mathcal{M})$ of type $\mathbb{B}^+$, the category $\X$ has the $\mathcal{M}$-mixed and $\mathcal{M}$-pushout decomposition properties.
\end{corollary}

Using \Cref{thm:good1} we finally get the main result of our appendix.
\begin{corollary}
  \label{cor:llr-is-strEnf}
Let $(\X, \mathcal{M})$ be a pair of type $\mathbb{B}^+$, then every left-linear rewriting system is strong enforcing.
\end{corollary}

\section{Omitted proofs}\label{omitted}
In this appendix we provide the proofs of the main results of this paper. 

\subsection{Proofs for \Cref{sec:ade}}

\begin{proposition}
  \label{prop:unique}
  Let $\X$ be an $\mathcal{M}$-adhesive category. Suppose that the two
  direct derivations $\dder{D}$ and $\dder{D}'$ below, with the same
  match and applying the same left-linear rule $\rho$, are given
  \[\xymatrix{L \ar[d]_{m}& K \ar[d]^{k}\ar@{>->}[l]_{l} \ar[r]^{r} &
      R \ar[d]^{h} & L \ar[d]_{m}& K \ar[d]^{k'}\ar@{>->}[l]_{l}
      \ar[r]^{r} & R \ar[d]^{h'}\\G & \ar@{>->}[l]^{f} D \ar[r]_{g}& H
      & G & \ar@{>->}[l]^{f'} D' \ar[r]_{g'}& H'}\]
  Then there is an abstraction equivalence between $\dder{D}$ and
  $\dder{D}'$, whose first component is $\id{G}$.
\end{proposition}
\begin{proof}
  Both pairs $(k, f)$ and $(k', f')$ are pushout complements of $l$
  and $n$, thus, \Cref{lem:pocomp} yields an isomorphism
  $\phi_D\colon D\to D'$. An easy computation shows that $g'\circ \phi_D \circ k=h'\circ r$.
  \noindent
  \parbox{3cm}{
    $\xymatrix{K \ar@/^.4cm/[rr]^{k'}\ar[d]_{r} \ar[r]_{k}&
      \ar[r]_{\phi_D} D \ar[d]^{g}& D' \ar[d]^{g'}\\ R
      \ar@/_.4cm/[rr]_{h'} \ar[r]^{h}& H \ar[r]^{\phi_H}& H'}$}
  \hfill
  \parbox{10cm}{ \hspace{15pt}Hence, we have
    $\phi_H\colon H\to H'$. To see that $\phi_H$ is an isomorphism,
    consider the diagram aside. By hypothesis the whole rectangle and
    its left half are pushouts, therefore, by \Cref{lem:po1} its right
    square is a pushout too. The claim now follows from the fact that
    the pushout of an isomorphism is an isomorphism. \qedhere }
\end{proof}

\subsection{Proofs for \Cref{sec:equi}}
This section provides further details and proofs for the concepts introduced in \Cref{sec:equi}.

\subsubsection{Proofs for \Cref{subsec:switch}}
Consider a switch $\der{E}=\{\dder{E}_i\}_{i=0}^1$ for $\dder{D}_0$
and $\dder{D}_1$, then by definition $\dder{E}_0$ and $\dder{E}_1$ are
sequentially independent. We can then wonder if they are switchable.

\begin{proposition}
  \label{prop:switch}
  Let $\der{D}=\{\dder{D}_i\}_{i=0}^1$ be a derivation and suppose
  that $\der{E}=\{\dder{E}_i\}_{i=0}^1$ is a switch for it. Then
  $\der{D}$ is a switch for $\dder{E}_0$ and $\dder{E}_1$ along
  $(i'_0, i'_1)$.
\end{proposition}

\begin{proof}
 Clearly $\der{D}$ uses the same rule of $\der{E}$ in reverse order. Moreover, since $\der{E}$ is a switch for $\der{D}$, we already know that
   \begin{center}   
 	$m_0=f_0' \circ i_1'$
 	\qquad $h_1=g_1' \circ i_0'$
 	\qquad $m_0'= f_0 \circ i_1$
 	\qquad $h_1'= g_{1}\circ i_0$.
 \end{center}
Since  these equations are those that must hold for $\der{D}$ to be a switch for $\der{E}$ along $(i_0', i_1')$ we can conclude.
\end{proof}

\begin{lemma}[Uniqueness of switches]
  \label{thm:switch_uni}
  Let $\der{D}=\{\dder{D}_{i}\}_{i=0}^1$ be a derivation and suppose
  that both $\der{E}=\{\dder{E}_i\}_{i=0}^1$ and
  $\der{F}=\{\dder{F}_i\}_{i=0}^1$ are switches of $\dder{D}_0$ and
  $\dder{D}_1$. Then $\der{E}$ and $\der{F}$ are abstraction
  equivalent.
\end{lemma}

\begin{proof}
    By definition of switch $\der{E}_0$ and $\der{F}_0$ have the same match $f_0\circ i_1$ and that $\der{E}_1$ and $\der{F}_1$ have the same comatch $g_1\circ i_0$. Thus we get the solid part of the diagram below
    \[\xymatrix@C=22pt{G_{0} && \ar@{>->}[ll]_{f_{\der{F},0}} D_{\der{F},0} \ar[rr]^{g_{\der{F},0}}&& G_{\der{F},1} && \ar@{>->}[ll]_{f_{\der{F},1}} D_{\der{F},1} \ar[rr]^{g_{\der{F},1}}&& G_{2}\\L_1 \ar[d]^{f_0\circ i_1} \ar[u]_{f_0\circ i_1}&& K_{1} \ar[d]^{k_{\der{E},0}} \ar[u]_{k_{\der{F},0}}\ar@{>->}[ll]_{l_1} \ar[r]^{r_1} & R_{1} \ar@/^.35cm/[drrr]|(.3)\hole_(.4){j_0} \ar[dr]|(.3)\hole_{h_{\der{E},0}} \ar@/_.35cm/[urrr]|(.3)\hole^(.4){a_0} \ar[ur]|(.3)\hole^{h_{\der{F},0}}& &L_{0} \ar@/_.35cm/[dlll]^(.4){j_1} \ar[dl]|(.3)\hole^{m_{\der{E},1}} \ar@/^.35cm/[ulll]_(.4){a_1} \ar[ul]|(.3)\hole_{m_{\der{F},1}}& K_{0} \ar[d]_{k_{\der{E},1}} \ar[u]^{k_{\der{F},1}}\ar@{>->}[l]_{l_{0}} \ar[rr]^{r_0} && R_{0} \ar[d]_{g_1\circ i_0} \ar[u]^{g_1\circ i_0}\\ \ar@/^.5cm/[uu]^{\id{G_0}} G_{0} && \ar[ll]^{f_{\der{E},0}} D_{\der{E},0} \ar@{.>}@/^.5cm/[uu]^(.33){\phi_{0}}|\hole \ar[rr]_{g_{\der{E},0}}&& G_{\der{E},1} \ar@{.>}[uu]^(.76){\psi_1}|(.45)\hole |(.55)\hole&& \ar@{>->}[ll]^{f_{\der{E},1}} D_{\der{E},1} \ar@{.>}@/_.5cm/[uu]_(.33){\phi_{1}}|\hole\ar[rr]_{g_{\der{E},1}}&& G_{2}\ar@{.>}@/_.5cm/[uu]_{\psi_2}}\]
    
    \parbox{7cm}{\hspace{15pt}Now, by \Cref{prop:unique} we get $\phi_0\colon D_{\der{E},0}\to D_{\der{F},0}$ and $\psi_1 \colon G_{\der{E},1}\to G_{\der{F},1}$ such that the diagram aside commutes. But then we have
    	\[
    	f_{\der{F},0}\circ \phi_0\circ j_1=f_{\der{E},0} \circ j_1=m_0=f_{\der{F},0}\circ a_1\]
    	which entails $a_1= \phi_0\circ j_1$.}
\parbox{4cm} {$\xymatrix{G_{0} & \ar@{>->}[l]_{f_{\der{F},0}} D_{\der{F},0} \ar[r]^{g_{\der{F},0}}& G_{\der{F},1} \\L_1 \ar[d]^{f_0\circ i_1} \ar[u]_{f_0\circ i_1}& K_{1} \ar[d]^{k_{\der{E},0}} \ar[u]_{k_{\der{F},0}}\ar[l]_{l_1} \ar[r]^{r_1} & R_{1}  \ar[d]_{h_{\der{E},0}}  \ar[u]^{h_{\der{F},0}}\\ \ar@/^.5cm/[uu]^{\id{G_0}} G_{0} & \ar@{>->}[l]^{f_{\der{E},0}} D_{\der{E},0} \ar@/^.5cm/[uu]^(.33){\phi_{0}}|\hole \ar[r]_{g_{\der{E},0}}& G_{\der{E},1} \ar@/_.5cm/[uu]_{\psi_1}}$}
    
     This, in turn, gives us that
    \[\psi\circ m_{\der{E},1}=\psi \circ g_{\der{E},0}\circ j_1=g_{\der{F},0}\circ \phi_0\circ j_1=g_{\der{F},0}\circ a_1=m_{\der{F},1}\]
    
    Now, the square on the left in the diagram below is a pushout and $\psi_1\circ m_{\der{E},1}=m_{\der{F},1}$. Thus, by \Cref{prop:unique} there are the dotted $\phi_1\colon D_{\der{E},1}\to D_{\der{F},1}$ and $\psi_2\colon G_2\to G_2$ in the diagram on the right
    \[\xymatrix{&&G_{\der{F},1} & \ar@{>->}[l]_{f_{\der{F},1}} D_{\der{F},1} \ar[r]^{g_{\der{F},1}}& G_2 \\K_0 \ar[d]_{k_{\der{E,1}}} \ar@{>->}[r]^{l_0}& L_0 \ar[d]^{\psi_1\circ m_{\der{E},1}}&L_0 \ar[d]^{m_{\der{E},1}} \ar[u]_{m_{\der{F},1}}& K_{1} \ar[d]^{m_{\der{E},1}} \ar[u]_{m_{\der{F},1}}\ar@{>->}[l]_{l_1} \ar[r]^{r_1} & R_{1}  \ar[d]_{g_1\circ i_0}  \ar[u]^{g_1\circ i_0}\\D_{\der{E},1} \ar@{>->}[r]_{\psi_1\circ f_{\der{E},1}}& G_{\der{F}, 1}& \ar@/^.5cm/[uu]^{\psi_1} G_{\der{E},1} & \ar@{>->}[l]^{f_{\der{E},1}} D_{\der{E},1} \ar@{.>}@/^.5cm/[uu]^(.33){\phi_{1}}|\hole \ar[r]_{g_{\der{E},1}}& G_2 \ar@{.>}@/_.5cm/[uu]_{\psi_2}}\]
    
    To conclude, it is now enough to notice that
    \[f_{\der{F},1} \circ \phi_1\circ j_0=\psi_1\circ f_{\der{E},1}\circ j_0=\psi_1\circ h_{\der{E,0}}=h_{\der{F},0}=f_{\der{F},1}\circ a_0\]
    allowing us to deduce $ \phi_1\circ j_0=a_0$.
    \qedhere
\end{proof}

\begin{proposition}
	\label{prop:tec}
	Let $(\X,\R)$ be a left-linear rewriting system and
	$(i_0, i_1)$ an independence pair between two direct
	derivations $\dder{D}_0$ and $\dder{D}_1$. Consider the pullback of
	$f_1\colon D_1\mto G_1$ along $g_0\colon D_0\to G_1$ (first
	square below), then there exist the arrows
	$u_1\colon K_1\to P$ and $u_0\colon K_0\to P$ fitting in the
	central and right square. Moreover, the right one is always
	a pullback.
	\[\xymatrix@R=15pt{P \ar@{>->}[r]^{p_0} \ar[d]_{p_1} & D_0
		\ar[d]^{g_0}& K_0 \ar[r]^{r_0} \ar@{.>}[d]_{u_0}& R_0
		\ar[d]^{i_0} & K_1 \ar@{>->}[r]^{l_1} \ar@{.>}[d]_{u_1}&
		L_1 \ar[d]^{i_1}\\ D_1 \ar@{>->}[r]_{f_1} &G_1 &P
		\ar[r]_{p_1} & D_1 &P \ar[r]_{p_0} & D_0}
	\]
\end{proposition}

\begin{proof}
	We start noticing that
	\[
	f_1\circ i_0\circ r_0=h_0\circ r_0=g_0\circ k_0 \qquad
	g_0\circ i_1\circ l_1=m_1\circ l_1= f_1 \circ k_1
	\]
	Thus there exists the dotted $u_0\colon K_0\to P$,
	$u_1\colon K_1\to P$ in the diagram
	\[\xymatrix@R=15pt{K_0 \ar[d]_{r_0}
		\ar@{.>}[r]_{u_0}\ar@/^.4cm/[rr]^{k_0} &P
		\ar@{>->}[r]_{p_0} \ar[d]_{p_1}& D_0 \ar[d]^{g_0}&K_1
		\ar@{>->}[d]_{l_1}
		\ar@{.>}[r]_{u_1}\ar@/^.4cm/[rr]^{k_1} &P \ar[r]_{p_1}
		\ar@{>->}[d]_{p_0}& D_1 \ar@{>->}[d]^{f_1}\\R_0
		\ar@/_.4cm/[rr]_{h_0}\ar[r]^{i_0}& D_1
		\ar@{>->}[r]^{f_1}& G_1&L_1
		\ar@/_.4cm/[rr]_{m_1}\ar[r]^{i_1}& D_0 \ar[r]^{g_0}&
		G_1}\]
	By \Cref{prop:pbpoad,lem:pb1} we get that the left half of
	the second rectangle is a pullback.
	\qedhere
\end{proof}

There is another
useful property of  sequentially independent direct
derivations that is worth mentioning.

\noindent
\parbox{9.5cm}
{\begin{proposition}\label{lem:cose}Let
    $\der{D}=\{\dder{D}_i\}_{i=0}^1$ be a derivation and suppose
    that $(i_0, i_1)$ is an independence pair between $\dder{D}_0$
    and $\der{D}_1$.  If $\der{E}=\{\dder{E}_i\}_{i=0}^1$ is a
    switch of $\der{D}$ along $(i_0, i_1)$, then there exists
    $q_0\colon P\to D_{\der{E},0}$ in the diagram
    aside. \end{proposition}}
\parbox{4cm}{$\xymatrix{D_{\der{E},0} \ar@{>->}[d]_{f_{\der{E},0}}&P
    \ar@{.>}[l]_{q_0} \ar[r]^{p_1} \ar@{>->}[d]_{p_0}& D_1
    \ar@{>->}[d]^{f_1}\\ G_{0}&D_0 \ar@{>->}[l]^{f_0}\ar[r]_{g_0}&
    G_1}$}

\begin{proof}
  Consider the arrow $u_1\colon K_1\to P$ obtained using
  \Cref{prop:tec}. It fits in the big square on the left
\[		\xymatrix@R=15pt{K_1 \ar[r]^{u_1} \ar@{>->}[d]_{l_1}&P
	\ar[r]^{\id{P}}\ar@{>->}[d]_{p_0}& P\ar@{>->}[d]^{p_0}\\L_1
	\ar[d]_{\id{L_1}}\ar[r]^{i_1}&D_0 \ar[r]^{\id{D_0}}
	\ar[d]_{\id{D_0}}& D_{0} \ar@{>->}[d]^{f_0}&K_1 \ar@{>->}[d]_{l_1}\ar[r]^{k_{\der{E},0}}&
	D_{\der{E},0} \ar@{>->}[d]^{f_{\der{E},0}} & K_1 \ar[r]^{u}
	\ar@{>->}[d]_{l_1}& P \ar@{>->}[d]^{f_0\circ
		p_0}\\
	L_1\ar[r]^{i_1} \ar@/_.4cm/[rr]_{m_{\der{E},0}}&D_0
	\ar@{>->}[r]^{f_0}&G_0 & L_1\ar[r]_{m_{\der{E},0}}&G_0 &L_1
	\ar[r]_{m_{\der{E},0}}&G_0}\]

  	Since $f_0$ is a monomorphism, every square in the diagram on the left above
  	is a pullback. Thus the whole big square is a pullback too.
  	Applying \Cref{lem:radj} to the pair of squares on the right now gives
  	us the wanted $q_0\colon P\to D_{\der{E},0}$.
\end{proof}

\subsubsection{Proofs for \Cref{subsec:CR}}

To prove the Local Church-Rosser Theorem, let us begin by deducing some properties from the existence of a strong independence pair. 

\begin{remark}\label{rem:deco} 
	Let $(i_0, i_1)$ be a strong
	independence pair between the direct derivations $\dder{D}_0$ and
	$\dder{D}_1$. We can then build the solid part of the diagram
	below
	
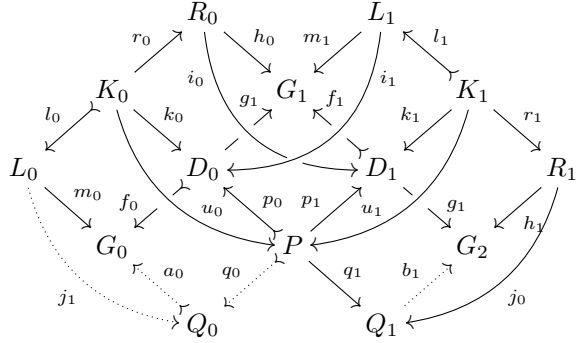
\[\xymatrix@C=15pt@R=15pt{&&R_0 \ar@/_.8cm/[ddrr]_(.2){i_0}|(.7)\hole
		\ar[dr]^{h_0}&& L_1\ar@/^.8cm/[ddll]^(.2){i_1}
		\ar[dl]_{m_1}\\&K_0\ar[dr]^{k_0}\ar@{>->}[dl]_{l_0}
		\ar@/_.6cm/[ddrr]^(.65){u_0}\ar[ur]^{r_0}&& G_1 && K_1
		\ar@/^.6cm/[ddll]_(.65){u_1}\ar[dl]_{k_1}\ar@{>->}[lu]_{l_1}
		\ar[dr]^{r_1}\\L_0 \ar[dr]^{m_0}
		\ar@{.>}@/_.6cm/[ddrr]_{j_1}&& D_0
		\ar@{>->}[dl]|(.42)\hole_(.64){f_0}\ar[ur]|(.48)\hole^(.7){g_1}&&D_1
		\ar[dr]|(.42)\hole^(.65){g_1}
		\ar@{>->}[ul]|(.48)\hole_(.7){f_1}&&R_1\ar@/^.6cm/[ddll]^{j_0}\ar[dl]^{h_1}\\&G_0
		&&P\ar[dr]^{q_1}
		\ar@{>.>}[dl]_{q_0}\ar[ur]^(.4){p_1}\ar@{>->}[ul]_(.4){p_0}&&G_2\\&&Q_0
		\ar@{>->}@{>.>}[ul]_{a_0} &&Q_1\ar@{.>}[ur]^{b_1}}
	\]
	
	Let us complete this diagram defining the dotted arrows. To get
	$j_1\colon L_0\to Q_0$ and $q_0\colon P\to Q_0$ it is enough to
	take a pushout of $l_0$ along $u_0$, which exists since
	$l_0\in \mathcal{M}$. Moreover, the existence of the wanted
	$a_0\colon Q_0\to G_0$ and $b_1\colon Q_1\to G_2$ follows from
	\[
	f_0\circ p_0 \circ u_0 = f_0\circ k_0 = m_0\circ l_0 \qquad
	g_1\circ p_1\circ u_1 = g_1\circ k_1=h_1\circ r_1
	\]
	
	\noindent
	\parbox{5.5cm}{\hspace{15pt}We can prove some other properties of
		the arrows appearing in the diagram above. The four rectangles
		aside are pushouts and their left halves are pushouts
		too. Therefore, by \Cref{lem:po1} also their right halves are
		pushouts.  Moreover, $q_0$ and $a_0$ are pushouts of
		respectively $l_0$ and $p_0$, thus they are elements of
		$\mathcal{M}$. By \Cref{prop:pbpoad} the left halves of the
		first and third rectangles are pullbacks too.}  \parbox{3.5cm}{
		$\xymatrix{K_0 \ar@/^.4cm/[rr]^{k_0}\ar[d]_{r_0}\ar[r]_{u_0}
			&P\ar[d]^{p_1} \ar@{>->}[r]_{p_0} & D_0 \ar[d]^{g_0}&K_1
			\ar@/^.4cm/[rr]^{k_1}\ar@{>->}[d]_{l_1}\ar[r]_{u_1}
			&P\ar@{>->}[d]^{p_0} \ar[r]_{p_1} & D_1 \ar@{>->}[d]^{f_1}\\
			R_0 \ar@/_.4cm/[rr]_{h_0} \ar[r]^{i_0}&D_1 \ar@{>->}[r]^{f_1}
			& G_1&L_1 \ar@/_.4cm/[rr]_{m_1} \ar[r]^{i_1}&D_0\ar[r]^{g_0} &
			G_1}$
		$\xymatrix{K_0
			\ar@/^.4cm/[rr]^{k_0}\ar@{>->}[d]_{l_0}\ar[r]_{u_0}
			&P\ar@{>->}[d]^{q_0} \ar@{>->}[r]_{p_0} & D_0
			\ar@{>->}[d]^{f_0}&K_1
			\ar@/^.4cm/[rr]^{k_1}\ar[d]_{r_1}\ar[r]_{u_1}
			&P\ar[d]^{q_1} \ar[r]_{p_1} & D_1 \ar[d]^{g_1}\\
			L_0 \ar@/_.4cm/[rr]_{m_0} \ar[r]^{j_1}&Q_0 \ar@{>->}[r]^{a_0}
			& G_0&R_1 \ar@/_.4cm/[rr]_{h_1} \ar[r]^{j_0}&Q_1 \ar[r]^{b_1}
			& G_2}$}
\end{remark}

\prChurch*
\label{prChurch-proof}

\begin{proof}

  Let us begin considering the diagram
  built in \Cref{rem:deco}
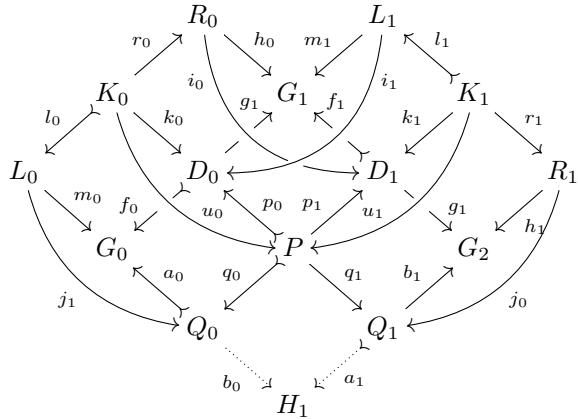
  \[\xymatrix@C=15pt@R=15pt{&&R_0
      \ar@/_.8cm/[ddrr]_(.2){i_0}|(.7)\hole \ar[dr]^{h_0}&&
      L_1\ar@/^.8cm/[ddll]^(.2){i_1}
      \ar[dl]_{m_1}\\&K_0\ar[dr]^{k_0}\ar@{>->}[dl]_{l_0}
      \ar@/_.6cm/[ddrr]^(.65){u_0}\ar[ur]^{r_0}&& G_1 && K_1
      \ar@/^.6cm/[ddll]_(.65){u_1}\ar[dl]_{k_1}\ar@{>->}[lu]_{l_1}
      \ar[dr]^{r_1}\\L_0 \ar[dr]^{m_0} \ar@{>}@/_.6cm/[ddrr]_{j_1}&& D_0
      \ar@{>->}[dl]|(.42)\hole_(.64){f_0}\ar[ur]|(.48)\hole^(.7){g_1}&&D_1
      \ar[dr]|(.42)\hole^(.65){g_1}
      \ar@{>->}[ul]|(.48)\hole_(.7){f_1}&&R_1\ar@/^.6cm/[ddll]^{j_0}\ar[dl]^{h_1}\\&G_0
      &&P\ar[dr]^{q_1}
      \ar@{>->}[dl]_{q_0}\ar[ur]^(.4){p_1}\ar@{>->}[ul]_(.4){p_0}&&G_2\\&&Q_0
      \ar@{>->}@{>->}[ul]_{a_0} \ar@{.>}[dr]_{b_0} &&Q_1\ar[ur]^{b_1}
      \ar@{>.>}[dl]^{a_1} \\ &&&H_1 }\] Since $q_0\colon P\to Q_0$ is in
  $\mathcal{M}$, we can take its pushout along $q_1$ to get the dotted
  arrows $a_1\colon Q_1\to H_1$ and $q_0\colon Q_0\to H_1$. But then,
  since the composite of two pushout squares is a pushout, the diagram
  below is a derivation
  \[\xymatrix@C=15pt{L_1 \ar[d]_{f_0 \circ i_1}&& K_1
      \ar[d]_{q_0\circ u_1}\ar@{>->}[ll]_{l_1} \ar[r]^{r_1} & R_1
      \ar@/^.35cm/[drrr]|(.3)\hole_(.4){j_0} \ar[dr]|(.3)\hole_{a_1\circ
        j_0} && L_0 \ar@/_.35cm/[dlll]^(.4){j_1} \ar[dl]|(.3)\hole^{b_0\circ
        j_1}& K_1 \ar[d]^{q_1\circ u_0}\ar@{>->}[l]_{l_0} \ar[rr]^{r_0} && R_0
      \ar[d]^{g_1\circ i_0} \\G_0 && \ar@{>->}[ll]^{a_0} D_0 \ar[rr]_{b_0}&&
      G_1 && \ar@{>->}[ll]^{a_1} D_1 \ar[rr]_{b_1}&& G_2}\] It is immediate
  now to see that it is a switch of $\dder{D}_0$ and $\dder{D}_1$ along
  $(i_0, i_1)$.
\end{proof}

\begin{remark}
	\label{rem:church}
	It is worth to point out that, by
	\Cref{prop:switch,thm:switch_uni}, the derivation we obtain from
	\Cref{pr:church} 
	is made by two switchable derivations that will
	give back the starting derivation (up to abstraction equivalence),
	if switched again.
	
	Moreover, by \Cref{prop:pbpoad}, the first square below is a
	pullback, therefore, by \cref{rem:deco} the last three squares
	witness that $(j_0, j_1)$ is a strong independence pair too.
	\[
	\xymatrix@R=15pt{P\ar[r]^{q_1} \ar@{>->}[d]_{q_0}&Q_1
		\ar@{>->}[d]^{a_1}&K_1 \ar[r]^{r_1} \ar[d]_{u_1}& R_1
		\ar[d]^{j_0} & K_0 \ar@{>->}[r]^{l_0} \ar[d]_{u_0}& L_0
		\ar[d]^{j_1}&K_0 \ar[r]^{r_0} \ar[d]_{u_0}& R_0\ar[d]^{i_0} \\
		Q_ 0\ar[r]_{b_0}&H_1&P \ar[r]_{q_1} & Q_1 &P \ar[r]_{q_0} &
		Q_0&P \ar[r]_{p_1} & D_1}
	\]
\end{remark}

\begin{proposition}
  \label{prop:equi}
  Every linear rewriting system $(\X, \R)$ is strong enforcing.
\end{proposition}

\begin{proof}
  Suppose that $\X$ is $\mathcal{M}$-adhesive and let $(i_0, i_1)$ be
  an independence pair as below
  \[
    \xymatrix@C=15pt{L_0 \ar[d]_{m_0}&& K_0
      \ar[d]_{k_0}\ar@{>->}[ll]_{l_0} \ar@{>->}[r]^{r_0} & R_0
      \ar@/^.35cm/[drrr]|(.3)\hole_(.4){i_0} \ar[dr]|(.3)\hole_{h_0}
      && L_1 \ar@/_.35cm/[dlll]^(.4){i_1} \ar[dl]|(.3)\hole^{m_1}& K_1
      \ar[d]^{k_1}\ar@{>->}[l]_{l_1} \ar@{>->}[rr]^{r_1} && R_1
      \ar[d]^{h_1} \\G_0 && \ar@{>->}[ll]^{f_0} D_0
      \ar@{>->}[rr]_{g_0}&& G_1 && \ar@{>->}[ll]^{f_1} D_1
      \ar@{>->}[rr]_{g_1}&& G_2}\]

  Since $(\X, \R)$ is linear, then $r_1\colon K_1\to R_1$ belongs to
  $\mathcal{M}$, thus it admits a pushout along $u_1\colon K_1\to P$,
  as desired. Moreover, let us consider the two rectangles below
  \[
    \xymatrix{K_0 \ar@/^.4cm/[rr]^{k_0}\ar@{>->}[d]_{r_0}\ar[r]_{u_0}
      &P\ar@{>->}[d]^{p_1} \ar@{>->}[r]_{p_0} & D_0
      \ar@{>->}[d]^{f_0}& K_1
      \ar@/^.4cm/[rr]^{k_1}\ar@{>->}[d]_{l_1}\ar[r]_{u_1}
      &P\ar@{>->}[d]^{p_0} \ar@{>->}[r]_{p_1} & D_1
      \ar@{>->}[d]^{f_1}\\ R_0 \ar@/_.4cm/[rr]_{h_0} \ar[r]^{i_0}&D_1
      \ar@{>->}[r]^{f_1} & G_1& L_1 \ar@/_.4cm/[rr]_{m_1}
      \ar[r]^{i_1}&D_0 \ar@{>->}[r]^{g_0} & G_1} \] By hypothesis
  $r_0$ and $l_1$ are in $\mathcal{M}$, thus $f_1$ and $g_0$ belong to
  it too. The first point of \Cref{lem:popb} yields the thesis.
\end{proof}

\begin{remark}\label{rem:fill} Let $\der{D}=\{\dder{D}_i\}_{i=0}^n$ be a derivation and $(i_0, i_1)$ a strong independence pair between $\dder{D}_0$ and $\dder{D}_1$. Let also $\{\der{E}_{i}\}_{i=0}^1$ be a switch for them along $(i_0, i_1)$. Since, by \Cref{thm:switch_uni}, all switches are abstraction equivalent, we get the diagram below
	\[\xymatrix@C=15pt@R=15pt{&&R_0 \ar@/_.8cm/[ddrr]_(.2){i_0}|(.7)\hole
		\ar[dr]^{h_0}&& L_1\ar@/^.8cm/[ddll]^(.2){i_1}
		\ar[dl]_{m_1}\\&K_0\ar[dr]^{k_0}\ar@{>->}[dl]_{l_0}
		\ar@/_.6cm/[ddrr]|(.36)\hole^(.65){u_0}\ar[ur]^{r_0}&& G_1 &&
		K_1
		\ar@/^.6cm/[ddll]|(.36)\hole_(.65){u_1}\ar[dl]_{k_1}\ar@{>->}[lu]_{l_1}
		\ar[dr]^{r_1}\\L_0
		\ar@/_.6cm/[ddrr]_(.2){i'_0}|(.31)\hole|(.81)\hole
		\ar[dr]^(.4){m_0}|(.61)\hole && D_0
		\ar@{>->}[dl]|(.4)\hole_(.65){f_0}\ar[ur]|(.5)\hole^(.7){g_0}&&D_1
		\ar[dr]|(.4)\hole^(.65){g_1}
		\ar@{>->}[ul]|(.5)\hole_(.65){f_1}&&R_1\ar@/^.6cm/[ddll]^(.2){i'_1}|(.31)\hole|(.81)\hole\ar[dl]_{h_1}|(.6)\hole\\&G_0
		&&P_1\ar[dr]_{q_1}	\ar@{>->}[dl]^{q_0}\ar[ur]^(.4){p_1}\ar@{>->}[ul]_(.4){p_0}&&G_2\\L_1	\ar@/^.6cm/[uurr]^(.2){i_1} \ar[ur]_(.35){m'_0}|(.61)\hole && D'_0	\ar@{>->}[ul]|(.4)\hole_(.65){f'_0}\ar[dr]|(.5)\hole^(.65){g'_0}	&&D'_1\ar[ur]|(.4)\hole^(.65){g'_1} \ar@{>->}[dl]|(.5)\hole_(.65){f'_1}	&& R_0 \ar[ul]^(.35){h'_1}|(.61)\hole\ar@/_.6cm/[uull]_(.2){i_0}\\ &K_1	\ar[ur]_{k'_0} \ar[dr]_{r_1}	\ar@{>->}[ul]^{l_1}\ar@/^.6cm/[uurr]_(.65){u_1}&&G'_1&& K_0\ar[ur]_{r_0} \ar[ul]^{k'_1} \ar@{>->}[dl]^{l_0} \ar@/_.6cm/[uull]^(.65){u_0}\\&& R_1	\ar[ur]_{h'_0}\ar@/^.8cm/[uurr]^(.2){i'_1}&& L_0 \ar[ul]^{m'_1} \ar@/_.8cm/[uull]_(.2){i'_0} |(.69)\hole }\] 

\end{remark}

\subsubsection{Proofs for \Cref{subsec:verytame}}

\begin{proposition}\label{pr:weak} Every linear rewriting system is well-switching.
\end{proposition}
\begin{proof}
	
	Let $(i_0, i_1)$ and $(i'_0, i'_1)$ be independence pairs for the
	direct derivations $\dder{D}$ and $\dder{D}'$. Notice that, by
	definition, we have
	\[
	f_1\circ i_0=h_0=f_1\circ i'_0 \qquad g_0\circ i_1=m_1= g_0\circ i'_1
	\]
	
	The arrow $f_1\colon D\to G_1$ is the pushout of
	$l_1\colon K_1\mto L_1$, thus it is in $\mathcal{M}$ implying
	$i_0=i'_0$. If, moreover, we suppose that the rule $\rho$ applied in
	$\dder{D}$ is linear, then $g_0\colon D_0\to G_1$ is in
	$\mathcal{M}$ too, entailing $i_1=i'_1$.
\end{proof}

\begin{remark}\label{rem:po}
	Notice that $\gph{X}$ is adhesive and that pushouts are computed component-wise, thus if $g\colon F\to G$ is the pushout of $r\colon K\to R$, then $g_X\colon F(X)\to G(X)$ is mono for every root $X$.
\end{remark}

\lemVTame*
\label{lemVTame}

\begin{proof}
	Every root-preserving graphical rewriting system is strong enforcing because $\gph{G}$ is in the class $\mathbb{B}^+$ defined in \Cref{app:fill}.  Consider the following diagram 
	
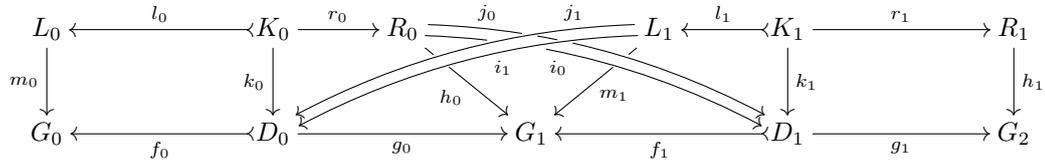
\[\xymatrix@R=25pt@C=30pt{L_0 \ar[d]_{m_0}&& K_0
		\ar[d]_{k_0}\ar@{>->}[ll]_{l_0} \ar[r]^{r_0} & R_0
		\ar@<-.5ex>@/^.35cm/[drrr]|(.23)\hole|(.34)\hole_(.4){i_0} 	\ar@<.5ex>@/^.35cm/[drrr]^(.2){j_0}|(.31)\hole |(.4)\hole
		\ar[dr]|(.3)\hole_{h_0} && L_1 \ar@<.5ex>@/_.35cm/[dlll]^(.4){i_1} \ar@<-.5ex>@/_.35cm/[dlll]_(.2){j_1}
		\ar[dl]|(.3)\hole^{m_1}& K_1 \ar[d]^{k_1}\ar@{>->}[l]_{l_1}
		\ar[rr]^{r_1} && R_1 \ar[d]^{h_1}\\G_0 && \ar@{>->}[ll]^{f_0}
		D_0 \ar[rr]_{g_0}&& G_1 && \ar@{>->}[ll]^{f_1} D_1
		\ar[rr]_{g_1}&& G_2}\] 
	
	Since $f_1$ is mono then $i_0=j_0$. For every  root $X$ we have
	\[g_{0,X}\circ i_{1,X}=m_{1,X}=g_{0,X}\circ j_{1,X}\]
	which entails $i_{1,X}=j_{1,X}$. Let now $Y$ be another object, and take $y\in L_1(Y)$. By definition, we have a root $X$ and an arrow $f:X\to Y$ and $x\in L_{1}(X)$ such that $L(f)(y)=x$. By naturality we have
	\[i_{1,Y}(y)=i_{1,Y}(L(f)(x))=L(f)(i_{1, X}(x))=L(f)(j_{1, X}(x))=j_{1,X}(L(f)(x))=j_{1,Y}(y)\]
Thus $i_{1,Y}=j_{1,Y}$, as wanted.
\end{proof}

\subsubsection{Proofs for \cref{subsec:canonical}}

The section opens recalling the definition of \emph{consistent permutation} for the interested reader.

\begin{definition}[Consistent permutation]
	\label{def:permcon}
	Let $\X$ be an $\mathcal{M}$-adhesive category and $(\X, \R)$ 
	a left-linear rewriting system.  Take two derivations $\der{D}:=\{\dder{D}_i\}_{i=0}^n$ and
	$\der{D}':=\{\dder{D}'_i\}_{i=0}^n$ with the same length and with isomorphic sources and targets.
	Let also $r(\der{D})=\{\rho_i\}_{i=0}^n$ and
	$r(\der{D}')=\{\rho'_i\}_{i=0}^n$ be the  associated sequence of rules. 
	A \emph{consistent permutation}
	between $\der{D}$ and
	$\der{D}'$ is a permutation
	$\sigma\colon [0,n]\to [0,n]$ such that, for every $i\in [0,n]$,
	$\rho_i=\rho'_{\sigma(i)}$ and, moreover, there exists a
	\emph{mediating isomorphism}
	$\xi_\sigma\colon \tpro{D} \to \lpro \der{D}' \rpro$ fitting in the
	following diagrams, where $m_i, m'_i, h_i$ and $h'_i$ are, the matches and co-matches of $\dder{D}_i$ and
	$\dder{D}'_i$.
	\[
	\xymatrix@C=30pt{L_i
		\ar[r]^{m_i} \ar[d]_{m'_{\sigma(i)}}& G_i \ar[r]^{\iota_{G_i}}
		&\tpro{D} \ar[d]^{\xi_\sigma} & R_i \ar[r]^{h_i}
		\ar[d]_{h'_{\sigma(i)}}& G_{i+1} \ar[r]^{\iota_{G_{i+1}}}
		&\tpro{D} \ar[d]^{\xi_\sigma} \\G'_{\sigma(i)}
		\ar[rr]_{\iota'_{G'_{\sigma(i)}}}&& \lpro \der{D}' \rpro&
		G'_{\sigma(i)+1} \ar[rr]_{\iota'_{G'_{\sigma(i)+1}}}&& \lpro
		\der{D}' \rpro}
	\]
\end{definition}

We can now turn to the proofs of the results of \Cref{subsec:canonical}.

\begin{notation}
  Recall that by \cref{thm:switch_uni} switches are uniquely
  determined, up to abstraction equivalence.
  In the sequel we denote any switch of two switchable direct
  derivations $\dder{D}_0$ and $\dder{D}_1$ along an independence pair
  $(i_0, i_1)$ by $S_{i_0, i_1}(\dder{D}_0, \dder{D}_1)$. Moreover, we
  use $S_{i_0,i_1}(\dder{D}_1)$ and $S_{i_0,i_1}(\dder{D}_0)$ to
  denote its components, in such a way that
  $S_{i_0, i_1}(\dder{D}_0, \dder{D}_1)=S_{i_0,i_1}(\dder{D}_1)\cdot
  S_{i_0,i_1}(\dder{D}_0)$.
\end{notation}

\lemIndepGlobalLeft*
\label{lemIndepGlobalLeft-proof}

\begin{proof}   
	Since $(\X, \R)$ is well-switching, there exists unique independence pairs 
$(i_0,i_1)$ between $\dder{D}_0$ and $\dder{D}_1$,	$(a_0,a_1)$ between $\dder{D}_1$ and $\dder{D}_2$ and	$(e_0, e_1)$ between $\dder{D}_0$ and
	$S_{a_0,a_1}(\dder{D}_2)$. 	Using \Cref{rem:fill} we can build the diagrams in \Cref{fi:first,fi:second,fi:third}.  The first depicts the  switching $\dder{D}_0$ and $\dder{D}_1$ along  $(i_0, i_1)$, the second that of $\dder{D}_1$ and $\dder{D}_2$ along $(a_0, a_1)$, while the third refers to $\dder{D}_0$ and
	$S_{a_0,a_1}(\dder{D}_2)$.
	
	\begin{figure}[h]
		\centering
		$\xymatrix@C=15pt@R=15pt{&&R_0 \ar@/_.8cm/[ddrr]_(.2){i_0}|(.7)\hole
			\ar[dr]^{h_0}&& L_1\ar@/^.8cm/[ddll]^(.2){i_1}
			\ar[dl]_{m_1}\\&K_0\ar@{>->}[dr]^{k_0}\ar[dl]_{l_0}
			\ar@/_.6cm/[ddrr]|(.36)\hole^(.65){v_0}\ar[ur]^{r_0}&& G_1 &&
			K_1
			\ar@/^.6cm/[ddll]|(.36)\hole_(.65){v_1}\ar[dl]_{k_1}\ar@{>->}[lu]_{l_1}
			\ar[dr]^{r_1}\\L_0
			\ar@/_.6cm/[ddrr]_(.2){j_0}|(.31)\hole|(.81)\hole
			\ar[dr]^(.4){m_0}|(.61)\hole && D_0
			\ar@{>->}[dl]|(.4)\hole_(.65){f_0}\ar[ur]|(.5)\hole^(.7){g_0}&&D_1
			\ar[dr]|(.4)\hole^(.65){g_1}
			\ar@{>->}[ul]|(.5)\hole_(.65){f_1}&&R_1\ar@/^.6cm/[ddll]^(.2){j_1}|(.31)\hole|(.81)\hole\ar[dl]_(.4){h_1}|(.6)\hole\\&G_0
			&&P_1\ar[dr]_{q_1}
			\ar@{>->}[dl]^{q_0}\ar[ur]^(.4){p_1}\ar@{>->}[ul]_(.4){p_0}&&G_2\\L_1
			\ar@/^.6cm/[uurr]^(.2){i_1} \ar[ur]_(.35){f_0\circ
				i_1}|(.61)\hole&&Q_0
			\ar@{>->}[ul]|(.4)\hole_(.65){s_0}\ar[dr]|(.5)\hole^(.65){t_0}
			&&Q_1\ar[ur]|(.4)\hole^(.65){t_1} \ar@{>->}[dl]|(.5)\hole_(.65){s_1}
			&& R_0 \ar[ul]^(.35){g_1\circ
				i_0}|(.61)\hole\ar@/_.6cm/[uull]_(.2){i_0}\\&K_1
			\ar[ur]_{q_0\circ v_1} \ar[dr]_{r_1}
			\ar@{>->}[ul]^{l_1}\ar@/^.6cm/[uurr]_(.65){v_1}&&H'_1&& K_0
			\ar[ur]_{r_0} \ar[ul]^{q_1\circ v_0} \ar@{>->}[dl]^{l_0}
			\ar@/_.6cm/[uull]^(.65){v_0}\\&& R_1
			\ar[ur]_{\hspace{-6pt}s_1\circ
				j_1}\ar@/^.8cm/[uurr]^(.2){j_1}&& L_0 \ar[ul]^{t_0\circ
				j_0\hspace{-6pt}} \ar@/_.8cm/[uull]_(.2){j_0} |(.69)\hole
		}	$
		\caption{Picture for the proof of \Cref{lem:indep-global-left}: switch of $\dder{D}_0$ and $\dder{D}_1$.}
		\label{fi:first}
	\end{figure}

	\begin{figure}[t]
		\centering 
		$\xymatrix@C=15pt@R=15pt{&&R_1 \ar@/_.8cm/[ddrr]_(.2){a_0}|(.7)\hole
			\ar[dr]^{h_1}&& L_2\ar@/^.8cm/[ddll]^(.2){a_1}
			\ar[dl]_{m_2}\\&K_1\ar[dr]^{k_1}\ar@{>->}[dl]_{l_1}
			\ar@/_.6cm/[ddrr]|(.36)\hole^(.65){u_1}\ar[ur]^{r_1}&& G_2 &&
			K_2
			\ar@/^.6cm/[ddll]|(.36)\hole_(.65){u_2}\ar[dl]_{k_2}\ar@{>->}[lu]_{l_2}
			\ar[dr]^{r_2}\\L_1
			\ar@/_.6cm/[ddrr]_(.2){b_0}|(.31)\hole|(.81)\hole
			\ar[dr]^(.4){m_1}|(.61)\hole && D_1
			\ar@{>->}[dl]|(.4)\hole_(.65){f_1}\ar[ur]|(.5)\hole^(.7){g_1}&&D_2
			\ar[dr]|(.4)\hole^(.65){g_2}
			\ar@{>->}[ul]|(.5)\hole_(.65){f_2}&&R_2\ar@/^.6cm/[ddll]^(.2){b_1}|(.31)\hole|(.81)\hole\ar[dl]_(.4){h_2}|(.61)\hole\\&G_1
			&&P_2\ar[dr]_{d_1}
			\ar@{>->}[dl]^{d_0}\ar[ur]^(.4){c_1}\ar@{>->}[ul]_(.4){c_0}&&G_3\\L_2
			\ar@/^.6cm/[uurr]^(.2){a_1} \ar[ur]_(.35){f_1\circ
				a_1}|(.61)\hole&&Q_2
			\ar@{>->}[ul]|(.4)\hole_(.65){x_1}\ar[dr]|(.5)\hole^(.65){y_1}
			&&Q_3\ar[ur]|(.4)\hole^(.65){y_2} \ar@{>->}[dl]|(.5)\hole_(.65){x_2}
			&& R_1 \ar[ul]^(.3){g_2\circ
				a_0}|(.61)\hole\ar@/_.6cm/[uull]_(.2){a_0}\\&K_2
			\ar[ur]_{d_0\circ u_2} \ar[dr]_{r_2}
			\ar@{>->}[ul]^{l_2}\ar@/^.6cm/[uurr]_(.65){u_2}&&G'_2&& K_1
			\ar[ur]_{r_1} \ar[ul]^{d_1\circ u_1} \ar@{>->}[dl]^{l_1}
			\ar@/_.6cm/[uull]^(.65){u_1}\\&& R_2
			\ar[ur]_{\hspace{-6pt}x_2\circ
				b_1}\ar@/^.8cm/[uurr]^(.2){b_1}&& L_1 \ar[ul]^{y_1\circ
				b_0\hspace{-7pt}} \ar@/_.8cm/[uull]_(.2){b_0} |(.69)\hole
		}	$		
		\caption{Picture for the proof of \Cref{lem:indep-global-left}: switch of $\dder{D}_1$ and $\dder{D}_2$.}
		\label{fi:second}
	\end{figure}

	\begin{figure}
		\centering
		$\xymatrix@C=15pt@R=15pt{&&R_0 \ar@/_.8cm/[ddrr]_(.2){e_0}|(.7)\hole
			\ar[dr]^{\hspace{-5pt}h_0}&& L_2\ar@/^.8cm/[ddll]^(.2){e_1}
			\ar[dl]_{f_1\circ a_1\hspace{-5pt}}\\&K_0\ar[dr]^{k_0}\ar@{>->}[dl]_{l_0}
			\ar@/_.6cm/[ddrr]|(.36)\hole^(.65){w_0}\ar[ur]^{r_0}&& G_1 &&
			K_2 \ar@/^.6cm/[ddll]|(.36)\hole_(.65){w_2}\ar[dl]_{d_0\circ
				u_2\hspace{-5pt}}\ar@{>->}[lu]_{l_2} \ar[dr]^{r_2}\\L_0
			\ar@/_.6cm/[ddrr]_(.2){o_0}|(.31)\hole|(.81)\hole
			\ar[dr]^(.4){m_0}|(.61)\hole && D_0
			\ar@{>->}[dl]|(.4)\hole_(.65){f_0}\ar[ur]|(.5)\hole^(.7){g_0}&&Q_2
			\ar[dr]|(.4)\hole^(.65){y_1}
			\ar@{>->}[ul]|(.5)\hole_(.65){x_1}&&R_2\ar@/^.6cm/[ddll]^(.2){o_1}|(.31)\hole|(.81)\hole\ar[dl]_(.3){x_2\circ
				b_1}|(.58)\hole \\&G_0 &&P_3\ar[dr]_{n_1}
			\ar@{>->}[dl]^{n_0}\ar[ur]^(.4){\theta_1}\ar@{>->}[ul]_(.4){\theta_0}&&G'_2\\L_2
			\ar@/^.6cm/[uurr]^(.2){e_1} \ar[ur]_(.35){f_0\circ
				e_1}|(.61)\hole&&Q_4
			\ar@{>->}[ul]|(.4)\hole_(.65){z_0}\ar[dr]|(.5)\hole^(.65){z'_0}
			&&Q_5\ar[ur]|(.4)\hole^(.65){z'_1}
			\ar@{>->}[dl]|(.5)\hole_(.65){z_1} && R_0 \ar[ul]^(.3){y_1\circ
				e_0}|(.6)\hole\ar@/_.6cm/[uull]_(.2){e_0}\\&K_2
			\ar[ur]_{\hspace{-2pt}n_0\circ w_2} \ar[dr]_{r_2}
			\ar@{>->}[ul]^{l_2}\ar@/^.6cm/[uurr]_(.65){w_2}&&G'_1&& K_0
			\ar[ur]_{r_0} \ar[ul]^{n_1\circ w_0} \ar[dl]^{l_0}
			\ar@/_.6cm/[uull]^(.65){w_0}\\&& R_2
			\ar[ur]_{\hspace{-4pt}z_1\circ
				o_1}\ar@/^.8cm/[uurr]^(.2){o_1}&& L_0 \ar[ul]^{z'_0\circ
				o_0\hspace{-7pt}} \ar@/_.8cm/[uull]_(.2){o_0} |(.69)\hole
		}	$
		\caption{Picture for the proof of \Cref{lem:indep-global-left}: switch of $\dder{D}_0$ and $S_{a_0,a_1}(\dder{D}_2)$.}
		\label{fi:third}
	\end{figure}
	
	\begin{enumerate}
		\item We have to construct the two dotted arrows in the diagram
		below
		\[\xymatrix@C=15pt{L_0 \ar[d]_{z'_0\circ o_0}&& K_0
			\ar[d]_{n_1\circ w_0}\ar@{>->}[ll]_{l_0} \ar[r]^{r_0} & R_0
			\ar@{.>}@/^.35cm/[drrr]_(.4){\beta_0}|(.29)\hole
			\ar[dr]|(.28)\hole_{y_1\circ e_0} && L_1
			\ar@{.>}@/_.35cm/[dlll]^(.4){\beta_1}
			\ar[dl]|(.28)\hole^{y_1\circ b_0}& K_1 \ar[d]^{d_1\circ
				u_1}\ar@{>->}[l]_{l_1} \ar[rr]^{r_1} && R_1 \ar[d]^{g_2\circ
				a_1}\\G'_1 && \ar@{>->}[ll]^{z_1} Q_5 \ar[rr]_{z'_1}&& G'_2 &&
			\ar@{>->}[ll]^{x_2} Q_3 \ar[rr]_{y_2}&& G_3}\]
		
		The arrows $i_0\colon R_0\to D_1$,
		$e_0\colon R_0\to Q_2$, $i_1\colon L_1\to D_0$ and $b_0\colon L_1 \to Q_2$ satisfy
		\[
		f_1\circ i_0  = h_0 =x_1\circ e_0 \qquad 
		g_0\circ i_1 = m_1 = x_1 \circ b_0 \]
		
		These identities entail the existence of the dotted
		$\beta'_0\colon R_0\to P_2$ and $\beta'_1\colon L_1\to P_3$ below
		
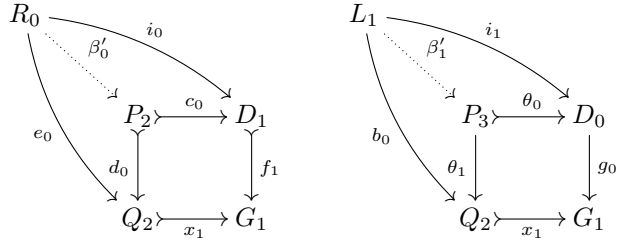
\[\xymatrix{R_0 \ar@{.>}[dr]^{\beta'_0} \ar@/^.3cm/[drr]^{i_0}
			\ar@/_.3cm/[ddr]_{e_0} &&& L_1 \ar@{.>}[dr]^{\beta'_1} \ar@/^.3cm/[drr]^{i_1}
			\ar@/_.3cm/[ddr]_{b_0}\\ &P_2 \ar@{>->}[r]^{c_0} \ar@{>->}[d]_{d_0}& D_1
			\ar@{>->}[d]^{f_1} & &P_3 \ar@{>->}[r]^{\theta_0} \ar[d]_{\theta_1}& D_0 \ar[d]^{g_0}\\ &Q_2\ar@{>->}[r]_{x_1} & G_1 & &Q_2\ar@{>->}[r]_{x_1} & G_1}\]
		
                    We can now define $\beta_0\colon R_0\to Q_3$ as $d_1\circ \beta'_1$ and $\beta_1\colon L_1\to Q_5$ be $n_1\circ \beta'_1$, so that
		\begin{align*}
			x_2\circ \beta_0 & =x_2\circ d_2\circ \beta'_0 = y_1\circ d_0\circ \beta'_0=y_1\circ e_0\\
			z'_1 \circ \beta_1 & = z'_1 \circ n_1\circ \beta'_1 =y_1\circ \theta_1\circ \beta'_1=y_1\circ b_0
		\end{align*}
		
		Therefore, $(\beta_0, \beta_1)$ is the wanted independence pair.
		
              \item Here we assume that
                $\dder{D}_1$ and $\dder{D}_2$ are switchable  with an
    independence pair $(\alpha_0, \alpha_1)$ and that
    $\dder{D}_0$ and $S_{\alpha_0, \alpha_1}(\dder{D}_2)$ are switchable.
                Hence, in addition to the three diagrams above, we have a fourth one, depicted in \Cref{fi:fourth}.
		\begin{figure}[h]
			\centering
			$\xymatrix@C=15pt@R=15pt{&&R_0 \ar@/_.8cm/[ddrr]_(.2){\alpha_0}|(.7)\hole \ar[dr]^{\hspace{-1pt}g_1\circ i_0}&& L_2\ar@/^.8cm/[ddll]^(.2){\alpha_1}  \ar[dl]_{m_2\hspace{-2pt}}\\
				&K_0\ar[dr]^{\hspace{-4pt}q_1\circ v_0}\ar@{>->}[dl]_{l_0} \ar@/_.6cm/[ddrr]|(.36)\hole^(.65){\phi_0}\ar[ur]^{r_0}&& G_2 && K_2 \ar@/^.6cm/[ddll]|(.36)\hole_(.65){\phi_2}\ar[dl]_{k_2}\ar@{>->}[lu]_{l_2} \ar[dr]^{r_2}\\
				L_0 \ar@/_.6cm/[ddrr]_(.2){\gamma_0}|(.31)\hole|(.81)\hole \ar[dr]^(.4){\hspace{-4pt}t_0\circ j_0}|(.6)\hole && Q_1 \ar@{>->}[dl]|(.4)\hole_(.65){s_1}\ar[ur]|(.5)\hole^(.7){t_1}&&D_2 \ar[dr]|(.4)\hole^(.65){g_2} \ar@{>->}[ul]|(.5)\hole_(.65){f_2}&&R_2\ar@/^.6cm/[ddll]^(.2){\gamma_1}|(.31)\hole|(.81)\hole\ar[dl]_(.35){h_2}|(.6)\hole\\&H'_1 &&P_4\ar[dr]_{\lambda_1} \ar@{>->}[dl]^{\lambda_0}\ar[ur]^(.4){\zeta_1}\ar@{>->}[ul]_(.4){\zeta_0}&&G_3\\
				L_2 \ar@/^.6cm/[uurr]^(.2){\alpha_1} \ar[ur]_(.35){\hspace{-2pt}s_1\circ
					\alpha_1}|(.61)\hole&&Q_6
				\ar@{>->}[ul]|(.4)\hole_(.6){\xi_0}\ar[dr]|(.5)\hole^(.65){\xi'_0}
				&&Q_7\ar[ur]|(.4)\hole^(.6){\xi'_1}
				\ar@{>->}[dl]|(.5)\hole_(.65){\xi_1} && R_0 \ar[ul]^(.35){g_2\circ
					\alpha_0}|(.61)\hole\ar@/_.6cm/[uull]_(.2){\alpha_0}\\&K_2
				\ar[ur]_{\lambda_0\circ \phi_2} \ar[dr]_{r_2}
				\ar@{>->}[ul]^{l_2}\ar@/^.6cm/[uurr]_(.65){\phi_2}&&H'_2&& K_0
				\ar[ur]_{r_0} \ar[ul]^{\lambda_1\circ \phi_0} \ar@{>->}[dl]^{l_0}
				\ar@/_.6cm/[uull]^(.65){\phi_0}\\&& R_2
				\ar[ur]_{\hspace{-5pt}\xi_1\circ
					\gamma_1}\ar@/^.8cm/[uurr]^(.2){\gamma_1}&& L_0
				\ar[ul]^{\xi'_0\circ \gamma_0\hspace{-7pt}}
				\ar@/_.8cm/[uull]_(.2){\gamma_0} |(.69)\hole 
			}	$
			\caption{Picture for the proof of \Cref{lem:indep-global-left}: switch of $S_{i_0, i_1}(\dder{D}_0)$ and $\dder{D}_2$.}
			\label{fi:fourth}
		\end{figure}
		
		Our aim is to construct the dotted arrows in the following diagram
		
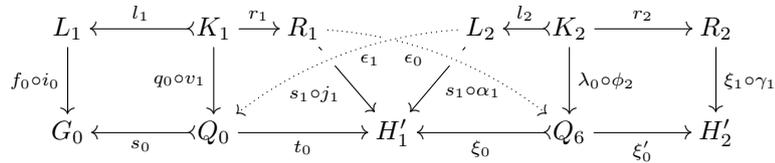
\[\xymatrix@C=15pt{L_1 \ar[d]_{f_0\circ i_0}&& K_1
			\ar[d]_{q_0\circ v_1}\ar@{>->}[ll]_{l_1} \ar[r]^{r_1} & R_1
			\ar@{.>}@/^.35cm/[drrr]_(.4){\epsilon_0}|(.29)\hole
			\ar[dr]|(.28)\hole_{s_1\circ j_1} && L_2
			\ar@{.>}@/_.35cm/[dlll]^(.4){\epsilon_1}
			\ar[dl]|(.28)\hole^{s_1\circ \alpha_1}& K_2
			\ar[d]^{\lambda_0\circ \phi_2}\ar@{>->}[l]_{l_2} \ar[rr]^{r_2} &&
			R_2 \ar[d]^{\xi_1\circ \gamma_1}\\G_0 && \ar@{>->}[ll]^{s_0} Q_0
			\ar[rr]_{t_0}&& H'_1 && \ar@{>->}[ll]^{\xi_0} Q_6
			\ar[rr]_{\xi'_0}&& H'_2}\]
		
		We can take the arrows $j_1\colon R_1\to Q_1$, 
		$a_0\colon R_1\to D_2$, $e_1\colon L_2 \to D_0$ and $a_1\colon L_2\to D_1$ and notice that
		\[	f_2\circ a_0 =h_1 =t_1\circ j_1 \qquad   g_0\circ e_1=f_1\circ a_1\]
		and thus we get the dotted arrows $\epsilon'_0\colon R_1\to P_4$, $\epsilon'_1\colon L_2\to P:1$ 
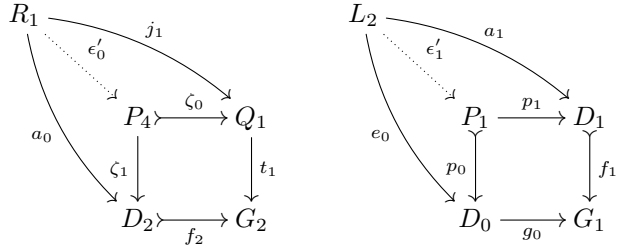
		\[\xymatrix{R_1 \ar@{.>}[dr]^{\epsilon'_0} 
			\ar@/^.3cm/[drr]^{j_1} \ar@/_.3cm/[ddr]_{a_0}&&& L_2 \ar@{.>}[dr]^{\epsilon'_1}
			\ar@/^.3cm/[drr]^{a_1} \ar@/_.3cm/[ddr]_{e_0}\\ &P_4
			\ar@{>->}[r]^{\zeta_0} \ar[d]_{\zeta_1}& Q_1 \ar[d]^{t_1} &&P_1
			\ar[r]^{p_1} \ar@{>->}[d]_{p_0}& D_1 \ar@{>->}[d]^{f_1}\\
			&D_2\ar@{>->}[r]_{f_2} & G_2& 	&D_0\ar[r]_{g_0} & G_1} \]
		
		Now, notice that
		\[
		t_1\circ q_1\circ \epsilon'_1  =g_1\circ p_1\circ \epsilon'_1  =g_1\circ a_1 =m_2
		\]
		Hence $(\alpha_0, q_1\circ \epsilon'_1)$ is an independence pair
		for $S_{i_0, i_1}(\dder{D}_0)$ and $\dder{D}_2$. By hypothesis $(\X, \R)$ is well-switching and therefore
		$q_1\circ \epsilon'_1$ must coincide with $\alpha_1$.

		We can then define $\epsilon_0\colon R_1\to Q_6$ and $\epsilon_1\colon L_2\to Q_0$ to be, respectively, 
		$\lambda_0\circ \epsilon'_0$ and $q_0\circ
		\epsilon'_1$. Then for these  arrows we have  identities
		\begin{align*}
			\xi_0\circ \epsilon_0 & =\xi_0\circ \lambda_0\circ \epsilon'_0 =s_1\circ \zeta_0 \circ \epsilon'_0=s_1\circ j_1\\	t_0\circ \epsilon_1 & = t_0\circ q_0\circ \epsilon'_1 =s_1\circ q_1\circ \epsilon'_1=s_1\circ \alpha_1
		\end{align*}
		Allowing us to conclude that
		$S_{i_0,i_1}(\dder{D}_1)$ and $S_{\alpha_0,
			\alpha_1}(\dder{D}_2)$ are switchable.  \qedhere
	\end{enumerate}
\end{proof}

\lemSwitchConfluence*
\label{lemSwitchConfluence-proof}

\begin{proof}
	To fix the notation, let us depict the derivation $\der{D}=\{\dder{D}_i\}_{i=0}^2$
	
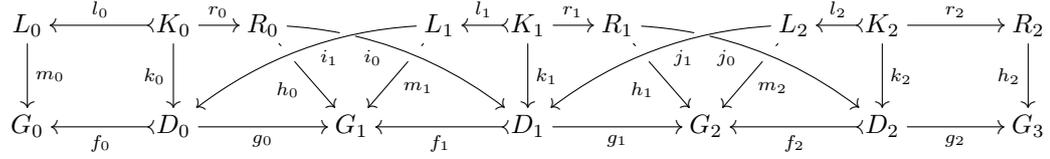
\[\xymatrix@C=15pt{L_0\ar[d]^{m_0}&& K_0 \ar[d]_{k_0}\ar@{>->}[ll]_{l_0} \ar[r]^{r_0} & R_0 \ar@/^.35cm/[drrr]_(.4){i_0}|(.29)\hole \ar[dr]|(.28)\hole_{h_0} && L_1 \ar@/_.35cm/[dlll]^(.4){i_1} \ar[dl]|(.28)\hole^{m_1}& K_1 \ar[d]^{k_1}\ar@{>->}[l]_{l_1} \ar[r]^{r_1} & R_1 \ar[dr]|(.28)\hole_{h_1} \ar@/^.35cm/[drrr]_(.4){j_0}|(.29)\hole  && L_2 \ar@/_.35cm/[dlll]^(.4){j_1} \ar[dl]|(.28)\hole^{m_2}& K_2 \ar[d]^{k_2}\ar@{>->}[l]_{l_2} \ar[rr]^{r_2} && R_2 \ar[d]_{h_2} \\G_0 && \ar@{>->}[ll]^{f_0} D_0 \ar[rr]_{g_0}&& G_1  && \ar@{>->}[ll]^{f_1} D_1 \ar[rr]_{g_1}&& G_2 && \ar@{>->}[ll]^{f_2} D_2 \ar[rr]_{g_2}&& G_3 }\]
	
	We are going to do the two sequences of switches and compare the results.
	
	\smallskip \noindent First sequence of switches.
	
	We begin switching $\dder{D}_0$ and $\dder{D}_1$ to get the following diagram
	\[\xymatrix@C=15pt{L_1\ar[d]^{f_0\circ i_1}&& K_1 \ar[d]_{k'_1}\ar@{>->}[ll]_{l_1} \ar[r]^{r_1} & R_1 \ar@/^.35cm/[drrr]_(.4){a_0}|(.29)\hole \ar[dr]|(.28)\hole_{h'_1} && L_0 \ar@/_.35cm/[dlll]^(.4){a_1} \ar[dl]|(.28)\hole^{m'_0}& K_0 \ar[d]^{k'_0}\ar@{>->}[l]_{l_0} \ar[r]^{r_0} & R_0 \ar[dr]|(.28)\hole_{g_1\circ i_0} \ar@/^.35cm/[drrr]_(.4){b_0}|(.29)\hole  && L_2 \ar@/_.35cm/[dlll]^(.4){b_1} \ar[dl]|(.28)\hole^{m_2}& K_2 \ar[d]^{k_2}\ar@{>->}[l]_{l_2} \ar[rr]^{r_2} && R_2 \ar[d]_{h_2} \\G_0 && \ar@{>->}[ll]^{f'_0} D'_0 \ar[rr]_{g'_0}&& G'_1  && \ar@{>->}[ll]^{f'_1} D'_1 \ar[rr]_{g'_1}&& G_2 && \ar@{>->}[ll]^{f_2} D_2 \ar[rr]_{g_2}&& G_3 }\]

	By \Cref{def:switch} we also know that $m_0=f'_0\circ a_1$ and $h_1=g'_1\circ a_0$. Now, let us switch $S_{i_0, i_1}(\dder{D}_0)$ and $\dder{D}_2$. In this case we obtain
	
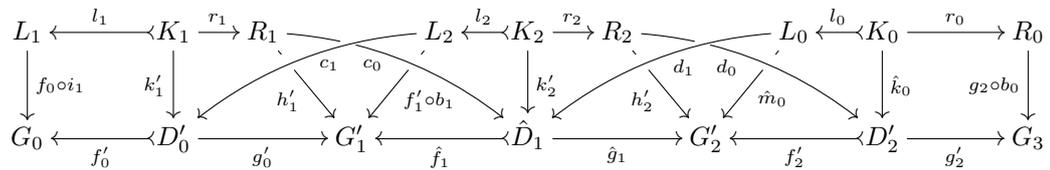
\[\xymatrix@C=15pt{L_1\ar[d]^{f_0\circ i_1}&& K_1 \ar[d]_{k'_1}\ar@{>->}[ll]_{l_1} \ar[r]^{r_1} & R_1 \ar@/^.35cm/[drrr]_(.4){c_0}|(.29)\hole \ar[dr]|(.28)\hole_{h'_1} && L_2 \ar@/_.35cm/[dlll]^(.4){c_1} \ar[dl]|(.28)\hole^{f'_1\circ b_1}& K_2 \ar[d]^{k'_2}\ar@{>->}[l]_{l_2} \ar[r]^{r_2} & R_2 \ar[dr]|(.28)\hole_{h'_2} \ar@/^.35cm/[drrr]_(.4){d_0}|(.29)\hole  && L_0 \ar@/_.35cm/[dlll]^(.4){d_1} \ar[dl]|(.28)\hole^{\hat{m}_0}& K_0 \ar[d]^{\hat{k}_0}\ar@{>->}[l]_{l_0} \ar[rr]^{r_0} && R_0 \ar[d]_{g_2\circ b_0} \\G_0 && \ar@{>->}[ll]^{f'_0} D'_0 \ar[rr]_{g'_0}&& G'_1  && \ar@{>->}[ll]^{\hat{f}_1} \hat{D}_1 \ar[rr]_{\hat{g}_1}&& G'_2 && \ar@{>->}[ll]^{f'_2} D'_2 \ar[rr]_{g'_2}&& G_3 }\]
	Using again the definition of switch, we know that $m'_0=\hat{f}_1\circ d_1$ and $h_2=g'_2\circ d_0$. Moreover, by hypothesis $(\X, \R)$ is well-switching and the hypothesis of the second point of \Cref{lem:indep-global-left} are satisfied by $\der{D}$. Thus we can characterise $c_0\colon R_1\to \hat{D}_1$ and $c_1\colon L_2\to D'_0$ as the unique arrows fitting in the diagrams below, where the bottom squares are pullbacks
	\[\xymatrix@R=15pt@C=15pt{&&R_1 \ar@/_1.2cm/[dddl]_(.4){a_0}|(.63)\hole\ar[dl]^{j_0}\ar[dd]^{c'_0} \ar@{>}[dr]^{c_0}&&&&&L_2 \ar[dl]^{j_1} \ar@/_1.2cm/[dddl]_(.4){z_1}|(.63)\hole\ar[dd]^{c'_1} \ar@{>}[dr]^{c_1}\\&D_2 \ar@{>->}[dl]^(.4){f_2}&&\hat{D}_1 \ar[dr]^{\hat{g}_1}  &&&D_1 \ar@{>->}[dl]^(.4){f_1}&&D'_0 \ar[dr]^{g'_0}\\G_2 && P\ar[dr]_{p_4} \ar@{>->}[ur]_{p_3}\ar@{>->}[dl]^{p_1} \ar[ul]^{p_2}&&G'_2&G_1 && Q\ar[dr]_{q_4} \ar@{>->}[ur]_{q_3}\ar@{>->}[dl]^{q_1} \ar[ul]^{q_2}&&G'_1\\& D'_1 \ar[ul]^{g'_1}&&D'_2 \ar@{>->}[ur]_{f'_2}&&& D_0 \ar[ul]^{g_0}&&D'_1 \ar@{>->}[ur]_{f'_1}}\]
	
Finally, we switch the first two direct derivations
	\[\xymatrix@C=15pt{L_2\ar[d]^{f'_0\circ c_1}&& K_2 \ar[d]_{\hat{k}_2}\ar@{>->}[ll]_{l_2} \ar[r]^{r_2} & R_2 \ar@/^.35cm/[drrr]_(.4){e_0}|(.29)\hole \ar[dr]|(.28)\hole_{\hat{h}_2} && L_1 \ar@/_.35cm/[dlll]^(.4){e_1} \ar[dl]|(.28)\hole^{\hat{m}_1}& K_1 \ar[d]^{\hat{k}_1}\ar@{>->}[l]_{l_1} \ar[r]^{r_1} & R_1 \ar@/^.35cm/[drrr]_(.4){t_0}|(.29)\hole \ar[dr]|(.28)\hole_{\hat{g}_1\circ c_0}   && L_0  \ar[dl]|(.28)\hole^{\hat{m}_0} \ar@/_.35cm/[dlll]^(.4){t_1}& K_0 \ar[d]^{\hat{k}_0}\ar@{>->}[l]_{l_0} \ar[rr]^{r_0} && R_0 \ar[d]_{g_2\circ b_0} \\G_0 && \ar@{>->}[ll]^{\hat{f}_0} \hat{D}_0 \ar[rr]_{\hat{g}_0}&& \hat{G}_1  && \ar@{>->}[ll]^{\tilde{f}_1} \tilde{D}_1 \ar[rr]_{\tilde{g}_1}&& G'_2 && \ar@{>->}[ll]^{f'_2} D'_2 \ar[rr]_{g'_2}&& G_3 }\]
	As before, we have identities $f_0\circ i_1=\hat{f}_0\circ e_1$ and $h'_2=\tilde{g}_1\circ e_0$. Now, the derivation $S_{i_0, i_1}(\dder{D}_1)\cdot S_{i_0, i_1}(\dder{D}_0)\cdot \dder{D}_2$ satisfies the hypothesis of the first point of \Cref{lem:indep-global-left}. Thus, since there is at most one independence pair between two direct derivations, we know that  $t_0\colon R_1\to \tilde{D}_1$ and $t_1\colon L_0\to D'_2$  fit in
	\[\xymatrix{&R_1 \ar[dl]_{a_0} \ar@{>}[dr]^{t_0} \ar@/^.4cm/[dd]|\hole^(.6){c_0}\ar[d]_{t'_0}&&L_ 0\ar[d]_{d_1} \ar@/_.4cm/[ddrr]|(.41)\hole_(.7){a_1}\ar[drr]^{t'_1} \ar@{>}[rr]^{t_1}&&\tilde{D}_1 \ar@{>->}[r]^{\tilde{f}_1}&\hat{G}_1\\ D'_1 \ar@{>->}[d]_{f'_1}&P\ar@{>->}[d]_{p_3} \ar[r]^{p_4} \ar@{>->}[l]_{p_1}&D'_2\ar@{>->}[d]^{f'_2}&\hat{D}_1 \ar@{>->}[d]_{\hat{f}_1}&&S\ar@{>->}[r]_{s_4} \ar[u]_{s_3}\ar@{>->}[d]^{s_1} \ar[ll]_{s_2} &\hat{D}_0 \ar[u]_{\hat{g}_0} \ar@{>->}[d]^{\hat{f}_0}\\G'_1&\hat{D}_1 \ar@{>->}[l]^{\hat{f}_1}\ar[r]_{\hat{g}_1}& G'_2&G'_1 && D'_0 \ar[ll]^{g'_0}  \ar@{>->}[r]_{f'_0}& G_0 }\]
	
	\smallskip \noindent Second sequence of switches.
	
	This time we start switching $\dder{D}_1$ and $\dder{D}_2$.
	
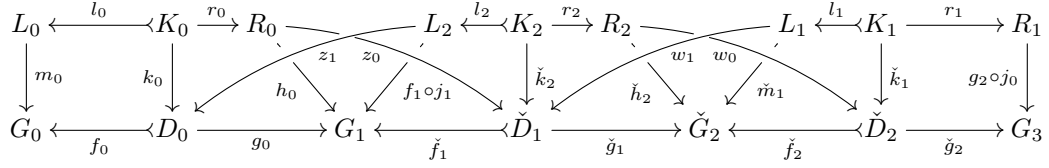
\[\xymatrix@C=15pt{L_0\ar[d]^{m_0}&& K_0 \ar[d]_{k_0}\ar@{>->}[ll]_{l_0} \ar[r]^{r_0} & R_0 \ar@/^.35cm/[drrr]_(.4){z_0}|(.29)\hole \ar[dr]|(.28)\hole_{h_0} && L_2 \ar@/_.35cm/[dlll]^(.4){z_1} \ar[dl]|(.28)\hole^{f_1\circ j_1}& K_2 \ar[d]^{\check{k}_2}\ar@{>->}[l]_{l_2} \ar[r]^{r_2} & R_2 \ar[dr]|(.28)\hole_{\check{h}_2} \ar@/^.35cm/[drrr]_(.4){w_0}|(.29)\hole  && L_1 \ar@/_.35cm/[dlll]^(.4){w_1} \ar[dl]|(.28)\hole^{\check{m}_1}& K_1 \ar[d]^{\check{k}_1}\ar@{>->}[l]_{l_1} \ar[rr]^{r_1} && R_1\ar[d]_{g_2\circ j_0} \\G_0 && \ar@{>->}[ll]^{f_0} D_0 \ar[rr]_{g_0}&& G_1  && \ar@{>->}[ll]^{\check{f}_1} \check{D}_1 \ar[rr]_{\check{g}_1}&& \check{G}_2 && \ar@{>->}[ll]^{\check{f}_2} \check{D}_2 \ar[rr]_{\check{g}_2}&& G_3 }\]
	As in the previous case, we already know that $m_1=\check{f}_1\circ w_1$ and $h_2=\check{g}_2\circ w_0$.
	
	Next, we switch $\dder{D}_0$ and $S_{j_0, j_1}(\dder{D}_2)$ getting
	
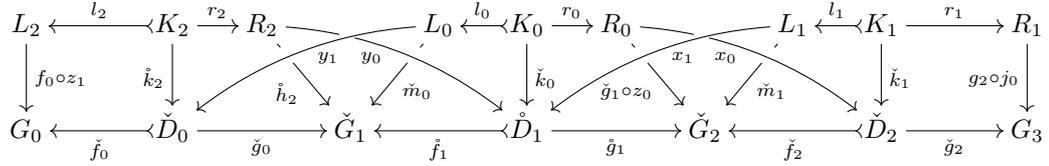
\[\xymatrix@C=15pt{L_2\ar[d]^{f_0\circ z_1}&& K_2 \ar[d]_{\mathring{k}_2}\ar@{>->}[ll]_{l_2} \ar[r]^{r_2} & R_2 \ar@/^.35cm/[drrr]_(.4){y_0}|(.29)\hole \ar[dr]|(.28)\hole_{\mathring{h}_2} && L_0 \ar@/_.35cm/[dlll]^(.4){y_1} \ar[dl]|(.28)\hole^{\check{m}_0}& K_0 \ar[d]^{\check{k}_0}\ar@{>->}[l]_{l_0} \ar[r]^{r_0} & R_0 \ar[dr]|(.28)\hole_{\check{g}_1\circ z_0} \ar@/^.35cm/[drrr]_(.4){x_0}|(.29)\hole  && L_1 \ar@/_.35cm/[dlll]^(.4){x_1} \ar[dl]|(.28)\hole^{\check{m}_1}& K_1 \ar[d]^{\check{k}_1}\ar@{>->}[l]_{l_1} \ar[rr]^{r_1} && R_1\ar[d]_{g_2\circ j_0} \\G_0 && \ar@{>->}[ll]^{\check{f}_0} \check{D}_0 \ar[rr]_{\check{g}_0}&& \check{G}_1  && \ar@{>->}[ll]^{\mathring{f}_1} \mathring{D}_1 \ar[rr]_{\mathring{g}_1}&& \check{G}_2 && \ar@{>->}[ll]^{\check{f}_2} \check{D}_2 \ar[rr]_{\check{g}_2}&& G_3 }\]
	The usual identities $m_0=\check{f}_0\circ y_1$, $\check{h}_2=\mathring{g}_1\circ y_0$ hold. But the really important thing to notice is that the first match of these derivation coincide with $f'_0\circ c_1$
	\[	f_0'\circ c_1  =f'_0\circ q_3\circ c'_1 =f_0\circ q_1\circ c'_1=f_0\circ z_1\]
	
	Moreover, we can characterise the pair $(x_0, x_1)$. If we apply the first point of \Cref{lem:indep-global-left} to $\dder{D}$ we get an independence pair between the last two steps of the derivation above, which, since $(\X, \R)$ is well-switching, must coincide with $(x_0, x_1)$. Thus we have two diagrams
	
	
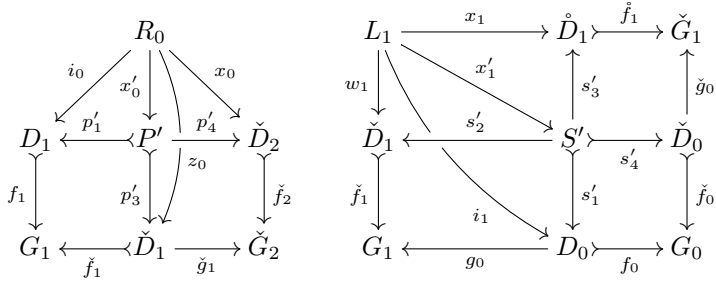
\[\xymatrix{&R_0 \ar[dl]_{i_0} \ar@{>}[dr]^{x_0} \ar@/^.4cm/[dd]|\hole^(.6){z_0}\ar[d]_{x'_0}&&L_ 1\ar[d]_{w_1} \ar@/_.4cm/[ddrr]|(.41)\hole_(.7){i_1}\ar[drr]^{x'_1} \ar@{>}[rr]^{x_1}&&\mathring{D}_1 \ar@{>->}[r]^{\mathring{f}_1}&\check{G}_1\\ D_1 \ar@{>->}[d]_{f_1}&P'\ar@{>->}[d]_{p'_3} \ar[r]^{p'_4} \ar@{>->}[l]_{p'_1}&\check{D}_2\ar@{>->}[d]^{\check{f}_2}&\check{D}_1 \ar@{>->}[d]_{\check{f}_1}&&S'\ar@{>->}[r]_{s'_4} \ar[u]_{s'_3}\ar@{>->}[d]^{s'_1} \ar[ll]_{s'_2} &\check{D}_0 \ar[u]_{\check{g}_0} \ar@{>->}[d]^{\check{f}_0}\\G_1&\check{D}_1 \ar@{>->}[l]^{\check{f}_1}\ar[r]_{\check{g}_1}& \check{G}_2&G_1 && D_0 \ar[ll]^{g_0}  \ar@{>->}[r]_{f_0}& G_0 }\]
	From this diagram we can recover another identity. Notice that, if $p'_2\colon P'\to D_2$ is the last projection from $P'$, then
	\[f_2\circ p'_2\circ x'_0=g_1\circ p'_1\circ x'_0=g_1\circ i_0=f_2\circ b_0\]
	Implying, since $f_2$ is mono, that $p'_2\circ x'_0=b_0$. From this, we can also deduce that 
	\[\check{g}_2\circ x_0=\check{g}_2\circ p'_4\circ x'_0=g_2\circ p'_2\circ x'_0=g_2\circ b_0 \]

	Finally, switching the second and the third step we get
	
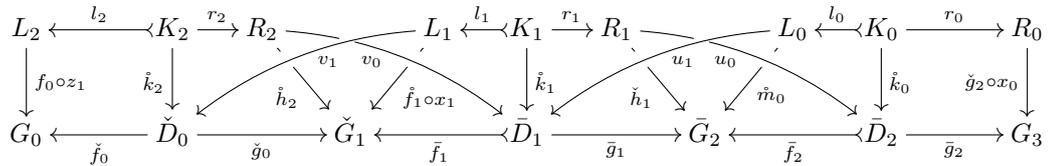
\[\xymatrix@C=15pt{L_2\ar[d]^{f_0\circ z_1}&& K_2 \ar[d]_{\mathring{k}_2}\ar@{>->}[ll]_{l_2} \ar[r]^{r_2} & R_2 \ar@/^.35cm/[drrr]_(.4){v_0}|(.29)\hole \ar[dr]|(.28)\hole_{\mathring{h}_2} && L_1 \ar@/_.35cm/[dlll]^(.4){v_1} \ar[dl]|(.28)\hole^{\mathring{f}_1 \circ x_1}& K_1 \ar[d]^{\mathring{k}_1}\ar@{>->}[l]_{l_1} \ar[r]^{r_1} & R_1 \ar@/^.35cm/[drrr]_(.4){u_0}|(.29)\hole\ar[dr]_{\check{h}_1}|(.28)\hole   && L_0 \ar@/_.35cm/[dlll]^(.4){u_1} \ar[dl]^{\mathring{m}_0}|(.28)\hole& K_0 \ar[d]^{\mathring{k}_0}\ar@{>->}[l]_{l_0} \ar[rr]^{r_0} && R_0\ar[d]_{\check{g}_2\circ x_0} \\G_0 && \ar[ll]^{\check{f}_0} \check{D}_0 \ar[rr]_{\check{g}_0}&& \check{G}_1  && \ar@{>->}[ll]^{\bar{f}_1} \bar{D}_1 \ar[rr]_{\bar{g}_1}&& \bar{G}_2 && \ar@{>->}[ll]^{\bar{f}_2} \bar{D}_2 \ar[rr]_{\bar{g}_2}&& G_3 }\]
	As usual, this last switch brings identities $\check{m}_0=\bar{f}_1\circ u_1$ and $g_2\circ j_0=\bar{g}_2\circ u_0$. Applying the second point of \Cref{lem:indep-global-left} to $\dder{D}_0\cdot S_{j_0, j_1}(\dder{D}_2)\cdot S_{j_0, j_1}(\dder{D}_1)$ and the fact that $(\X, \R)$ is well-switching we get the following diagrams
	
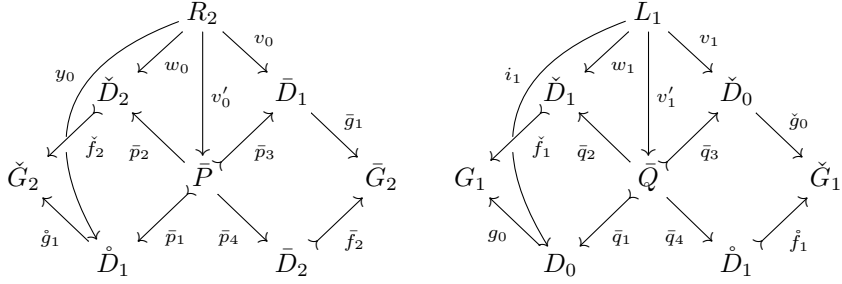
\[\xymatrix@R=15pt@C=15pt{&&R_2 \ar@/_1.2cm/[dddl]_(.4){y_0}|(.63)\hole\ar[dl]^{w_0}\ar[dd]^{v'_0} \ar@{>}[dr]^{v_0}&&&&&L_1 \ar[dl]^{w_1} \ar@/_1.2cm/[dddl]_(.4){i_1}|(.63)\hole\ar[dd]^{v'_1} \ar@{>}[dr]^{v_1}\\&\check{D}_2 \ar@{>->}[dl]^(.4){\check{f}_2}&&\bar{D}_1 \ar[dr]^{\bar{g}_1}  &&&\check{D}_1 \ar@{>->}[dl]^(.4){\check{f}_1}&&\check{D}_0 \ar[dr]^{\check{g}_0}\\\check{G}_2 && \bar{P}\ar[dr]_{\bar{p}_4} \ar@{>->}[ur]_{\bar{p}_3}\ar@{>->}[dl]^{\bar{p}_1} \ar[ul]^{\bar{p}_2}&&\bar{G}_2&G_1 && \bar{Q}\ar[dr]_{\bar{q}_4} \ar@{>->}[ur]_{\bar{q}_3}\ar@{>->}[dl]^{\bar{q}_1} \ar[ul]^{\bar{q}_2}&&\check{G}_1\\& \mathring{D}_1 \ar[ul]^{\mathring{g}_1}&&\bar{D}_2 \ar@{>->}[ur]_{\bar{f}_2}&&& D_0 \ar[ul]^{g_0}&&\mathring{D}_1 \ar@{>->}[ur]_{\mathring{f}_1}}\]
	
We have to build the dotted arrows in the diagram below
	
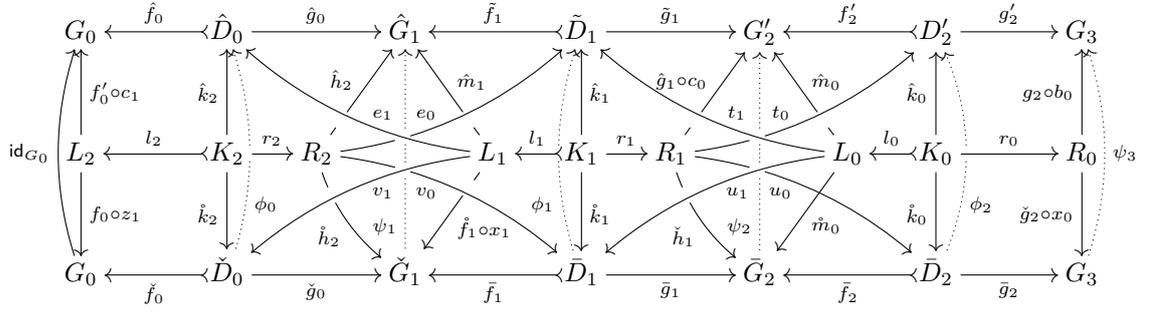
\[\xymatrix@C=15pt@R=30pt{G_0 && \ar@{>->}[ll]_{\hat{f}_0} \hat{D}_0 \ar[rr]^{\hat{g}_0}&& \hat{G}_1  && \ar@{>->}[ll]_{\tilde{f}_1} \tilde{D}_1 \ar[rr]^{\tilde{g}_1}&& G'_2 && \ar@{>->}[ll]_{f'_2} D'_2 \ar[rr]^{g'_2}&& G_3 \\
		L_2\ar[d]^{f_0\circ z_1} \ar[u]_{f'_0\circ c_1}&& K_2 \ar[u]^{\hat{k}_2}\ar[d]_{\mathring{k}_2}\ar@{>->}[ll]_{l_2} \ar[r]^{r_2} & R_2 \ar@/_.35cm/[urrr]^(.4){e_0}|(.29)\hole \ar[ur]|(.28)\hole^{\hat{h}_2} \ar@/^.35cm/[drrr]_(.4){v_0}|(.29)\hole \ar@/_.2cm/[dr]|(.28)\hole_{\mathring{h}_2} && L_1 \ar@/^.35cm/[ulll]_(.4){e_1} \ar[ul]|(.28)\hole_{\hat{m}_1} \ar@/_.35cm/[dlll]^(.4){v_1} \ar[dl]|(.28)\hole^{\mathring{f}_1 \circ x_1}& K_1  \ar[u]_{\hat{k}_1}\ar[d]^{\mathring{k}_1}\ar@{>->}[l]_{l_1} \ar[r]^{r_1} & R_1 \ar@/^.35cm/[drrr]_(.4){u_0}|(.29)\hole\ar@/_.2cm/[dr]_{\check{h}_1}|(.28)\hole \ar@/_.35cm/[urrr]^(.4){t_0}|(.29)\hole \ar[ur]|(.28)\hole^{\hat{g}_1\circ c_0}  && L_0\ar[ul]|(.28)\hole_{\hat{m}_0} \ar@/^.35cm/[ulll]_(.4){t_1} \ar@/_.35cm/[dlll]^(.4){u_1} \ar[dl]^{\mathring{m}_0}& K_0 \ar[d]_{\mathring{k}_0}\ar[u]^{\hat{k}_0}\ar@{>->}[l]_{l_0} \ar[rr]^{r_0} && R_0 \ar[u]^{g_2\circ b_0}\ar[d]_{\check{g}_2\circ x_0} \\G_0 \ar@/^.3cm/[uu]^{\id{G_0}}&& \ar@{>->}[ll]^{\check{f}_0} \check{D}_0 \ar@{.>}@/_.3cm/[uu]_(.3){\phi_{0}}|\hole  \ar[rr]_{\check{g}_0}&& \check{G}_1  \ar@{.>}@/_.0cm/[uu]^(.2){\psi_{1}}|(.43)\hole |(.57)\hole  && \ar@{>->}[ll]^{\bar{f}_1} \bar{D}_1 \ar@{.>}@/^.3cm/[uu]^(.3){\phi_{1}}|\hole \ar[rr]_{\bar{g}_1}&& \bar{G}_2 \ar@{.>}@/_.0cm/[uu]^(.2){\psi_{2}}|(.43)\hole |(.57)\hole && \ar@{>->}[ll]^{\bar{f}_2} \ar@{.>}@/_.4cm/[uu]_(.3){\phi_{2}}|\hole \bar{D}_2 \ar[rr]_{\bar{g}_2}&& G_3 \ar@{.>}@/_.3cm/[uu]_{\psi_{3}} }\]
	
	Now, we have already proved that $f_0\circ z_1=f'_0\circ c_1$, thus by \Cref{prop:unique} we have isomorphisms $\phi_{0}\colon \check{D}_0\to \hat{D}_0$ and $\psi_{1}\colon \check{G}_1\to \hat{G}_1$.
	Notice that
	\[
	\hat{f}_0\circ \phi_{0}\circ v_1  = \check{f}_0\circ v_1 =\check{f}_0\circ \check{q}_3\circ v'_1=f_0\circ \check{q}_1\circ v'_1=f_0\circ i_1=\hat{f}_0\circ e_1\]
Thus, using the fact that $\hat{f}_0$ is mono, we can conclude that $\phi_0\circ v_1=e_1$. Thus we have
\[\psi_1\circ \mathring{f}\circ x_1=\psi_1\circ \check{g}_0\circ v_1= \hat{g}_0\circ \phi_0\circ v_1=\hat{g}_0\circ e_1=\hat{m}_1 \]
	 \Cref{prop:unique} now yields the $\phi_1\colon \bar{D}_1\to\tilde{D}_1$ and $\psi_2\colon \bar{G}_2\to G'_2$.  We can compute to get
	 \[
	 \tilde{f}_1\circ \phi_1\circ v_0=\psi_1\circ \bar{f}_1\circ v_0=\psi_1\circ \mathring{h}_2=\hat{h}_2=\tilde{f}_1\circ e_0\]
	 Hence $\phi_1\circ v_0=e_0$. On the other hand, notice that
	 \[\hat{f}_0\circ s_4\circ t'_1=f'_0\circ s_1\circ t'_1=f'_0\circ a_1=m_0=\check{f}\circ y_1=\hat{f}_0\circ \phi_0\circ y_1\]
	 Yielding $s_4\circ t'_1=\phi_0\circ y_1$. But then
	 \begin{align*}
&\tilde{f}_1\circ \phi_1\circ u_1=\psi_1\circ \bar{f}_1\circ u_1=\psi_1\circ \check{m}_0= \psi_1\circ \check{g}_0\circ y_1\\=&\hat{g}_0\circ \phi_0\circ y_1=\hat{g}_0\circ s_4\circ t'_1=\tilde{f}_1\circ s_3\circ t'_1=\tilde{f}_1\circ t_1 
	 \end{align*}
	 
	 In turn, the previous identity entails that
	\[\psi_2\circ \mathring{m}_0=\psi_2\circ \bar{g}_1\circ u_1 = \tilde{g}_1\circ \phi_1\circ u_1 = \tilde{g}_1\circ t_1=\hat{m}_0\]
	
	We can use \Cref{prop:unique} a third time to get $\phi_2\colon \bar{D}_2\to D'_2$ and $\psi_3\colon G_3\to G_3$.  Now, as a last computation, we have
	\[f'_2\circ t_0= \hat{g}_2\circ c_0=\psi_2\circ h_1=\psi_2\circ \bar{f}_2\circ u_0=f'_2\circ \phi_2\circ u_0\]
	Therefore $\phi_2\circ u_0=t_0$ and this concludes the proof.
\end{proof}

We next prove some technical lemma that will be useful to show the
possibility of transforming a derivation into a switch equivalent one
by switching steps only along inversions.

\begin{lemma}
  \label{le:switchA}
  Let $\der{D}_i$ for $i \in \interval[0]{n}$ be derivation sequences
  and consider a switching sequence consisting only of inversions
  \begin{center}
    $\der{D}_0 \shift{\nu_1} \der{D}_1 \shift{\nu_2} \ldots
    \shift{\nu_n} \der{D}_n$
  \end{center}
  Let $i \in \interval{n}$ be such that
  $\transp{i} \in \inv{\nu_{1,n}}$ and for all $k > i+1$ it holds
  $\nu_{1,n}(k) > \nu_{1,n}(i+1)$.  Then there is a switching sequence
  consisting only of inversions that starts with $\transp{i}$, i.e.
  \begin{center}
    $\der{D}_0 \shift{\transp{i}} \der{D}_1' \shift{\nu_2'} \ldots
    \shift{\nu_n'} \der{D}_n$
  \end{center}
\end{lemma}

\begin{proof}
  We proceed by induction on $n$. If $n=1$, necessarily
  $\nu_1 =  \transp{i}$ and we are done.

  If $n>1$, let $\nu_1 = \transp{j}$.  We need three cases,
  depending on how $\nu_1$ relates to $\transp{i}$.

  \bigskip
  \noindent
  (1) $j=i$\\
  In this case there is nothing to prove: $\nu_1 = \transp{i}$, hence
  the original switching sequence is already of the desired shape.

  \bigskip
  \noindent
  (2) $j<i-1$ or $j>i-1$\\
  In this case $\nu_1$ does not ``interfere'' with $\transp{i}$, i.e.~$i$ still satisfies the hypotheses of the lemma with respect to the
  switching sequence
  $\der{D}_1 \shift{\nu_2} \ldots \shift{\nu_n} \der{D}_n$.

  Hence we apply the inductive hypothesis and derive the existence
  of a switching sequence
  \begin{center}
    $\der{D}_1 \shift{\transp{i}} \der{D}_2' \shift{\nu_3'} \ldots \shift{\nu_n'} \der{D}_n$
  \end{center}

  Recall that $\der{D}_0 \shift{\nu_1} \der{D}_1$ and, given the
  working assumption (ii), the switches $\nu_1$ and $\transp{i}$ can
  be done in reverse order, thus getting, as desired
  \begin{center}
    $\der{D}_0 \shift{\transp{i}} \der{D}_1' \shift{\nu_1} \ldots
    \shift{\nu_n'} \der{D}_n$
  \end{center}

  \bigskip
  \noindent
  (3) $j=i-1$, i.e., $\nu_1 = \transp{i-1}[i]$\\
  In this case $\der{D}_1$ is obtained from $\der{D}_0$ by switching
  steps $i-1$ and $i$, as pictorially reported in the first two
  columns of Fig.~\ref{fi:switchA}. Since we are dealing with
  switching sequence consisting only of inversions and by hypothesis
  $\nu_{1,n}(i) > \nu_{1,n}(i+1)$, whence
  $\nu_{2,n}(i-1) > \nu_{2,n}(i+1)$, we deduce that necessarily
  $\nu_{2,n}(i) > \nu_{2,n}(i+1)$ and $\nu_{2,n}(k) > \nu_{2,n}(i+1)$
  for all $k > i+1$. Hence we can use the inductive hypothesis to
  deduce that the switching sequence from $\der{D}_1$ can start by
  switching $\transp{i}$, i.e.~we obtain (see the third column in
  Fig.~\ref{fi:switchA})
  \begin{center}
    $\der{D}_1 \shift{\transp{i}} \der{D}_2' \shift{\nu_3'} \ldots
    \shift{\nu_n'} \der{D}_n$
  \end{center}

  Now, if we focus on $\der{D}_2'$ we observe that
  $\nu_{3,n}'(i-1) > \nu_{3,n}'(i)$ and
  $\nu_{3,n}'(k) > \nu_{3,n}'(i)$ for all $k > i$. Hence we can apply
  again the inductive hypothesis and get a switching sequence
    \begin{center}
    $\der{D}_2' \shift{\transp{i-1}[i]} \der{D}_3'' \shift{\nu_3''} \ldots
    \shift{\nu_n''} \der{D}_n$
  \end{center}
  as depicted in Fig.~~\ref{fi:switchA}, starting from the fourth column.

  By Lemma~\ref{lem:indep-global-left}(\ref{lem:indep-global-left:1})
  we can reorganise the first three switchings as

  \begin{center}
    $\der{D}_0 \shift{\transp{i}} \der{D}_1'''
    \shift{\transp{i-1}[i]} \der{D}_2'''
    \shift{\transp{i}} \der{D}_3''$
  \end{center}

  The fact that the three switchings above lead to $\der{D}_3''$ is a
  consequence of Lemma~\ref{lem:switch-confluence}.

  \bigskip

  Note that it cannot be $j=i+1$, otherwise we would have
  $\nu_{1,n}(j+1) < \nu_{1,n}(i)$, contradicting the hypotheses. Thus
  this case can be ignored and the proof is completed.  
\end{proof}

\begin{figure}
  \centering
  \begin{tikzpicture}[font=\small]
    \def\colorA{green!60!black}
    \def\colorB{blue}
    \def\colorC{red!100}

    \def\vlen{5mm}
    \def\hlen{18mm}
    \def\gap{1mm}

    \node[above] at (0*\hlen,0) {$\der{D}_0$};
    \node[above] at (0.5*\hlen,0) {$\shift{\nu_1}$};

    \draw (0*\hlen,0) -- (0*\hlen,-2*\vlen+\gap);

    \draw[->, \colorA] (0*\hlen,-2*\vlen) -- (0*\hlen,-3*\vlen+\gap) node[midway, left] {$i-1$};
    \draw[->, \colorB] (0*\hlen,-3*\vlen) -- (0*\hlen,-4*\vlen+\gap) node[midway, left] {$i$};
    \draw[->, \colorC] (0*\hlen,-4*\vlen) -- (0*\hlen,-5*\vlen+\gap) node[midway, left] {$i+1$};
    \draw (0*\hlen,-5*\vlen) -- (0*\hlen,-7*\vlen);
    
    \node[above] at (1*\hlen,0) {$\der{D}_1$};
    \node[above] at (1.5*\hlen,0) {$\shift{\transp{i}}$};

    \draw (1*\hlen,0) -- (1*\hlen,-2*\vlen+\gap);

    \draw[->, \colorB] (1*\hlen,-2*\vlen) -- (1*\hlen,-3*\vlen+\gap) node[midway, left] {};
    \draw[->, \colorA] (1*\hlen,-3*\vlen) -- (1*\hlen,-4*\vlen+\gap) node[midway, left] {};
    \draw[->, \colorC] (1*\hlen,-4*\vlen) -- (1*\hlen,-5*\vlen+\gap) node[midway, left] {};
    \draw (1*\hlen,-5*\vlen) -- (1*\hlen,-7*\vlen);

    
    \node[above] at (2*\hlen,0) {$\der{D}_2'$};
    \node[above] at (2.5*\hlen,0) {$\shift{\transp{i-1}[i]}$};
 
    \draw (2*\hlen,0) -- (2*\hlen,-2*\vlen+\gap);
    
    \draw[->, \colorB] (2*\hlen,-2*\vlen) -- (2*\hlen,-3*\vlen+\gap) node[midway, left] {};
    \draw[->, \colorC] (2*\hlen,-3*\vlen) -- (2*\hlen,-4*\vlen+\gap) node[midway, left] {};
    \draw[->, \colorA] (2*\hlen,-4*\vlen) -- (2*\hlen,-5*\vlen+\gap) node[midway, left] {};
    \draw (2*\hlen,-5*\vlen) -- (2*\hlen,-7*\vlen);

    \node[above] at (3*\hlen,0) {$\der{D}_3''$};
    \node[above] at (3.5*\hlen,0) {$\shift{\nu_3''}$};

    \draw (3*\hlen,0) -- (3*\hlen,-2*\vlen+\gap);
    \draw[->, \colorC] (3*\hlen,-2*\vlen) -- (3*\hlen,-3*\vlen+\gap) node[midway, left] {};
    \draw[->, \colorB] (3*\hlen,-3*\vlen) -- (3*\hlen,-4*\vlen+\gap) node[midway, left] {};
    \draw[->, \colorA] (3*\hlen,-4*\vlen) -- (3*\hlen,-5*\vlen+\gap) node[midway, left] {};
    \draw (3*\hlen,-5*\vlen) -- (3*\hlen,-7*\vlen);
  
    \draw (4*\hlen,0*\vlen) node[above] {$\dots$};

    \node[above] at (5*\hlen,0) {$\der{D}_n$};
    \node[above] at (4.5*\hlen,0) {$\shift{\nu_{n-1}}$};

  \end{tikzpicture}
  \caption{Picture for the proof of Lemma~\ref{le:switchA}}
  \label{fi:switchA}
\end{figure}
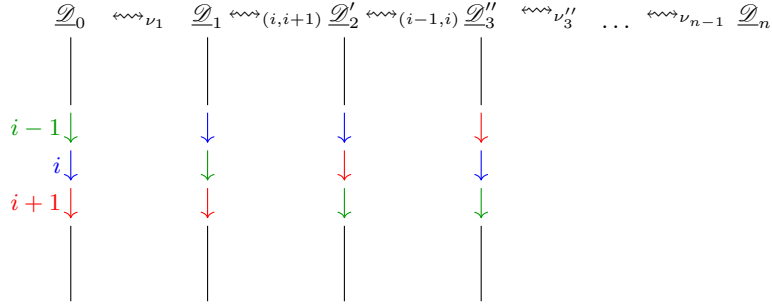

\begin{remark}
  \label{rem:switchA}
  Observe that given a switching sequence as in the lemma above
  $\der{D}_0 \shift{\nu_1} \der{D}_1 \shift{\nu_2} \ldots
  \shift{\nu_n} \der{D}_n$, if we consider the inversion with largest index
  \begin{center}
    $k = \max \{j \mid \transp{j} \in \inv{\nu_{1,n}}\}$
  \end{center}
  then it satisfies the hypotheses of the lemma.
\end{remark}

\begin{lemma}
  \label{le:switchB}
  Let $\der{D}_i$ for $i \in \interval[0]{n}$ be derivation sequences
  and consider a switching sequence consisting only of inversions
  \begin{center}
    $\der{D}_0 \shift{\nu_1} \der{D}_1 \shift{\nu_2} \ldots
    \shift{\nu_n} \der{D}_n$
  \end{center}
  Let $i \in \interval{n}$ be such that
  $\transp{i} \in \inv{\nu_{1,n}}$ and assume that steps $i$ and $i+1$ in $\der{D}_0$ are sequentially independent. Then there is a switching sequence
  consisting only of inversions that starts with $\transp{i}$, i.e.,
  \begin{center}
    $\der{D}_0 \shift{\transp{i}} \der{D}_1' \shift{\nu_2'} \ldots
    \shift{\nu_n'} \der{D}_n$
  \end{center}
\end{lemma}

\begin{proof}
  We proceed by induction on $n$. If $n=1$, necessarily
  $\nu_1 =  \transp{i}$ and we are done.

  If $n>1$, take $k = \max \{j \mid \transp{j} \in \inv{\nu_{1,n}}\}$, which,
  as observed in Remark~\ref{rem:switchA}, satisfies the hypotheses of
  Lemma~\ref{le:switchA}. Hence, by the mentioned lemma, we can start
  the switching sequence with the transposition $\transp{k}$, thus
  obtaining
  \begin{center}
    $\der{D}_0 \shift{\transp{k}} \der{D}_1' \shift{\nu_2'} \ldots
    \shift{\nu_n'} \der{D}_n$
  \end{center}

  Moreover, since $\transp{i}$ is an inversion, we have that
  $\i \leq k$.  We distinguish various case depending on how $k$
  relates to $i$.

  \bigskip
  \noindent
  (i) $j=i$\\
  In this case we are done: $\nu_1' = \transp{i}$, hence we obtained a
  switching sequence of the desired shape.

  \bigskip
  \noindent
  (ii) $k > i+1$\\
  In this case the switching $\transp{k}$ does not ``interfere'' with
  $\transp{i}$, as depicted in Fig.~\ref{fi:switchB1} (first two
  columns).

  Hence if we focus on
  \begin{center}
    $\der{D}_1' \shift{\nu_2'} \ldots
    \shift{\nu_n'} \der{D}_n$
  \end{center}
  we can observe that
  $\nu_{2,n}'(i)=\nu_{1,n}(i) > \nu_{1,n}(i+1)=\nu_{2,n}'(i+1)$.
  Thus we can use the inductive hypothesis and obtain (see
  Fig.~\ref{fi:switchB1}, third column)
  \begin{center}
    $\der{D}_1' \shift{\transp{i}} \der{D}_2'' \shift{\nu_3''} \ldots
    \shift{\nu_n''} \der{D}_n$
  \end{center}

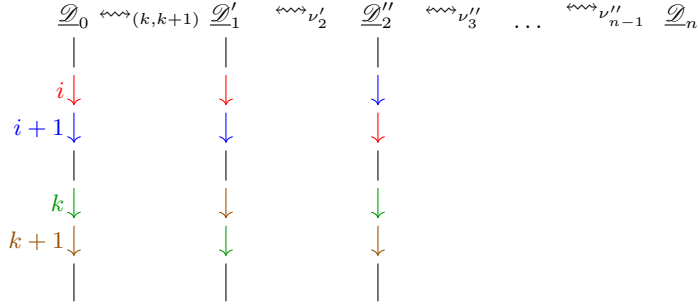
\begin{figure}
  \centering
  \begin{tikzpicture}[font=\small]
    \def\colorA{red!100}
    \def\colorB{blue}
    \def\colorC{green!60!black}
    \def\colorD{orange!60!black}

    \def\vlen{5mm}
    \def\hlen{20mm}
    \def\gap{1mm}

    \node[above] at (0*\hlen,0) {$\der{D}_0$};
    \node[above] at (0.5*\hlen,0) {$\shift{\transp{k}}$};

    \draw (0*\hlen,0) -- (0*\hlen,-1*\vlen+\gap);

    \draw[->, \colorA] (0*\hlen,-1*\vlen) -- (0*\hlen,-2*\vlen+\gap) node[midway, left] {$i$};
    \draw[->, \colorB] (0*\hlen,-2*\vlen) -- (0*\hlen,-3*\vlen+\gap) node[midway, left] {$i+1$};
    \draw (0*\hlen,-3*\vlen) -- (0*\hlen,-4*\vlen+\gap);

    \draw[->, \colorC] (0*\hlen,-4*\vlen) -- (0*\hlen,-5*\vlen+\gap) node[midway, left] {$k$};
    \draw[->, \colorD] (0*\hlen,-5*\vlen) -- (0*\hlen,-6*\vlen+\gap) node[midway, left] {$k+1$};
    \draw (0*\hlen,-6*\vlen) -- (0*\hlen,-7*\vlen);

    \node[above] at (1*\hlen,0) {$\der{D}_1'$};
    \node[above] at (1.5*\hlen,0) {$\shift{\nu_2'}$};

    \draw (1*\hlen,0) -- (1*\hlen,-1*\vlen+\gap);

    \draw[->, \colorA] (1*\hlen,-1*\vlen) -- (1*\hlen,-2*\vlen+\gap) node[midway, left] {};
    \draw[->, \colorB] (1*\hlen,-2*\vlen) -- (1*\hlen,-3*\vlen+\gap) node[midway, left] {};
    \draw (1*\hlen,-3*\vlen) -- (1*\hlen,-4*\vlen+\gap);

    \draw[->, \colorD] (1*\hlen,-4*\vlen) -- (1*\hlen,-5*\vlen+\gap) node[midway, left] {};
    \draw[->, \colorC] (1*\hlen,-5*\vlen) -- (1*\hlen,-6*\vlen+\gap) node[midway, left] {};
    \draw (1*\hlen,-6*\vlen) -- (1*\hlen,-7*\vlen);

    \node[above] at (2*\hlen,0) {$\der{D}_2''$};
    \node[above] at (2.5*\hlen,0) {$\shift{\nu_3''}$};

    \draw (2*\hlen,0) -- (2*\hlen,-1*\vlen+\gap);

    \draw[->, \colorB] (2*\hlen,-1*\vlen) -- (2*\hlen,-2*\vlen+\gap) node[midway, left] {};
    \draw[->, \colorA] (2*\hlen,-2*\vlen) -- (2*\hlen,-3*\vlen+\gap) node[midway, left] {};
    \draw (2*\hlen,-3*\vlen) -- (2*\hlen,-4*\vlen+\gap);

    \draw[->, \colorC] (2*\hlen,-4*\vlen) -- (2*\hlen,-5*\vlen+\gap) node[midway, left] {};
    \draw[->, \colorD] (2*\hlen,-5*\vlen) -- (2*\hlen,-6*\vlen+\gap) node[midway, left] {};
    \draw (2*\hlen,-6*\vlen) -- (2*\hlen,-7*\vlen);
    
    \draw (3*\hlen,0*\vlen) node[above] {$\dots$};

    \node[above] at (4*\hlen,0) {$\der{D}_n$};
    \node[above] at (3.5*\hlen,0) {$\shift{\nu_{n-1}''}$};

  \end{tikzpicture}
  \caption{Picture for the proof of Lemma~\ref{le:switchB}, item (ii)}
  \label{fi:switchB1}
\end{figure}

  Now, since $k>i+1$, the switchings $\transp{i}$ and $\transp{k}$ can
  be performed in reverse order and thus we obtain
  \begin{center}
    $\der{D}_0 \shift{\transp{i}} \der{D}_1'' \shift{\transp{k}}
    \der{D}_2'' \shift{\nu_3''} \ldots \shift{\nu_n''} \der{D}_n$
  \end{center}
  as desired.

  \bigskip
  \noindent
  (iii) $k = i+1$\\
  The situation is depicted in Fig.~\ref{fi:switchB2}. In
  $\der{D}_1'$, by the choice of $k$ and the fact that all $\nu_i$ are
  inversions, one can deduce that $i$ satisfies the hypotheses of
  Lemma~\ref{le:switchA}. Hence we obtain (see Fig.~\ref{fi:switchB2}, first two columns)
  \begin{center}
    $\der{D}_1' \shift{\transp{i}} \der{D}_2'' \shift{\nu_2''}
    \der{D}_3'' \shift{\nu_4''} \ldots \shift{\nu_n''} \der{D}_n$
  \end{center}

  Since in $\der{D}_0$ steps in positions $i$ and $i+1$ are sequentially
  independent, we can use
  Lemma~\ref{lem:indep-global-left}(\ref{lem:indep-global-left:1}) and
  we have that the corresponding steps in positions $i+1$ and $i+2$ in
  $\der{D}_2''$ are still sequentially independent. Hence, we can use
  the inductive hypothesis on $\der{D}_2''$ and deduce the existence
  of a switching sequence where $\transp{i+1}[i+2]$ can be switched
  first, i.e.~we obtain (see Fig.~\ref{fi:switchB2}, third column)
  \begin{center}
    $\der{D}_2'' \shift{\transp{i+1}[i+2]}
    \der{D}_3''' \shift{\nu_4'''} \ldots \shift{\nu_n'''} \der{D}_n$
  \end{center}
  
  Putting things together, we obtain
  \begin{center}
    $\der{D}_0 \shift{\transp{i+1}[i+2]} \der{D}_1' \shift{\transp{i}}
    \der{D}_2'' \shift{\transp{i+1}[i+2]} \der{D}_3''' \shift{\nu_4'''}\ldots \shift{\nu_n'''} \der{D}_n$
  \end{center}
  Since we additionally know that in $\der{D}_0$ the steps $i$ and
  $i+1$ are sequentially independent, we can use
  Lemma~\ref{lem:indep-global-left}(\ref{lem:indep-global-left:2}) to
  reorder the first three switchings, obtaining
  \begin{center}
    $\der{D}_0 \shift{\transp{i}} \der{D}_1' \shift{\transp{i+1}[i+2]}
    \der{D}_2'' \shift{\transp{i}} \der{D}_3''' \shift{\nu_4'''}\ldots \shift{\nu_n'''} \der{D}_n$
  \end{center}
  as desired.  The fact that the three switchings above lead to
  $\der{D}_3'''$ is a consequence of Lemma~\ref{lem:switch-confluence}.
\end{proof}

\begin{figure}
  \centering
  \begin{tikzpicture}[font=\small]
    \def\colorA{red!100}
    \def\colorB{blue}
    \def\colorC{green!60!black}

    \def\vlen{5mm}
    \def\hlen{22mm}
    \def\gap{1mm}

    \node[above] at (0*\hlen,0) {$\der{D}_0$};
    \node[above] at (0.5*\hlen,0) {$\shift{\transp{i+1}[i+2]}$};

    \draw (0*\hlen,0) -- (0*\hlen,-2*\vlen+\gap);

    \draw[->, \colorA] (0*\hlen,-2*\vlen) -- (0*\hlen,-3*\vlen+\gap) node[midway, left] {$i$};
    \draw[->, \colorB] (0*\hlen,-3*\vlen) -- (0*\hlen,-4*\vlen+\gap) node[midway, left] {$k=i+1$};
    \draw[->, \colorC] (0*\hlen,-4*\vlen) -- (0*\hlen,-5*\vlen+\gap) node[midway, left] {$i+2$};
    \draw (0*\hlen,-5*\vlen) -- (0*\hlen,-7*\vlen);
    
    \node[above] at (1*\hlen,0) {$\der{D}_1'$};
    \node[above] at (1.5*\hlen,0) {$\shift{\transp{i}}$};

    \draw (1*\hlen,0) -- (1*\hlen,-2*\vlen+\gap);

    \draw[->, \colorA] (1*\hlen,-2*\vlen) -- (1*\hlen,-3*\vlen+\gap) node[midway, left] {};
    \draw[->, \colorC] (1*\hlen,-3*\vlen) -- (1*\hlen,-4*\vlen+\gap) node[midway, left] {};
    \draw[->, \colorB] (1*\hlen,-4*\vlen) -- (1*\hlen,-5*\vlen+\gap) node[midway, left] {};
    \draw (1*\hlen,-5*\vlen) -- (1*\hlen,-7*\vlen);

    
    \node[above] at (2*\hlen,0) {$\der{D}_2''$};
    \node[above] at (2.5*\hlen,0) {$\shift{\transp{i+1}[i+2]}$};
 
    \draw (2*\hlen,0) -- (2*\hlen,-2*\vlen+\gap);
    
    \draw[->, \colorC] (2*\hlen,-2*\vlen) -- (2*\hlen,-3*\vlen+\gap) node[midway, left] {};
    \draw[->, \colorA] (2*\hlen,-3*\vlen) -- (2*\hlen,-4*\vlen+\gap) node[midway, left] {};
    \draw[->, \colorB] (2*\hlen,-4*\vlen) -- (2*\hlen,-5*\vlen+\gap) node[midway, left] {};
    \draw (2*\hlen,-5*\vlen) -- (2*\hlen,-7*\vlen);

    \node[above] at (3*\hlen,0) {$\der{D}_3'''$};
    \node[above] at (3.5*\hlen,0) {$\shift{\nu_3'''}$};

    \draw (3*\hlen,0) -- (3*\hlen,-2*\vlen+\gap);
    \draw[->, \colorC] (3*\hlen,-2*\vlen) -- (3*\hlen,-3*\vlen+\gap) node[midway, left] {};
    \draw[->, \colorB] (3*\hlen,-3*\vlen) -- (3*\hlen,-4*\vlen+\gap) node[midway, left] {};
    \draw[->, \colorA] (3*\hlen,-4*\vlen) -- (3*\hlen,-5*\vlen+\gap) node[midway, left] {};
    \draw (3*\hlen,-5*\vlen) -- (3*\hlen,-7*\vlen);

    \draw (4*\hlen,0*\vlen) node[above] {$\dots$};

    \node[above] at (5*\hlen,0) {$\der{D}_n$};
    \node[above] at (4.5*\hlen,0) {$\shift{\nu_{n-1}'''}$};

  \end{tikzpicture}
  \caption{Picture for the proof of Lemma~\ref{le:switchB} (item (iii))}
  \label{fi:switchB2}
\end{figure}
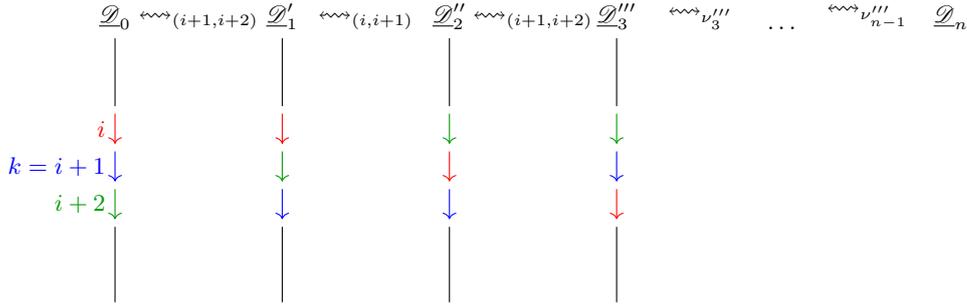

The next lemma intuitively states that if we have a switching sequence
where the first switch is applied to $\transp{i}$ and the last one
reverses the same switch, then, assuming that in the middle there are
only inversions, we can strip the first and last switches.

\begin{lemma}
  \label{le:switchC}
  Let $\der{D}_i$ for $i \in \interval[0]{n+1}$ be derivation sequences
  and consider a switching sequence
  \begin{center}
    $\der{D}_0 \shift{\nu_0} \der{D}_1 \shift{\nu_1} \ldots
     \der{D}_n \shift{\nu_n} \der{D}_{n+1}$
  \end{center}
  If the subsequence starting from $\der{D}_1$ consists of
  inversions and $\nu_0=\transp{i}$ while
  $\nu_n = \transp{\nu_{1,n}(i)}$, then there is shorter switching sequence
  consisting only of inversions
  \begin{center}
    $\der{D}_0 \equiv_a \der{D}_0' \shift{\nu_2'} \der{D}_3' \shift{\nu_3'} \ldots
    \shift{\nu_{n-1}'} \der{D}_n'$
  \end{center}
\end{lemma}

\begin{proof}
  Just observe that in $\der{D}_1$ the steps of position $i$ and $i+1$
  are sequentially independent and $\nu_{1,n}(i) >
  \nu_{1,n}(i+1)$. Hence we can use Lemma~\ref{le:switchB} and obtain
  a switching sequence from $\der{D}_1$ where $\transp{i}$ is
  performed first, hence we obtain
  \begin{center}
    $\der{D}_0 \shift{\transp{i}}
    \der{D}_1 \shift{\transp{i}}
    \der{D}_2' \shift{\nu_2'} \ldots
    \shift{\nu_n'} \der{D}_{n+1}$
  \end{center}

  Now, switching $\transp{i}$ twice one obtains a derivation
  which is abstraction equivalent to the starting one.  Hence we have
  \begin{center}
    $\der{D}_0 \equiv_a
    \der{D}_2' \shift{\nu_2'} \ldots
    \shift{\nu_n'} \der{D}_{n+1}$
  \end{center}
  as desired.
\end{proof}

\begin{lemma}
  \label{le:switchD}
  Let $\der{D}_i$ for $i \in \interval[0]{n+1}$ be derivation sequences
  and consider a switching sequence
  \begin{center}
    $\der{D}_0 \shift{\nu_1} \der{D}_1 \shift{\nu_2} \ldots
    \shift{\nu_n}\der{D}_n$
  \end{center}
  If the switching sequence does not consists only of inversions,
  there is a shorter switching sequence relating
  $\der{D}_0' \equiv_a \der{D}_0$ and $\der{D}_n$.
\end{lemma}

\begin{proof}
  Take a maximal subsequence
  $\der{D}_i \shift{\nu_{i+1}} \der{D}_{i+1} \shift{\nu_{i+2}} \ldots
  \shift{\nu_j}\der{D}_j$ consisting only of inversions.
  If $i=0$ and $j=n$, there is nothing to show. Hence assume that
  $i>0$ and consider the sequence
  \begin{center}
    $\der{D}_{i-1} \shift{\nu_{i}} \der{D}_{i} \shift{\nu_{i+1}} \ldots
    \shift{\nu_j}\der{D}_j$
  \end{center}
  Then, by maximality, if $\nu_{i-1}=\transp{k}$ it must hold that
  $\nu_{j}=\transp{\nu_{i,j}(i)}$. Then, by Lemma~\ref{le:switchC},
  the switching sequence from $\der{D}_{i-1}$ to $\der{D}_{j}$ can be
  shortened by $2$, i.e.
  \begin{center}
    $\der{D}_{i-1} \equiv_a \der{D}_{i+1}' \shift{\nu_{i+1}'} \ldots
    \shift{\nu_j'}\der{D}_j$
  \end{center}
  and thus the same applies to the original switching sequence.
\end{proof}

\thNoNeed*
\label{thNoNeed-proof}

\begin{proof}
  Let $\der{D} \shifteq \der{D}'$. Then there is a switching
  sequence
  \begin{center}
    $\der{D} \equiv_a \der{D}_0 \shift{\nu_1} \der{D}_1 \shift{\nu_2} \ldots
    \shift{\nu_n} \der{D}_n \equiv_a \der{D}'$
  \end{center}
  If it does not consist only of inversions, by Lemma~\ref{le:switchD}
  we can reduce its length by $2$. An inductive reasoning allows us to
  conclude.
\end{proof}

\coCanonical*
\label{coCanonical-proof}

\begin{proof}
  Immediate consequence of Theorem~\ref{th:no-need} and Lemma~\ref{le:switchA}.
\end{proof}

\end{document}